\documentclass[11pt]{article}

\usepackage[margin=1in]{geometry}
\usepackage{amsmath,amssymb,amsthm,mathtools}
\usepackage{bm}
\usepackage{booktabs}
\usepackage{graphicx}
\usepackage{natbib}
\usepackage{hyperref}
\usepackage{enumitem}
\usepackage[titletoc,title]{appendix}
\usepackage{authblk}
\usepackage{tikz}
\usepackage{mathrsfs}

\hypersetup{colorlinks=true, linkcolor=blue, citecolor=blue, urlcolor=blue}

\newtheorem{proposition}{Proposition}
\newtheorem{lemma}[proposition]{Lemma}

\newtheorem{corollary}[proposition]{Corollary}

\newtheorem{auxlemma}{Lemma}[section]

\theoremstyle{definition}
\newtheorem{assumption}{Assumption}

\theoremstyle{remark}

\newcommand{\R}{\mathbb{R}}
\newcommand{\E}{\mathbb{E}}
\newcommand{\N}{\mathcal{N}}

\newcommand{\Uunins}{\mathcal{U}}
\newcommand{\Iins}{\mathcal{I}}
\newcommand{\indic}{\mathbf{1}}
\DeclareMathOperator*{\argmax}{argmax}

\DeclareMathOperator{\Var}{Var}
\DeclareMathOperator{\Cov}{Cov}

\title{Bayesian Rational Search Engine User}
\author{Shichao Ma\thanks{E-mail: \texttt{shichao.ma@outlook.com}.}}
\affil{Independent Researcher}
\date{This Version: November 2026}

\begin{document}

\maketitle

\begin{abstract}
A user faces a list ordered by a noisy proxy for relevance and decides whether to pay to inspect another item or stop with the best she has found. Because she does not know the page's quality, each inspection both produces a candidate and updates her beliefs about the remaining items. We show that the optimal policy is a standout rule: the user stops when her best find is sufficiently good relative to what she expects from the page, with the required margin changing with depth. The browsing dynamics collapse to a one-dimensional Markov chain, yielding the full distribution of inspection depth through a closed-form recursion. The model uncovers three reasons for stopping: trust, commit, and cut-losses, and yields testable implications. It also turns inspection depth into a learning-to-rank likelihood. Observed depth restricts the latent relevance path to a polyhedron of survival inequalities, while conversion identifies the terminal winner. Simulations show that the estimator recovers relevance under correct specification and improves ranking using only the incumbent ranker's features. Applied to public Baidu search logs, the inspection depth likelihood produces rankings that align more closely with expert relevance judgments than click-trained alternatives.
\end{abstract}

\bigskip


\newpage

\section{Introduction}

Search engines, e-commerce platforms, and recommender systems all deliver their output as a ranked list. The system retrieves a pool of items and orders them by a learned score that is meant to track how relevant each item is to the user who made the request. That score, however, is unavoidably a proxy: the user's private notion of relevance is not observable to the system, and any signal the system can compute is informative but noisy. The top item is therefore not necessarily the most relevant, and whether to commit to it or keep scrolling is a real economic decision. Each additional look costs the user some time and attention, and what she will find is uncertain.

This paper studies that decision from first principles. A user faces a list whose order she takes to be informative but imperfect. She can inspect items one at a time at a fixed cost; each inspection reveals an item's true relevance. She stops whenever she likes and keeps the best item she has uncovered, or falls back to an outside option that is available to her off-list. Crucially, she does not enter the page knowing how good its items are. Instead, she holds a prior over the page's underlying relevance, and every inspection serves two purposes at once: it produces a candidate item, and it sharpens her view of the distribution from which the rest of the page is drawn.

The first channel raises the user's payoff directly. The second learning channel is what makes the problem richer than the classical sequential-search exercise. Even if an item is not as relevant as a previously inspected one, learning it still improves every continuation decision that follows. As the user goes deeper, her belief stabilizes, the items still ahead are expected to be worse on average than the ones she has already seen, and yet both channels have to be paid for by the same fixed cost. Optimal behavior must reconcile these forces.

Our first contribution is a clean characterization of that reconciliation. We show that, however the user reorganizes her search, she does no worse by linear traversal. This collapses what looks like a combinatorial problem onto a single stopping time. The optimal policy then takes a strikingly simple form, which we call a \emph{standout} rule. At every step the user keeps two summaries of her state in mind: the best item she has revealed so far, and her current best guess at the quality of the page. She continues as long as her best find does not yet stand out far enough above that quality guess, and stops as soon as the gap clears a depth-dependent threshold. 

The threshold itself has two interpretable pieces: an option-value term for the residual uncertainty about what the user might still find, and a rank-specific shift that absorbs how much a given rank is expected to outperform an average position. Both pieces shrink as the user moves down the page, which is why users stop somewhere along the list rather than only at the very end. A further result reduces the apparently two-dimensional dynamics to a one-dimensional Markov chain in the gap between the best find and the posterior quality, with a closed-form transition that recovers the full distribution of inspection depth through a simple recursion.

The economic content of the model lies in what the optimal policy says about the empirical signals that search and recommendation systems harvest from log data. We highlight three behavioral predictions and a separate methodological contribution for relevance learning.

First, one-click sessions, a dominant phenomenon in real-world search logs, are produced by qualitatively different mechanisms that the depth log alone cannot tell apart. In one mechanism, which we call the \emph{trust} regime, the user inspects the top item only to confirm what she already expected, and would have stopped at rank~$1$ for every realization of that item: the decision to stop was effectively made before the page rendered. In the other, the user is in principle prepared to look further, but the first item resolves the matter: either by being so good that further looking is wasteful (the \emph{commit} branch) or by being so disappointing that her outside option dominates the rest of the list in expectation (the \emph{cut-losses} branch). The same observable footprint admits a satisfied user, a disappointed user, and a user who never planned to engage with the rest of the page in the first place.

Second, the trust channel produces what we frame as a \emph{winner's curse} for ranking informativeness. A more informative ranker has a more decisive top slot, which widens the gap that triggers the trust regime, and it places objectively better items at the top, which makes more first-look outcomes good enough to commit on, reinforcing the same concentration. Such a ranker therefore concentrates inspections at rank~$1$ and starves the logger of the deeper observations needed to evaluate it, retrain it, or detect that some lower-ranked item should have been promoted. This is not biased logging or model misuse. It is an equilibrium consequence of rational stopping, and it tightens precisely as the ranking becomes more informative.

Third, the optimal continuation interval predicts a distinctive nonmonotone relationship between realized rank-$i$ relevance and the probability of continuing to rank $i+1$. With history held fixed, a mediocre item at rank $i$ keeps the user going more reliably than an excellent one does, because an excellent item raises her lead above the standout threshold and stops her. Standard cascade models with monotone satisfaction cannot match this shape: satisfaction with the inspected item only ever raises the probability of stopping, and the mediocre-driven continuation channel is absent by assumption. 

Taken together, these results clarify what common behavioral measures can and cannot reveal. Scroll depth, dwell time, and session length (the workhorse implicit-feedback signals in modern search and recommendation, \citep{fox2005evaluating, hassan2010beyond}) are not univariate signals of either user satisfaction or page quality. Generically, the same depth measurement is consistent with a satisfied user who has found what she wanted and a disappointed user who has confirmed the page is not for her. A downstream conversion event (a purchase, a long-dwell read, or a completed booking) supplies the auxiliary observation that identifies the terminal winner and separates an item choice from an outside-option choice. Platforms that natively observe conversions therefore face a much easier inference problem than those that rely on depth alone.

A separate methodological contribution turns inspection depth into a direct signal for learning relevance. Reaching depth $t$ reveals that the history of earlier items left the user willing to continue, so the observation restricts the entire latent relevance path rather than labeling items one at a time. The standout rule makes these joint restrictions explicit, while a conversion further reveals the terminal winner. Under a feature-based relevance model, we turn the probability of these restrictions into a differentiable learning-to-rank likelihood based directly on depth and conversion records.

We next evaluate what the new likelihood can learn from simulated and real search data. In simulation, the structural estimator recovers relevance and the parameters governing stopping under correct specification. When restricted to the features available to the incumbent ranker, it continues to outperform conventional learners regardless of whether terminal choice is observed. We then apply the inspection depth likelihood to public Baidu search logs \citep{zou2022large,hager2024unbiased}, where it produces rankings that align more closely with separate expert judgments than comparable click-trained models. Together, these results show that inspection depth can support relevance learning even when explicit relevance labels and reliable terminal choices are unavailable.

The remainder of the paper develops these results in order. The next two sections set up the environment, and derive the optimal policy and the standout rule. Section~\ref{sec:inspection_depth} reduces the dynamics to a one-dimensional lead chain, characterizes the full distribution of inspection depth, and discusses the winner's curse as a corollary of a conditional-continuation result. Section~\ref{sec:partial-id} constructs the survival polyhedron in latent-relevance space and the depth-and-conversion likelihood it generates. Section~\ref{sec:experiments} evaluates the resulting estimator in simulation and on Baidu search logs.

\subsection*{Related Literature}

This paper speaks first to the literature on position bias in information retrieval. A broad literature defines position bias as a tendency to interact with higher-ranked items regardless of their relevance and therefore treats clicks as biased signals of relevance \citep{craswell2008experimental,chapelle2009dynamic,dupret2008userbrowsing,chuklin2015click,chen2023bias}. Our model gives this rank dependence a different interpretation. Rank conveys information about relevance, and the user understands that information. How deeply to inspect is a deliberate choice rather than a nuisance. The position dependence in her behavior therefore arises as an equilibrium response to the ranked list, even when she is perfectly rational, rather than simply as a distortion layered onto fixed item-level preferences.

This interpretation distinguishes our framework from conventional unbiased learning-to-rank (ULTR) methods. Propensity-based methods reweight interactions to correct selective exposure \citep{joachims2017unbiased,wang2016learning,wang2018position}; policy-aware methods extend this logic to truncated rankings \citep{oosterhuis2020policy}; dual learning estimates relevance and examination jointly \citep{ai2018unbiased}; and later work addresses trust bias through affine correction, latent examination through doubly robust estimation, and sparse feedback through user simulation \citep{vardasbi2020affine,oosterhuis2023doubly,luo2023model}. These methods differ in how they estimate exposure and which biases they correct, but they separate the observation process from an underlying relevance signal and seek to remove the former's influence on the latter. In our model, observation is not a contaminating layer to be removed. It is an endogenous outcome of the same search problem that governs stopping and therefore carries information about latent relevance. 

The click-modeling tradition shares our goal of explaining how users browse ranked lists. The cascade model of \citet{craswell2008experimental}, the dynamic Bayesian network of \citet{chapelle2009dynamic}, and the user browsing model of \citet{dupret2008userbrowsing} encode examination and continuation through position, attraction, satisfaction, and previous clicks. These models motivate their behavioral assumptions with eye-tracking evidence and discipline them by their fit to observed clicks, but they do not derive the user's choice from a unified decision problem. In our model, by contrast, the user maximizes expected payoff under an inspection cost and Bayesian beliefs about the page, so linear traversal and the reservation rule follow from optimality rather than being imposed as assumptions. In cascade-style models, satisfaction increases with an item's attractiveness, so better items weakly increase the probability of stopping. Our reservation rule instead implies that continuation is nonmonotone in current relevance: search ends after either a sufficiently good or a sufficiently bad realization, while intermediate realizations sustain it.

Recent work further underscores the importance of behavioral structure in learning from user interactions. Whole-page ULTR identifies the choice of a suitable user-behavior hypothesis as a central challenge and uses causal discovery to accommodate multiple presentation features \citep{mao2024whole}. Other methods account for confounding by the logging policy, listwise context, and user heterogeneity \citep{luo2024unconfounded,yu2024contextual,niu2025personalized}. Work on two-tower models shows that flexible decompositions of clicks into relevance and examination need not identify either component without document swaps or sufficient feature overlap across ranks \citep{hager2025unidentified,fris2026diagnosing}. Our approach assigns behavioral structure a stronger role: the same assumptions that explain browsing also determine what observed behavior reveals about relevance.

This behavioral foundation yields a new training label: inspection depth. \citet{joachims2005accurately} show that clicks are difficult to interpret as absolute relevance labels and motivate reading them as relative preferences. We show that the decision to inspect rank~$t$ is itself a joint statement about ranks $1$ through $t-1$. Section~\ref{sec:partial-id} expresses these restrictions as a survival polyhedron and uses its probability as a likelihood for relevance. The training label is therefore not an item-level click or non-click, but the depth at which the user optimally stops. It remains informative even if the user does not click any item.

We take this training label to the public Baidu search logs \citep{zou2022large,hager2024unbiased}. These logs provide an especially informative setting because \citet{hager2024unbiased} find that standard propensity corrections do not consistently improve agreement with expert relevance judgments. Studies using the same data enrich examination models with additional presentation signals, incorporate listwise context and query-level propensities, or directly learn a list-level reward \citep{chen2023multifeature,yu2024contextual,yu2025query,bhatt2025rewardrank}. By holding the feature representation and behavioral sample fixed, our matched comparison isolates the difference between learning from inspection depth and learning from item-level clicks. Under this design, the inspection-depth likelihood produces rankings that agree more closely with expert judgments. The stopping point is therefore not merely another engagement statistic; it contains ranking information that item-by-item click labels leave unused.

The closest recent reduced-form counterpart is the query-level propensity introduced by \citet{yu2025query}. DualIPW recognizes that relevant results can remain unclicked after the user's need has been satisfied and augments position-level weighting with a query-level correction. With homogeneous users, this propensity can be read as a reduced-form approximation to a marginal probability in our model: the session-level probability that a list produces observed engagement before search ends.

Our results also have direct implications for the long-running effort to design search success metrics from implicit signals \citep{fox2005evaluating,hassan2010beyond}. Scroll depth, dwell time, and session length have all been proposed as proxies for satisfaction. The model clarifies their interpretation: depth or duration alone is not a univariate signal of either user welfare or search quality. The same depth can describe a satisfied user, who commits after finding a good item, or a disappointed user, who falls back to the outside option after learning that the page is weak. More observations of depth cannot distinguish these mechanisms by themselves. A conversion signal, such as a purchase, supplies the missing branch information. Platforms that observe such an outcome directly therefore face a more informative inference problem than those that must construct a suitable conversion proxy.

The winner's curse for ranking informativeness adds a distinct mechanism to the literature on feedback loops in ranking and recommendation \citep{schnabel2016treatments,chaney2018algorithmic,chen2023bias}. A more informative ranker gives users stronger reasons to trust its top positions and makes good early realizations more common. Both forces shorten search and remove lower-ranked items from the log. The system can therefore lose precisely the feedback it needs to improve further because its current ranking has become more informative. This result requires neither mistaken users nor an initially misspecified learner. It is an equilibrium consequence of successful information provision. Research on click-model robustness to ranking-policy changes and distributionally robust ULTR examines how learning performs when the feedback distribution shifts \citep{deffayet2023robustness,niu2025robust}. Our result identifies a structural source of that shift: the ranking policy changes what rational users choose to inspect.

Finally, our analytical model belongs to the same costly-search lineage as \citet{weitzman1979optimal}, who derives reservation values for independently distributed boxes, and \citet{rothschild1974searching}, who shows how learning an unknown opportunity distribution changes sequential search. \citet{choi2020optimal} studies selection effects when each option combines an observed factor with an independently hidden factor, and \citet{ursu2018power} uses a Weitzman-style model to distinguish the effect of ranking on hotel search from its effect on purchase conditional on search. Our core assumptions differ in two respects. First, items share an unknown page-quality component, so inspecting one item changes beliefs about all remaining items. Second, the list is ordered by noisy relevance signals whose realizations the user does not observe, making rank itself informative. The model likewise shares with interactive information retrieval an emphasis on querying and assessment costs, page-level stopping, and user effort \citep{azzopardi2011economics,maxwell2018scent,mcgregor2023cost}. We place these forces inside a Bayesian decision problem in which both rank and the realized inspection history shape the value of continued search.

\subsection*{Research Limitation}

We comment briefly on three abstractions.

The standard Gaussian-conjugate specification in Assumptions~\ref{assn:gaussian} and~\ref{assn:quantile} is adopted for tractability. The closed-form expressions we get from the conjugate prior help sharpen the discussions and pinpoint the mechanisms. The general qualitative architecture of our analysis, including linear traversal, the two-dimensional reservation rule, and the trust/commit decomposition, does not rely on the parametric form. 

We also treat the SERP as a single ranked list. Real webpages mix organic results with sponsored slots, carousels, and side panels, whose placement may reflect business considerations rather than predicted relevance. The abstraction is standard in the classical click-modeling literature \citep{craswell2008experimental, chapelle2009dynamic, dupret2008userbrowsing, chuklin2015click}, and accommodating multiple surfaces is orthogonal to the questions studied here.

Finally, we model inspection as a binary action (skip or fully reveal) whereas real users may have access to partial signals before deciding to click (e.g., reading a snippet). Layering a snippet stage on the present framework introduces a hierarchical signal structure, which significantly lengthens the paper. Given the density of this paper's current form, we decide to leave it to future work.

\section{Basic Setup}\label{sec:setup}

A risk-neutral user makes a request to a search engine. The system retrieves $N$ items in response. Each item is characterized by a latent relevance $x \in \R$.
The true relevance of each item, however, is unobserved to the system. For any retrieved item $i$, the system only has a noisy signal $z_i = x_i + e_i$, where $e_i$ is an idiosyncratic noise term independent of $x_i$. The system observes $z_1, \ldots, z_N$, and returns items as a ranked list in decreasing order of $z$. We relabel items so that $z_1 \ge z_2 \ge \cdots \ge z_N$, with ties broken arbitrarily.\footnote{Under the continuous distributions specified below, ties occur with probability zero and the inequalities are strict almost surely.} 

The user observes the ranked list but not the underlying signals $z$, but she understands that items are ordered by decreasing $z$.
Before inspecting the list, the user also does not observe the relevance of each individual item.
The user has an outside option and let $x_b \in \R$ denote the relevance of this outside option.

At each decision epoch $\tau = 0, 1, 2, \ldots, N$,\footnote{We adopt a $0$-indexed convention so that $\tau$ coincides with the number of inspections already performed.} the user may choose to inspect an arbitrary item from the set of uninspected items $\Uunins_{\tau}$, bearing an inspection cost $c > 0$. Inspecting item $i$ reveals its true relevance $x_i$ to the user. Alternatively, the user may choose to stop. If the user stops at epoch $\tau$, she takes the most relevant item she has ever inspected or her outside option. Let $\Iins_{\tau}$ denote the set of already-inspected items and define the best available relevance so far $M_\tau := \max\big(x_b, \max_{i \in \Iins_\tau} x_i\big)$, with $M_0 := x_b$. Her realized payoff is simply $U = M_\tau - \tau c$.

Let $\theta_x$ index the relevance distribution and $\theta_e$ index the noise distribution of the retrieved items. Conditional on $(\theta_x, \theta_e)$, relevances $\{x_i\}$ and noises $\{e_i\}$ are independent, with $x_i \mid \theta_x \sim F(\cdot \mid \theta_x)$ and $e_i \mid \theta_e \sim G(\cdot \mid \theta_e)$. In general, the user may hold priors on $\theta_x$ and $\theta_e$.

Let $H_\tau$ denote the user's information at epoch $\tau$: the ranked list, the priors, and the revealed relevances $\{(i, x_i) : i \in \Iins_\tau\}$ in the order they were inspected. A policy is a mapping from histories to $\{\mathrm{stop}\} \cup \Uunins_\tau$. The user maximizes $\E[U]$ over measurable policies.

Inspecting an item serves two roles. First, a relevance that exceeds the running maximum $M_\tau$ directly raises the user's payoff. Second, even an item whose relevance falls below $M_\tau$ sharpens the posterior on~$\theta_x$, improving the user's forecast of the remaining uninspected items. Continuation therefore has option value from both sources: the chance of finding a new best and the informational gain that refines future decisions.
At each decision epoch the user weighs this option value against the fixed cost~$c$. Because the posterior and the running maximum both evolve with each inspection, the marginal benefit of the next inspection depends on the entire history~$H_\tau$ dynamically. The user needs to find the policy that optimally balances the declining marginal benefit of additional inspections against her constant cost.

\paragraph{A note on interpreting inspection.} The relevance-revealing event we call inspection does not have a single empirical analogue. One reading, implicit in the cascade family and its descendants \citep{craswell2008experimental, chuklin2015click}, treats viewing a result on the SERP as the event, with a click read as a downstream signal of commitment or conversion. A second reading treats the click itself as part of inspection: a shopper who clicks through ten product pages to check whether any of them fits her need is sampling, not yet buying, and only the eventual checkout counts as a conversion. The theoretical results do not depend on choosing between these readings; each empirical application must make the mapping explicit. In the analytical model, an inspection is whichever action resolves the latent relevance $x_i$ of an item, and a conversion is whichever downstream event the user produces once she fixes on the item she will keep. A reader who prefers the first reading can identify $\Iins_\tau$ with the set of viewed results and, when the application supports it, read clicks as conversion events; a reader who prefers the second can identify $\Iins_\tau$ with the set of viewed and clicked items and read purchases or other completions as conversions. The reservation rule, the depth distribution, and the partial-identification results that follow go through under either reading; only the empirical mapping changes.

\subsection*{A Tractable Parametric Specification}

To obtain a model that is both transparent and amenable to closed-form analysis, we impose the following simplifying assumptions.

\begin{assumption}[Gaussian Specification]\label{assn:gaussian}
Conditional on $\mu$, the sequences $\{x_i\}$ and $\{e_i\}$ are each iid Gaussian and mutually independent, with $x_i \mid \mu \sim \N(\mu, \sigma_x^2)$, $e_i \sim \N(0, \sigma_e^2)$,
and $\sigma_x^2$, $\sigma_e^2$ fixed and known. The unknown parameter is the relevance mean, endowed with the conjugate prior $\mu \sim \N(m_0, v_0)$, where $m_0 \in \R$ is the \emph{prior mean} and $v_0 > 0$ the \emph{prior variance} of the relevance mean $\mu$ shared by the retrieved items.
\end{assumption}

Throughout, let $\sigma_z^2 := \sigma_x^2 + \sigma_e^2$ and define the \emph{reliability ratio} $\rho := {\sigma_x^2}/{\sigma_z^2} \in (0, 1)$. The reliability ratio governs how informative the ranking is about true relevance. At $\rho \uparrow 1$, the signal is noise-free and rank coincides with true relevance order; at $\rho \downarrow 0$, the ranking is uninformative.

To avoid carrying the full joint law of order statistics as a latent state, we adopt the standard Thurstonian approximation in which the position of an item in the ranked list is mapped to a fixed quantile of the signal distribution.

\begin{assumption}[Rank-to-Quantile Reduction]\label{assn:quantile}
The rank $i$ of an item is identified with the $i/(N+1)$ quantile from the top of the signal distribution:
$z_i = \mu + \sigma_z q_i$, where $ q_i := \Phi^{-1}\left(1 - {i}/{(N+1)}\right)$ and $\Phi$ denotes the standard Normal CDF.
\end{assumption}

The quantile sequence $\{q_i\}$ is strictly decreasing in $i$ and known. The approximation error vanishes as $N \to \infty$ and is small for moderate $N$. The sequence summarizes the ranking as a deterministic rank-specific shift. A derivation via the probability integral transform is provided in Appendix~\ref{app:quantile}.

\subsection*{Interior Inspection}\label{sec:noinspect}

Two summaries of the primitives will recur throughout the analysis: the \emph{rank-specific shift} $\alpha_i := \rho\sigma_z q_i = \sigma_x^2 q_i / \sigma_z$, which is the expected boost in relevance from being placed at rank $i$, and the \emph{residual variance} $\sigma_\eta^2 := \sigma_x^2(1-\rho) = {\sigma_x^2 \sigma_e^2}/{\sigma_z^2}$, the irreducible idiosyncratic noise that remains after extracting the rank-specific shift. Both quantities will reappear formally in Lemma~\ref{lem:predictive}; we introduce them here only to state the next assumption compactly.

To rule out the uninteresting corner case in which the user does not bother to inspect any item and walks away with the outside option, we impose a condition guaranteeing that the first inspection is worth it.

\begin{assumption}[Interior Solution]\label{assn:interior}
The fixed inspection cost satisfies
\begin{equation}\label{eq:noinspect}
c < \sqrt{v_0 + \sigma_\eta^2}\phi(d_0) - (x_b - m_0 - \alpha_1)\Phi(-d_0),
\end{equation}
where $d_0 := (x_b - m_0 - \alpha_1)/\sqrt{v_0 + \sigma_\eta^2}$ and $\phi$ denotes the standard Normal density function. Moreover, for any fixed $c$, \eqref{eq:noinspect} holds for all sufficiently large $v_0$.
\end{assumption}

The right-hand side of \eqref{eq:noinspect} is the expected gain from inspecting the top-ranked item once, accounting for the user's option to fall back to $x_b$ if the realized $x_1$ disappoints. Assumption~\ref{assn:interior} says this one-shot gain strictly covers the inspection cost $c$. Lemma~\ref{lem:tau-positive} in Appendix~\ref{app:auxiliary} shows that, under this assumption, the user always inspects at least one item. The complementary case admits the uninteresting corner case, so we exclude it for the remainder of the paper.

\section{Bellman Equation and Optimal Stopping Rule}

Combining the Gaussian conditional $x_i \mid \mu, z_i \sim \N(\rho z_i + (1-\rho)\mu, \sigma_x^2(1-\rho))$ with Assumption~\ref{assn:quantile} yields the following lemma.

\begin{lemma}[Rank Specific Predictive]\label{lem:predictive}
Under Assumptions \ref{assn:gaussian} and \ref{assn:quantile}, $x_i \mid \mu \sim \N(\mu + \alpha_i, \sigma_\eta^2)$ independent across $i$.
\end{lemma}

Complete proofs of all formal results appear in Appendix~\ref{app:auxiliary}.

Lemma~\ref{lem:predictive} splits inspection into two parts that the user can reason about separately: a \emph{known} adjustment that depends only on the rank and a \emph{random} component whose distribution is the same everywhere on the list. Each position comes pre-labeled with the rank-specific shift $\alpha_i$; once it is subtracted off, every position delivers an iid draw $y_i := x_i - \alpha_i \sim \N(\mu, \sigma_\eta^2)$. The rank-specific shift $\alpha_i$ is a feature of the \emph{position}, not of the item that landed in it. It is fixed as soon as the list is rendered. It is positive at the top and negative at the bottom, summing to zero across the list. The formula $\alpha_i = (\sigma_x^2/\sigma_z)q_i$ has a simple reading: $q_i$ is the quantile of rank $i$ on the latent score, and $\sigma_x^2/\sigma_z$ is the slope of the best linear predictor of $x$ from $z$. Multiplying them gives the expected change in $x$ per rank slot. Since $\rho \in (0,1)$, $\alpha_i$ is a shrunken version of the raw quantile $\sigma_z q_i$. The shrinkage factor $\rho$ is the signal-to-noise ratio of the ranker. The scaling is the Bayesian update of the prior mean given the one piece of information ``this item is in position $i$.''

This setup makes the model analytically clean: all the subtlety of ``how much should I trust the ranking'' is absorbed into a set of constants that can be computed up front, leaving inspection to do one simple job, Bayesian updating on $\mu$.

\subsection*{Linear Traversal}\label{sec:rankorder}

On its face the user's decision problem is combinatorial: at each step she must choose from up to $N$ options to inspect. The next result shows this choice is in fact binary: she only decides between inspecting the highest-ranked uninspected item or stopping. 

\begin{lemma}[Rank-Order Dominance]\label{lem:rankorder}
Under Assumptions \ref{assn:gaussian} and \ref{assn:quantile}, for every policy $\pi$ there exists a rank-order policy $\pi'$ (i.e., one that, whenever it inspects, inspects the highest-ranked uninspected item) such that $\E[U(\pi')] \ge \E[U(\pi)]$.
\end{lemma}

The reason behind this result is simple. Higher ranked slots offer better expected relevance, while each inspection has the same cost and supplies the same statistical information about page quality. Thus, inspecting the highest-ranked uninspected item dominates inspecting any other item. Consequently, the user's remaining choice is when to stop. An optimal rank-order policy is characterized by a history-dependent stopping time $T \in \{0, 1, \ldots, N\}$, where $T = 0$ means she stops immediately with the outside option.

\subsection*{Posterior Dynamics}

Let the user have inspected ranks $1, \ldots, t$ by epoch $t$. Throughout what follows, we denote by $m_t$ and $v_t$ the \emph{posterior mean} and \emph{posterior variance} of $\mu$ after $t$ inspections. Slightly abusing the notation, we let $H_t := \{(i, x_i) : i = 1, \ldots, t\}$ denote her information at that step, which includes the inspected ranks and their realized relevances. Combining with the Gaussian prior $\mu \sim \N(m_0, v_0)$ gives the standard conjugate posterior, as follows. 

\begin{lemma}[Posterior and Predictive Update]\label{lem:posterior}
Under Assumptions \ref{assn:gaussian} and \ref{assn:quantile}, after inspecting ranks $1, \ldots, t$, the posterior of $\mu$ is $\mu \mid H_t \sim \N(m_t, v_t)$, with $v_t^{-1} = v_0^{-1} + t/\sigma_\eta^2$ and $m_t = v_t({m_0}/{v_0} + \sum_{i=1}^{t}(x_i - \alpha_i) / {\sigma_\eta^2})$. For $t \le N - 1$, the posterior predictive for the relevance of rank $t+1$ prior to its inspection is $x_{t+1} \mid H_t \sim \N(m_t + \alpha_{t+1}, v_t + \sigma_\eta^2)$.
\end{lemma}

Intuitively, the user strips the known rank-specific shift $\alpha_i$ out of each revealed relevance to recover an iid signal of the underlying mean $\mu$, then averages these signals with her prior belief to update her belief.
The posterior update admits a standard interpretation. Precisions (reciprocals of variances) are additive: each inspection contributes one unit of information worth $1/\sigma_\eta^2$ to the precision, on top of the prior precision $1/v_0$. The posterior mean is a precision-weighted average of the prior mean $m_0$ and the observations $x_i - \alpha_i$ obtained after subtracting the rank-specific shifts.

The predictive then follows by adding the rank-specific shift and the irreducible item-level noise. Before inspecting rank $t+1$, the user's best guess of its relevance is the sum of two components. The first, $m_t$, is her current posterior mean about the page relevance $\mu$; the second, $\alpha_{t+1}$, is the known rank-specific shift that captures how much higher (or lower) a rank-$(t+1)$ item's relevance should be relative to an average item, based purely on its position in the list. As the user inspects more items, $m_t$ adjusts toward the true $\mu$ while $\alpha_{t+1}$ remains fixed as a pre-computed feature of rank $t+1$. Her uncertainty about $x_{t+1}$ has two sources as well: $v_t$, which measures what she does not yet know about $\mu$ and shrinks as she inspects more items, and $\sigma_\eta^2$, the residual idiosyncratic noise that cannot be learned away.

\subsection*{Bellman Equation}

Under linear traversal, the sufficient state at epoch $t$ is the pair $(M_t, m_t)$, where $M_t$ is the best relevance revealed so far (including $x_b$). Let $V_t(M_t, m_t)$ denote the value function at epoch $t \in \{0, 1, \ldots, N\}$. It satisfies
\begin{equation}\label{eq:bellman}
V_t(M_t, m_t) = \max\left\{M_t,\E_{x_{t+1}}\left[V_{t+1}\big(\max(M_t, x_{t+1}),m_{t+1}(x_{t+1})\big)\right]-c\right\},
\end{equation}
where the expectation is taken under $x_{t+1} \sim \N(m_t + \alpha_{t+1}, v_t + \sigma_\eta^2)$, and the belief update given $x_{t+1}$ is the affine map
\begin{equation}\label{eq:update}
m_{t+1}(x_{t+1}) = \frac{v_{t+1}}{v_t}m_t + \frac{v_{t+1}}{\sigma_\eta^2}(x_{t+1} - \alpha_{t+1}), \qquad v_{t+1}^{-1} = v_t^{-1} + \sigma_\eta^{-2}.
\end{equation}
The terminal condition is $V_N(M_N, m_N) = M_N$.

The Bellman equation \eqref{eq:bellman} makes the tradeoff the user faces at each epoch explicit. The left branch inside the max is the \emph{stop} option: she walks away with the best relevance she has revealed to date, $M_t$. The right branch is the \emph{inspect} option: she pays the deterministic cost $c$, reveals $x_{t+1}$, and transitions to a new state in which the current best becomes $\max(M_t, x_{t+1})$ and her posterior mean becomes $m_{t+1}(x_{t+1})$. The expected payoff from the right branch has two components of value beyond $M_t$: the direct \emph{upside option}, $\E[(x_{t+1} - M_t)^+]$, which captures the possibility that the inspected item beats the current best; and the \emph{learning option}, which captures how observing $x_{t+1}$ refines her belief about $\mu$ and thereby improves future continuation decisions.

\subsection*{Optimal Stopping Rule}

The Bellman equation \eqref{eq:bellman} induces a stopping rule of remarkably simple form: at each epoch $t$, the user stops as soon as the difference between her best find and her current posterior mean, $M_t - m_t$, exceeds a deterministic, depth-dependent threshold. The decision depends on the running maximum $M_t$ and the posterior mean $m_t$ only through their difference, and on the user's outside option $x_b$ only insofar as it enters $M_t$.

\begin{proposition}[Standout Stopping Rule]\label{prop:reservation}
For each $t \in \{0, 1, \ldots, N-1\}$, there exists a deterministic constant $\kappa_t^\star \in \R$, depending on $t$, $v_0$, $\sigma_\eta^2$, $c$, $\{\alpha_i\}_{i \ge t+1}$, and $N$ but not on $m_t$ or on the outside option $x_b$, such that the optimal policy stops at epoch $t$ if and only if
\begin{equation}\label{eq:R-affine-opt}
M_t - m_t  \ge \alpha_{t+1} + \kappa_t^\star.
\end{equation}
\end{proposition}

Because the threshold \eqref{eq:R-affine-opt} depends on $M_t$ and $m_t$ only through their difference, the decision-relevant scalar at every epoch is the gap between the best find so far and the posterior mean of the page. Throughout the paper we refer to this scalar as the \emph{lead}: $L_t := M_t - m_t$. The lead carries the entire content of the standout rule: the user continues whenever $L_t < \alpha_{t+1} + \kappa_t^\star$ and stops otherwise.\footnote{We implicitly assume the user stops when she is indifferent. As before, under continuous distributions, the inequality is strict almost surely.}

The linear form of \eqref{eq:R-affine-opt} sharpens the qualitative picture by decomposing the user's patience into two pieces. The user's expected relevance of the next item is $m_t + \alpha_{t+1}$, and $\kappa_t^\star$ is the margin by which her current best $M_t$ must exceed that expected relevance before she commits. A tractable closed form for the threshold appears below under a one-step look-ahead heuristic.

Empirically, holding the query fixed, stop probability is more predictable by the lead $L_t$ than by $M_t$ alone.
Three forces drive the lead toward the threshold as inspection proceeds, eventually tipping the rule in favor of stopping. First, every additional inspection can only raise $M_t$, while $m_t$ is a martingale under the predictive measure, so the lead on average trends upward over time. Second, the rank-specific shift $\alpha_{t+1}$ is strictly decreasing in $t$ (on average the items ahead are less promising), so the bar that the lead must clear falls as the user moves down the list. Third, the residual uncertainty about $\mu$ shrinks with each inspection. This lowers the option value of further learning and pulls the intercept $\kappa_t^\star$ down toward a noise floor that depends only on $\sigma_\eta^2$ and the remaining horizon. Costs, by contrast, remain constant per inspection. The user optimally stops at the first rank where her accumulated lead clears a threshold that signals a further inspection is not worth it.

\subsection*{Myopic Policy}

The optimal intercept $\kappa_t^\star$ does not admit a clean closed form, but a tractable benchmark does: the one-step look-ahead (myopic) policy, which inspects whenever the expected one-step gain exceeds the cost, treating the next step as if it were terminal. The myopic threshold shares the same standout format and isolates the option-value channel in closed form. Define
\begin{equation}\label{eq:gdef}
g(d) := \phi(d) - d\Phi(-d),
\end{equation}
which is a strictly decreasing bijection from $\R$ onto $(0, \infty)$: $g'(d)=-\Phi(-d)<0$, and $g(d)$ decreases from infinity to zero. Hence $g^{-1}$ is well defined.

\begin{proposition}[Myopic Policy]\label{prop:myopic}
Let $\sigma_t^{\star 2} := v_t + \sigma_\eta^2$ denote the one-step predictive variance after $t$ inspections. Under the myopic policy, the user stops at epoch $t$ if and only if
\begin{equation}\label{eq:Rmyo-general}
M_t - m_t  \ge \alpha_{t+1} + \kappa_t, \qquad \kappa_t := \sigma_t^\star g^{-1}\left(c/\sigma_t^\star\right)
\end{equation}
and the intercept $\kappa_t$ is strictly decreasing in $t$, with
\begin{equation}\label{eq:kappa-bounds-general}
\kappa_0 = \sqrt{v_0 + \sigma_\eta^2}\, g^{-1}\!\left(\tfrac{c}{\sqrt{v_0 + \sigma_\eta^2}}\right), \qquad \lim_{t \to \infty} \kappa_t = \sigma_\eta\, g^{-1}\!\left(c/\sigma_\eta\right) =: \kappa_\infty.
\end{equation}
\end{proposition}

The constant $\kappa_t$ is the margin at which the option value of one more inspection just covers its cost. It does not vanish even in the absence of learning: a user with known $\mu$ would still apply the adjustment $\kappa_\infty = \sigma_\eta g^{-1}(c/\sigma_\eta)$ set by the residual idiosyncratic noise in the next draw. Learning adds a surcharge $\kappa_t - \kappa_\infty > 0$, with $\kappa_0 > \kappa_1 > \cdots > \kappa_\infty$, driven by the predictive-variance decomposition $\sigma_t^{\star 2} = v_t + \sigma_\eta^2$. Posterior uncertainty about $\mu$ fattens the upper tail of the next draw and inflates the option value of one more inspection; it shrinks to zero by $v_{t+1}^{-1} = v_t^{-1} + \sigma_\eta^{-2}$, leaving only the irreducible item-level noise $\sigma_\eta^2$.

The myopic gap bounds the optimal gap $\kappa_t^\star$ from below: $\kappa_t \le \kappa_t^\star$ for every $t$. Lemma~\ref{lem:myopic-lb} in Appendix~\ref{app:auxiliary} proves this bound, and Appendix~\ref{app:numerical} shows numerically that the two rules induce nearly identical gap values and inspection-depth distributions (Figures~\ref{fig:kappa-compare} and~\ref{fig:depth-compare}).

\section{Inspection Depth}\label{sec:inspection_depth}
Write $T \in \{0, 1, \ldots, N\}$ for the user's realized inspection depth, and this section characterizes its full distribution. The optimal policy of Proposition~\ref{prop:reservation} seemingly demands the two-dimensional state $(M_t, m_t)$, but the lead $L_t$ is a one-dimensional Markov chain with a closed-form one-step transition density. The recursion that follows exposes the mechanism of SERP view: a tug-of-war between a rising running max and a refining posterior mean.

For the remainder of the section, write $\xi_t := x_t - (m_{t-1} + \alpha_t)$ for the user's \emph{surprise} at epoch $t$, the difference between the realized draw and her conditional expectation $m_{t-1} + \alpha_t$. By Lemma~\ref{lem:posterior}, conditional on the history $H_{t-1}$ the surprise is mean-zero Gaussian, $\xi_t \mid H_{t-1} \sim \N(0, \sigma_{t-1}^{\star 2})$, and the sequence $\{\xi_t\}_{t \ge 1}$ is conditionally independent across epochs.

Let $\omega_t := v_{t-1}/(v_{t-1} + \sigma_\eta^2) \in (0,1)$ denote the Bayesian weight on the $t$-th observation. The posterior update can be written as $m_t - m_{t-1} = \omega_t\xi_t$, so $\omega_t$ is the fraction of the surprise that the user folds into her posterior mean about $\mu$. The Bayesian weight $\omega_t$ is large when she is still uncertain ($v_{t-1}$ large) and her belief tracks new draws aggressively, small when her uncertainty has shrunk and she barely moves. The sequence $\{\omega_t\}_{t \ge 1}$ is deterministic, strictly decreasing in $t$, and tends to zero as $t \to \infty$.

By Proposition~\ref{prop:reservation}, the user stops at epoch $t$ if and only if $L_t \ge \alpha_{t+1} + \kappa_t^\star$, with the convention $\alpha_{N+1} + \kappa_N^\star := -\infty$ so the user is forced to stop at the horizon. The next lemma reduces the dynamics to the lead alone.

\begin{lemma}[Lead Chain]\label{lem:lead-markov}
The lead $L_t$ is a Markov chain on $\R$, with one-step recursion
\begin{equation}\label{eq:L-recursion}
L_t = \max\bigl(L_{t-1}, \alpha_t + \xi_t\bigr) - \omega_t\xi_t.
\end{equation}
The conditional law of $L_t$ given $L_{t-1} = l$ is supported on $[L_{\min}(l), \infty)$, where $L_{\min}(l) := (1-\omega_t)l + \omega_t\alpha_t$, and on this support has density
\begin{equation}\label{eq:L-density}
f_{L_t \mid L_{t-1} = l}(y) = \frac{1}{(1-\omega_t)\sigma_{t-1}^\star}\phi\left(\frac{y - \alpha_t}{(1-\omega_t)\sigma_{t-1}^\star}\right) + \frac{1}{\omega_t\sigma_{t-1}^\star}\phi\left(\frac{l - y}{\omega_t\sigma_{t-1}^\star}\right).
\end{equation}
\end{lemma}

The recursion \eqref{eq:L-recursion} is piecewise affine in $\xi_t$, with kink at the surprise level $\xi^\star_t := L_{t-1} - \alpha_t$ at which the new draw ties the running maximum. Above the kink the draw beats the running maximum (a \emph{discovery}); below it the draw falls short (a \emph{disappointment}). The two branches
\begin{equation}\label{eq:L-twobranch}
L_t = \begin{cases}
L_{t-1} - \omega_t\xi_t & \text{if } \xi_t \le \xi^\star_t \quad (\text{disappointment: } x_t \le M_{t-1}),\\
\alpha_t + (1 - \omega_t)\xi_t & \text{if } \xi_t > \xi^\star_t \quad (\text{discovery: } x_t > M_{t-1}),
\end{cases}
\end{equation}
contribute the two terms in the transition density \eqref{eq:L-density}.

On a \emph{disappointment}, the running maximum is frozen ($M_t = M_{t-1}$), so the lead moves only because the belief moves: $m_t = m_{t-1} + \omega_t\xi_t$ produces $L_t - L_{t-1} = -\omega_t\xi_t$. A draw below expectation revises $\mu$ down and grows the lead; a draw above expectation but short of the running maximum revises $\mu$ up and shrinks it. Either way the move is dampened by the Bayesian weight $\omega_t$, which controls how aggressively the user revises. The disappointment term in \eqref{eq:L-density} accordingly has scale $\omega_t\sigma_{t-1}^\star$ around $L_{t-1}$: large for an early user who is still revising sharply, small for a converged user whose belief barely budges.

On a \emph{discovery}, both halves of the lead move: $M_t$ jumps onto the new draw and $m_t$ shifts by $\omega_t\xi_t$, giving $L_t = \alpha_t + (1-\omega_t)\xi_t$ with $L_{t-1}$ dropping out completely. A discovery \emph{overwrites} the previous lead: the new value depends only on the rank-specific shift $\alpha_t$ and the fraction $1-\omega_t$ of the surprise that is not absorbed into the posterior. A converged user keeps almost the entire upside as added lead and is much more likely to stop next epoch; an early user sees most of the windfall vanish into her revised expectations and is correspondingly less impressed. The discovery term in \eqref{eq:L-density} accordingly has scale $(1-\omega_t)\sigma_{t-1}^\star$ around $\alpha_t$, with the opposite dependence on $\omega_t$ from the disappointment term.

Lemma~\ref{lem:lead-markov}, albeit sharp, does not directly speak to how each inspection's revelation engages with the user's continuation and stopping decision. Translating the lead recursion back to item relevance using $x_t = m_{t-1} + \alpha_t + \xi_t$ converts the stopping rule into a continuation interval on the realized draw, with an explicit V-shape geometry at every epoch.

\begin{proposition}[Conditional Continuation]\label{prop:cont-interval}
Define $r_t := \alpha_{t+1} + \kappa_t^\star$. For an epoch $t \ge 1$ and a history $H_{t-1}$ on which the user has reached rank $t$,
\begin{enumerate}
\item[(a)] \emph{Dynamic trust regime:} if $r_t \le (1-\omega_t)L_{t-1} + \omega_t\alpha_t$, then the user stops at rank $t$ for every realization of $x_t$.
\item[(b)] \emph{Dynamic explore regime:} if $r_t > (1-\omega_t)L_{t-1} + \omega_t\alpha_t$, then the user continues when 
\begin{equation}\label{eq:cont-interval}
\underline{x}_t(H_{t-1}) < x_t < \overline{x}_t(H_{t-1}),
\end{equation}
where $\underline{x}_t(H_{t-1}) := m_{t-1} + \alpha_t + (L_{t-1} - r_t)/\omega_t$ and $\overline{x}_t(H_{t-1}) := m_{t-1} + \alpha_t + (r_t - \alpha_t)/(1 - \omega_t)$.
\end{enumerate}
\end{proposition}

Proposition~\ref{prop:cont-interval} partitions every epoch's continuation problem into two regimes. The dynamic trust regime says that, before the user has inspected rank $t$, her history has already locked in stopping: $x_t$ does not enter the decision at all. In the dynamic explore regime, continuation past rank $t$ requires $x_t$ to land in a middle band, not so low as to make the outside option look safer (cut losses), not so high as to make rank $t$ a standout worth committing to.

The behavioral story is uniform across epochs. A very low draw signals that the page is weak and nudges the user toward her outside option; a very high draw delivers a standout whose lead beats the option value of further search; the intermediate range is the set of items that neither resolve the question nor settle it. The arresting implication is that, with history fixed, a user who saw a mediocre rank-$t$ item is more likely to keep going than a user who saw an excellent one. No monotone click model \citep{craswell2008experimental, chapelle2009dynamic, dupret2008userbrowsing} can match this shape, because in those models satisfaction with the inspected item only ever increases the probability of stopping; the mediocre-driven continuation channel is absent by assumption. If the platform has any proxy for realized rank-$t$ relevance, the rational-user model predicts a nonmonotonic relationship between that proxy and continuation to rank $t+1$, and the cascade family predicts monotonicity. The two are distinguishable in the data without any need to manipulate the ranker.

The upper endpoint $\overline{x}_t$ depends on the Bayesian weight $\omega_t$, which is large when prior uncertainty is high. When the user is still uncertain about the page mean, a single high draw raises her belief about the rest of the list almost as much as it raises the running max, and the lead barely moves. The commitment tail is then unavailable to her, because no single observation can convince her that she has reached the best the list has to offer. As belief stabilizes with inspections, the upper threshold descends, and the commitment tail opens up. The model accordingly predicts an attenuated upper truncation when the user's prior is diffuse (novel queries, or users new to the topic) and a more symmetric inverted-U when her prior is tight (familiar queries, or experienced users). Two natural cross-section tests follow: comparing the shape of continuation against relevance proxies across queries of varying familiarity, and tracking how that shape changes within a single search as the user's prior tightens.

The first-stop, where the user has not yet inspected any item, is a special case that the two regimes are pinned down only by primitives. With $\omega_1 := v_0/(v_0 + \sigma_\eta^2)$ and $L_0 = x_b - m_0$, the predictive draw is $x_1 \mid H_0 \sim \N(m_0 + \alpha_1, v_0 + \sigma_\eta^2)$, $m_1(x_1) = (1-\omega_1)m_0 + \omega_1(x_1 - \alpha_1)$, and $M_1(x_1) = \max(x_b, x_1)$. Define
\[
h(x_1) := M_1(x_1) - m_1(x_1) - \alpha_2 - \kappa_1^\star = \begin{cases}
  -\omega_1 x_1 + x_b + \omega_1\alpha_1 - (1-\omega_1)m_0 - \alpha_2 - \kappa_1^\star & \text{if } x_1 \le x_b, \\
  (1-\omega_1) x_1 + \omega_1\alpha_1 - (1-\omega_1)m_0 - \alpha_2 - \kappa_1^\star & \text{if } x_1 > x_b,
\end{cases}
\]
a V-shape in $x_1$ with kink at $x_1 = x_b$. Then the user stops at $T = 1$ on the event $\{h(x_1) \ge 0\}$.

\begin{corollary}[First-Stop Specialization]\label{cor:tau-one}
Specializing Proposition~\ref{prop:cont-interval} to $t = 1$ gives two cases:
\begin{enumerate}
\item[(a)] \emph{Trust regime:} if $\alpha_1 - \alpha_2 \ge \kappa_1^\star + (1-\omega_1)(m_0 + \alpha_1 - x_b)$, then $\Pr(T = 1) = 1$.
\item[(b)] \emph{Explore regime:} if the trust condition fails, define
\begin{equation}\label{eq:s1-explicit}
\underline{x}_1 := m_0 + \alpha_1 - \frac{m_0 + \alpha_2 + \kappa_1^\star - x_b}{\omega_1}, \qquad \overline{x}_1 := m_0 + \alpha_1 + \frac{\kappa_1^\star - (\alpha_1 - \alpha_2)}{1 - \omega_1}.
\end{equation}
Then $\underline{x}_1 < x_b < \overline{x}_1$ are the two solutions of $h(\cdot) = 0$, and
\begin{equation}\label{eq:Ptau1-twopiece}
\begin{aligned}
\Pr(T = 1)
&= \Phi\left(\frac{\underline{x}_1 - m_0 - \alpha_1}{\sqrt{v_0 + \sigma_\eta^2}}\right)
&& \text{(cut-losses branch: } x_1 \le \underline{x}_1\text{)} \\
&\quad + \Phi\left(\frac{m_0 + \alpha_1 - \overline{x}_1}{\sqrt{v_0 + \sigma_\eta^2}}\right)
&& \text{(commit branch: } x_1 \ge \overline{x}_1\text{)}.
\end{aligned}
\end{equation}
\end{enumerate}
\end{corollary}

In the explore regime, the continuation interval $(\underline{x}_1, \overline{x}_1)$ is the set of \emph{mediocre} first observations. When $x_1 \ge \overline{x}_1$, the user stops because rank~$1$ is so good that further search is unlikely to surface anything better worth paying $c$ to look for. When $x_1 \le \underline{x}_1$, she stops because rank~$1$ being this disappointing tells her the whole list is probably weak, so the outside option becomes the safer bet. The same single-inspection observable is therefore consistent with both satisfaction and disappointment.

The trust regime is the case in which the ranker's ability to deliver on the very top position has earned the user's trust. Before any item is inspected, the rank-1 advantage $\alpha_1 - \alpha_2$ clears $\kappa_1^\star$ plus a buffer $(1-\omega_1)(m_0 + \alpha_1 - x_b)$ sized for the worst case $x_1 = x_b$. In the trust regime, the user still inspects rank~$1$, but only to learn what $x_1$ is. She then chooses between $x_1$ and her outside option without moving down the list.
In other words, a search that falls in the trust regime leaves no trace below the top slot. Whether a search lies in this regime depends only on query-time primitives, not on the realized $x_1$.
A more informative ranking widens the rank-1 advantage and shrinks the residual uncertainty around rank~$1$, pushing the buffer toward zero. Consequently, such a ranker drives more searches into the trust regime and starves the system of deeper observations. 

Formally, index ranking informativeness by the reliability ratio $\rho$, hold $\sigma_x^2$ fixed, and let the ranker noise vary as $\sigma_e^2(\rho)=\sigma_x^2(1-\rho)/\rho$, $\sigma_\eta^2(\rho)=\sigma_x^2(1-\rho)$, $\alpha_i(\rho)=\sigma_x\sqrt{\rho}\,q_i$. The next corollary shows that a sufficiently informative ranker generates no observations below the top slot.

\begin{corollary}[Winner's Curse]\label{cor:reliable-ranker-starvation}
Fix $N\ge 2$, $c>0$, $\sigma_x>0$, $v_0>0$, and $m_0,x_b\in\R$. Suppose Assumption~\ref{assn:interior} holds for all $\rho$ sufficiently close to one.\footnote{This is a nonempty restriction: it holds when inspection is sufficiently cheap, when the outside option is sufficiently unattractive, or when the prior variance $v_0$ is sufficiently large (Lemma~\ref{lem:tau-positive}).} Let $T_\rho$ be the optimal inspection depth in the environment with reliability ratio $\rho$. Then there exists $\bar\rho<1$ such that, for every $\rho\in(\bar\rho,1)$, $\Pr_\rho(T_\rho=1)=1$.
\end{corollary}

The limit $\rho\uparrow 1$ is the noise-free case in which the ranking coincides with true relevance order and the residual uncertainty around each rank vanishes. So every lower-ranked item is known to be worse than the one just inspected. A further inspection cannot improve the terminal maximum but costs $c>0$, so stopping is strictly optimal. The trust condition of Corollary~\ref{cor:tau-one}(a) holds at the limit with strictly positive slack, and therefore on an open neighborhood. 

The winner's curse is not biased logging or model misuse, but an equilibrium consequence of rational stopping against a ranker that has become sufficiently informative, and it tightens precisely as the ranking becomes more informative. For $\rho<1$ inside the neighborhood, the top signal, albeit informative, is not conclusive, so lower-ranked items can still beat rank~$1$. The trust condition holding nonetheless means lower ranks are never inspected and these residual inversions never enter the log.

Let $\nu_t$ denote the law of $L_t$ on the survival event $\{T \ge t\}$, with $\nu_0 = \delta_{x_b - m_0}$. The surviving population $\nu_t$ fans out in the early epochs as different users experience different draw sequences, and contracts later as the stopping region $[\alpha_{t+1} + \kappa_t^\star, \infty)$ absorbs mass. The Markov structure of $L_t$ yields a clean recursion for the stopping probabilities $\Pr(T = t)$.

\begin{proposition}[Stopping Time Recursion]\label{prop:tau-recursion}
Let $K_t^L$ denote the transition kernel of the lead chain (Lemma~\ref{lem:lead-markov}). For $t = 0, 1, \ldots, N-1$:
\begin{align}
\Pr(T = t) &= \nu_t\bigl([\alpha_{t+1} + \kappa_t^\star, \infty)\bigr), \label{eq:tau-mass}\\
\nu_{t+1}(\cdot) &= \int_{(-\infty, \alpha_{t+1} + \kappa_t^\star)} K_t^L(l, \cdot) \nu_t(dl), \label{eq:tau-pushforward}
\end{align}
and $\Pr(T = N) = \nu_N(\R)$. The masses $\Pr(T = 0), \ldots, \Pr(T = N)$ sum to one.
\end{proposition}

The recursion form tracks a single one-dimensional law $\nu_t$ across epochs and advances it by one transition at a time. For empirical work, where one wants $\Pr(T = t)$ as a function of model parameters whose gradient can be backpropagated through, the next section develops an equivalent representation as a truncated multivariate Gaussian probability.

\section{Identification from Inspection Depth}\label{sec:partial-id}

So far we have asked what the user does given the latent relevances. The empirical question often is the reverse. A scroll trace records that the user inspected five items and left; it does not say which items thrilled her, which bored her, or which made her bail. What does an observed $T = t$ reveal about the unobserved relevances $(x_1, \ldots, x_t)$, and how can that information be turned into a training signal for a feature-based relevance model?

In a common pointwise training construction, items the user did not engage with are labeled negatives and no further structure is imposed. The Bayesian-rational view says depth carries much more content than that. Reaching depth $t$ certifies that all previous inspections failed to clear the standout threshold; when $t<N$, a voluntary stop at $t$ additionally certifies that the rank-$t$ realization did clear it. Since both statements are inequalities in the inspected relevances, the observed depth confines the latent path to the region they define. Depth thus partially reveals what the user saw.

The section has two subsections. The first constructs the survival region explicitly, formally turning the user's reservation rule into a system of linear inequalities in the inspected relevances. The second builds a learning-to-rank likelihood on top of the survival region, develops a Geweke--Hajivassiliou--Keane (GHK) estimator \citep{hajivassiliou1996simulation} for its Gaussian probability, and reduces the user-side primitives to two scalars.

\subsection*{Survival as a System of Inequalities}

Survival to depth $t$ does not point-identify the latent path $(x_1, \ldots, x_{t-1})$, nor does it point-identify any single relevance. Instead, it reveals that the path lies inside the user's continuation region. The user keeps going only when each new draw lands in a middle band, and the bands themselves are determined endogenously by everything she has already seen. Reaching rank $t$ is thus a \emph{censoring event} of a high-dimensional latent vector, and the censoring region inherits its structure from the user's stopping rule.

By Proposition~\ref{prop:reservation}, not stopping at rank $i$ is the event $\{L_i < r_i\}$. Reaching rank $t$ is then the intersection of continuation events across the first $t-1$ epochs: $\{T \ge t\} = \{T \ge 1\} \cap \bigcap_{i=1}^{t-1} \{L_i < r_i\}$.
Under Assumption~\ref{assn:interior}, $\{T \ge 1\}$ has probability one, so we suppress it and work with the $t-1$ inequalities indexed by $i = 1, \ldots, t-1$.

Two ingredients turn each $\{L_i < r_i\}$ into a finite list of linear inequalities in $(x_1, \ldots, x_i)$. First, the posterior mean is affine in the inspected draws. Reading off Lemma~\ref{lem:posterior},
\[
m_i = v_i\left(\frac{m_0}{v_0} + \frac{1}{\sigma_\eta^2}\sum_{k=1}^{i}(x_k - \alpha_k)\right) = a_i + \gamma_i \sum_{k=1}^{i} x_k,
\]
where $\gamma_i := v_i/\sigma_\eta^2$ and $a_i := {v_i m_0} / {v_0} - \gamma_i \sum_{k=1}^{i}\alpha_k$.
Second, the running maximum is a max of $i + 1$ linear forms, $M_i = \max\{x_b, x_1, \ldots, x_i\}$, so the event $\{L_i < r_i\} = \{M_i < m_i + r_i\}$ is the conjunction of one inequality per argument of the max: one for the outside option,
\begin{equation}\label{eq:ineq-outside}
x_b < a_i + \gamma_i \sum_{k=1}^{i} x_k + r_i,
\end{equation}
and one for each previously inspected item,
\begin{equation}\label{eq:ineq-inspected}
x_j < a_i + \gamma_i \sum_{k=1}^{i} x_k + r_i, \qquad j = 1, \ldots, i.
\end{equation}
Each constraint says ``this option after inspecting through rank $i$ is not yet decisive.'' 

Stacking \eqref{eq:ineq-outside} and \eqref{eq:ineq-inspected} across $i = 1, \ldots, t-1$ defines the \emph{survival region} $\mathcal{C}_t := \{(x_1, \ldots, x_{t-1}) \in \R^{t-1} : L_i < r_i \text{ for every } 1 \le i \le t - 1\}$. Formally, this is a polyhedron in $\R^{t-1}$ with $\sum_{i=1}^{t-1}(i + 1) = t(t+1)/2 - 1$ inequalities. By construction, $\{T \ge t\}$ is the event $\{(x_1, \ldots, x_{t-1}) \in \mathcal{C}_t\}$. Thus, reaching rank $t$ set-identifies the pre-$t$ history through this system.

\paragraph{A $t = 3$ example.} Two epochs of inequalities define $\mathcal{C}_3 \subset \R^2$. The $i = 1$ pair, $x_b < m_1 + r_1$ and $x_1 < m_1 + r_1$ with $m_1 = a_1 + \gamma_1 x_1$, reduces to $\underline{x}_1 = (x_b - a_1 - r_1)/\gamma_1 < x_1 < (a_1 + r_1)/(1 - \gamma_1) = \overline{x}_1$. Therefore, epoch $1$ contributes a vertical strip on $x_1$, independent of $x_2$.

The $i = 2$ inequalities are $x_b < m_2 + r_2$, $x_1 < m_2 + r_2$, and $x_2 < m_2 + r_2$. Rearranging with $m_2 = a_2 + \gamma_2(x_1 + x_2)$ gives
\begin{align*}
\text{(cut-losses at $i = 2$)} \qquad & x_1 + x_2 > \frac{x_b - a_2 - r_2}{\gamma_2}, \\
\text{(rank-1 not yet decisive)} \qquad & x_2 > \frac{(1 - \gamma_2) x_1 - a_2 - r_2}{\gamma_2}, \\
\text{(rank-2 not yet decisive)} \qquad & x_2 < \frac{\gamma_2 x_1 + a_2 + r_2}{1 - \gamma_2}.
\end{align*}
Stacking all five inequalities, $\mathcal{C}_3$ is a bounded region in the $(x_1, x_2)$-plane illustrated in Figure~\ref{fig:c3-region}.

\begin{figure}[ht]
\centering
\begin{tikzpicture}[x=1.6cm, y=0.9cm, >=stealth, line cap=round, font=\small]
  \begin{scope}
    \clip (-3.434,-3.817) rectangle (2.434,3.000);

    \draw[black, dash dot] (-3.434,-0.2839) -- (2.434,-6.1510);
    \draw[black, dotted, thick] (-3.434,-9.0846) -- (2.434,2.6497);
    \draw[black, dash dot dot, thick] (-3.434,-0.6081) -- (2.434,2.3255);

    \draw[black, dashed] (-2.8336,-3.817) -- (-2.8336,3.000);
    \draw[black, dashed] (1.8336,-3.817) -- (1.8336,3.000);

    \fill[black!8, draw=black, line width=0.9pt] (1.8336,1.4497) -- (1.8336,2.0255) -- (-2.8336,-0.3081) -- (-2.8336,-0.8839) -- (-0.5000,-3.2175) -- cycle;

    \draw[black!85, densely dashed, line width=0.75pt] (-0.500,-3.817) -- (-0.500,3.000);
    \draw[black!85, densely dashed, line width=0.75pt] (-3.434,-0.500) -- (2.434,-0.500);
  \end{scope}

  \draw[->, gray!70] (-3.434,0) -- (2.484,0) node[right, black] {$x_1$};
  \draw[->, gray!70] (0,-3.817) -- (0,3.050) node[above, black] {$x_2$};
  \foreach \xv in {-2,-1,1,2} {
    \draw (\xv,-0.06) -- (\xv,0.06) node[below=2pt, scale=0.8] {$\xv$};
  }
  \foreach \yv in {-2,-1,1,2} {
    \draw (-0.04,\yv) -- (0.04,\yv) node[left=2pt, scale=0.8] {$\yv$};
  }

  \node[anchor=south, black, scale=0.85] at (-2.8336,3.000) {$\underline{x}_1$};
  \node[anchor=south, black, scale=0.85] at (1.8336,3.000) {$\overline{x}_1$};
  \node[anchor=south west, black!85, scale=0.95] at (-1.250,3.000) {$x_1{=}x_b$};
  \node[anchor=south west, black!85, scale=0.95] at (2.500,-0.800) {$x_2=x_b$};
  \node[black] at (-1.0,-1.5) {$\mathcal{C}_3$};
\end{tikzpicture}
\caption{Survival region $\mathcal{C}_3 \subset \R^2$ for the $t = 3$ example with parameters $v_0 = \sigma_\eta^2 = 1$, $m_0 = 0$, $x_b = -0.5$, $(\alpha_1, \alpha_2, \alpha_3) = (0.3, 0.1, 0)$, $c = 0.15$, and $\kappa_i$ taken from the myopic formula of Proposition~\ref{prop:myopic}. The shaded pentagon is the intersection of the $i = 1$ strip $(\underline{x}_1, \overline{x}_1) \times \R$ (dashed vertical boundaries) with the three $i = 2$ half-planes: the cut-losses constraint (dash-dotted), the rank-$1$ not-yet-decisive constraint (dotted, a lower bound on $x_2$ with positive slope in $x_1$), and the rank-$2$ not-yet-decisive constraint (dash-dot-dotted, an upper bound).}
\label{fig:c3-region}
\end{figure}

Figure~\ref{fig:c3-region} also shows why the inspected history contains more information than a sequence of negative labels. A sizable portion of $\mathcal{C}_3$ lies in $\{x_1 < x_b\} \cap \{x_2 < x_b\}$: both inspected items can be worse than the outside option, leaving $x_b$ as the user's running best. Survival to depth $3$ nevertheless confines their joint relevance. The cut-loss bound at $i = 2$ rules out the deeply negative tail $x_1 + x_2 \le (x_b - a_2 - r_2)/\gamma_2$, which would push the posterior below the continuation threshold. The rank-$1$ and rank-$2$ not-yet-decisive bounds exclude positive outliers that would already have triggered commitment, and the $i = 1$ strip $(\underline{x}_1, \overline{x}_1)$ supplies the corresponding restriction one epoch earlier. Thus, the two latent relevances lie in a bounded, history-dependent polygon whose location relative to $x_b$ is determined by the posterior path. An item-level negative label, in contrast, discards this joint restriction.

\subsection*{A Learning-to-Rank Likelihood}

The set-identification of the latent relevance path through the survival region has a natural empirical counterpart: it can be turned into a likelihood for training a relevance model directly from depth and conversion observations, without any explicit relevance labels. This subsection makes that translation in four steps. We introduce a system-side feature representation and a data-generating process, decompose the observable outcome of a session into a region of latent-relevance space, write the session likelihood as a Gaussian probability over that region, and evaluate the resulting probability with a differentiable GHK simulator.

\paragraph{Feature Model.} The system observes a feature vector $\mathbf{w}_i$ for each item displayed at rank $i$.\footnote{Assuming sessions are iid, we first work within a single session and suppress the session subscript.} These features encode what the system knows about the query-item pair. They are visible to the system and the analyst but not to the user, who continues to see only rank and the inspected relevances.

Following standard learning-to-rank practice, the system computes a score for each item as the deterministic output of a parametric model of the features,
\begin{equation}\label{eq:score-feature}
z_i=\mathscr{F}(\beta; \mathbf{w}_i),
\end{equation}
where $\mathscr F$ is an almost-everywhere differentiable parametric function and $\beta$ collects its parameters. In learning-to-rank language, $z_i$ is the system's prediction and the true relevance $x_i$ is the label against which it would ordinarily be trained. Here $x_i$ remains latent.

The analyst's task is to recover $\beta$ from the user's observed behavior. The analyst's conditional model of the latent relevance given features is
\begin{equation}\label{eq:x-feature}
x_i \mid \mathbf{w}_i \sim \N\bigl(z_i, \sigma_{\mathscr{F}}^2\bigr),
\end{equation}
with residual variance $\sigma_{\mathscr{F}}^2$.\footnote{The residual variance $\sigma_{\mathscr F}^2:=\Var(x_i\mid\mathbf w_i)$ differs from $\sigma_e^2=\Var(e_i)$ in the earlier relation $z_i=x_i+e_i$. The former includes both model misspecification and irreducible variation in $x_i$ conditional on $\mathbf w_i$, whereas the latter is the variance of the score-relevance gap in the theoretical ranking model. Only $\sigma_{\mathscr F}^2$ enters the likelihood.} We maintain conditional independence of these residuals across ranks, so the joint model is
\[
\mathbf x\mid\mathbf w_{1:N}\sim\N\left(
\bigl(z_1,\ldots,z_N\bigr),
\sigma_{\mathscr F}^2 I_N\right).
\]
This joint conditional model determines how much probability the analyst assigns to each stopping-consistent region.

Because the user does not observe $\mathbf w_i$ or $z_i$, \eqref{eq:score-feature} and \eqref{eq:x-feature} do not enter her decision rule. Her belief, learning rule, and stopping policy remain indexed by policy primitives $(m_0, v_0, c, \{\alpha_i\}, \sigma_\eta)$. The parameter $\beta$ is therefore not an argument of her decision rule. Effectively, the user holds a Gaussian belief about page relevance that marginalizes over the features she cannot see and behaves as if $x_i \mid \mu \sim \N(\mu + \alpha_i, \sigma_\eta^2)$ with prior $\mu \sim \N(m_0, v_0)$. The analyst, by contrast, uses features to make sharper predictions about $x_i$ than the user can.

\paragraph{Observed Outcomes as Relevance Regions.} Two latent observables describe how the session ends: an inspection depth $t \in \{1, \ldots, N\}$ and a terminal conversion choice $J \in \{0, 1, \ldots, t\}$, where $J = 0$ is the outside option and $J = j \ge 1$ is the rank-$j$ item. We assume the inspection depth is always logged. Whether conversion is logged depends on the platform: some sites record a booking, a purchase, or a callback that pins down $J$, while others record only that the user clicked or scrolled through to inspect each item, with commitment invisible. We treat both cases below.

Conditional on the policy primitives, $(t, J)$ is a deterministic function of the latent path $(x_1, \ldots, x_t)$: the user inspects ranks in order, stops as soon as the lead $L_t(H_{t-1}, x_t)$ exceeds the reservation threshold $r_t$, and at the stop chooses the maximum of $\{x_b, x_1, \ldots, x_t\}$. Each observable event therefore corresponds to a region of $\R^t$, and the session likelihood is the Gaussian probability of the path lying in that region under \eqref{eq:x-feature}. The region decomposes into three building blocks. The first two apply regardless of whether conversion is observed; the third arises only when it is.

\emph{Block 1: the survival region $\mathcal{C}_t$.} As established earlier, the survival region $\mathcal{C}_t \subset \R^{t - 1}$ collects the histories on which the user has not yet stopped at any of the first $t - 1$ ranks.

\emph{Block 2: the stopping set $\mathcal{S}_t$ at rank $t$.} Define
\[
\mathcal{S}_t(H_{t - 1}) := \begin{cases}
\{x \in \R : L_t(H_{t-1}, x) \ge r_t\}, & t < N, \\
\R, & t = N,
\end{cases}
\]
the set of rank-$t$ realizations that trigger termination. By Proposition~\ref{prop:cont-interval}, $\mathcal{S}_t(H_{t-1}) = (-\infty, \underline{x}_t(H_{t-1})] \cup [\overline{x}_t(H_{t-1}), \infty)$ in the dynamic explore regime, $\R$ in the dynamic trust regime, and $\R$ at the end of the list.

Blocks 1 and 2 together pin down the depth-only event: the surviving history $H_{t-1}$ lies in $\mathcal{C}_t$ and, conditional on that history, the rank-$t$ relevance lies in $\mathcal{S}_t(H_{t-1})$.

\emph{Block 3: the conversion set $\mathcal{R}_{t, j}$ and the restricted survival region $\bar{\mathcal{C}}_{t, j}$ (when conversion is observed).} Set $x_0 := x_b$, the outside-option value; so $J = \argmax_{0 \le i \le t} x_i$ almost surely. Three cases describe $\mathcal{R}_{t, j}$, depending on which item the conversion label points to. If $j = t$, the rank-$t$ item is the realized maximum, so $\mathcal{R}_{t, t}(H_{t-1}) = \mathcal{S}_t(H_{t-1}) \cap [\max\{x_0, x_1, \ldots, x_{t-1}\}, \infty)$, the commitment branch at rank $t$. If $j = 0$, the outside option wins, which requires the surviving history entirely below $x_b$ (an additional restriction on $H_{t-1}$) and $x_t < x_b$ on the rank-$t$ relevance, giving $\mathcal{R}_{t, 0}(H_{t-1}) = \mathcal{S}_t(H_{t-1}) \cap (-\infty, x_b)$, the cut-losses tail in its strongest form. If $j \in \{1, \ldots, t - 1\}$, an earlier inspected item remains best, which requires $x_j > \max\{x_b, x_i : i \le t - 1, i \ne j\}$ inside the surviving history (an additional restriction on $H_{t-1}$) and $x_t < x_j$ on the rank-$t$ relevance, giving $\mathcal{R}_{t, j}(H_{t-1}) = \mathcal{S}_t(H_{t-1}) \cap (-\infty, x_j)$.

The latter two cases attach a constraint to the surviving history itself, so $H_{t-1}$ does not range over all of $\mathcal{C}_t$ when $j \ne t$ but over the \emph{restricted survival region}
\begin{equation}\label{eq:restricted-survival}
\bar{\mathcal{C}}_{t, j} := \begin{cases}
\mathcal{C}_t, & j = t, \\[2pt]
\mathcal{C}_t \cap \bigl\{H_{t-1} : x_i < x_b \text{ for all } 1 \le i \le t-1\bigr\}, & j = 0, \\[2pt]
\mathcal{C}_t \cap \bigl\{H_{t-1} : x_j > \max\{x_b, x_i : i \le t-1, i \ne j\}\bigr\}, & 1 \le j \le t - 1.
\end{cases}
\end{equation}

\paragraph{Session Likelihood.} Denote by $\phi_{\mathscr{F}}(\cdot) := \sigma_{\mathscr{F}}^{-1} \phi(\cdot/\sigma_{\mathscr{F}})$ the centered Gaussian density at scale $\sigma_{\mathscr{F}}$, and collect the rank-$1$ through rank-$t$ feature vectors as $\mathbf w_{1:t}$. The inspected path has joint density $p_{x \mid \mathbf w_{1:t}}(x_1, \ldots, x_t) = \prod_{i=1}^{t} \phi_{\mathscr{F}}(x_i-z_i)$. The session likelihood integrates this density over the region associated with the observed event. Which region enters depends on whether the terminal choice is observed.

\emph{Case A: depth observed, conversion not observed.} The observed event is that the history survives through rank $t-1$ and the rank-$t$ realization falls in its history-dependent stopping set. The session likelihood is
\begin{multline}\label{eq:depth-likelihood}
\ell^{T}(\beta, \sigma_{\mathscr{F}}) := \Pr\bigl(T = t \big| \mathbf{w}_{1:t}; \beta, \sigma_{\mathscr{F}}\bigr) \\
	= \int_{\mathcal{C}_t} \left[\int_{\mathcal{S}_t(H_{t-1})} \phi_{\mathscr{F}}(x_t-z_t)dx_t\right] \prod_{i=1}^{t-1} \phi_{\mathscr{F}}(x_i-z_i)dx_1 \cdots dx_{t-1}.
\end{multline}
The outer integral weights each surviving history $H_{t-1} \in \mathcal{C}_t$ by its joint Gaussian density under the feature-driven model; the inner integral over $\mathcal{S}_t(H_{t-1})$ is a univariate Gaussian probability that, in the dynamic explore regime, has two tails with endpoints $\underline{x}_t(H_{t-1})$ and $\overline{x}_t(H_{t-1})$, and equals one in the dynamic trust regime or at the end of the list $t = N$. This is the appropriate likelihood for platforms that observe inspection but not commitment.

\emph{Case B: depth and conversion both observed.} The observed event is $\{T = t, J = j\}$, whose region refines the Case A region by replacing the stopping set $\mathcal{S}_t$ with the conversion set $\mathcal{R}_{t, j}$ from Block 3 and the survival region $\mathcal{C}_t$ with the restricted survival region $\bar{\mathcal{C}}_{t, j}$ from \eqref{eq:restricted-survival}:
\begin{multline}\label{eq:feature-likelihood}
\ell^{T, J}(\beta, \sigma_{\mathscr{F}}) := \Pr\bigl(T = t, J = j \big| \mathbf{w}_{1:t}; \beta, \sigma_{\mathscr{F}}\bigr) \\
	= \int_{\bar{\mathcal{C}}_{t, j}} \left[\int_{\mathcal{R}_{t, j}(H_{t-1})} \phi_{\mathscr{F}}(x_t-z_t)dx_t\right] \prod_{i=1}^{t-1} \phi_{\mathscr{F}}(x_i-z_i)dx_1 \cdots dx_{t-1}.
\end{multline}

\paragraph{Sequential GHK Evaluation.} The likelihoods above integrate over the unobserved relevance path $(x_1,\ldots,x_t)$. Direct evaluation becomes difficult as $t$ grows because the admissible value of each $x_i$ depends on the relevances that came before it. The key simplification is that these restrictions arrive sequentially, as expressed by Proposition~\ref{prop:cont-interval}.

The GHK simulator follows this sequential structure. At each rank before the observed stopping depth, it computes the Gaussian probability of the admissible interval, draws a relevance from the Gaussian distribution conditional on falling inside that interval, and uses the draw to update the user's belief and running best. At the terminal rank, it computes the Gaussian probability of the observed stopping event. The product of these probabilities is the contribution of a simulated path. Averaging these contributions across paths approximates the integral over the latent histories consistent with the observed session.

Formally, suppose the data contain $S$ iid sessions indexed by $n=1,\ldots,S$. Write $\mathbf{w}_{n,i}$ for the rank-$i$ feature vector and $z_{n,i}=\mathscr F(\beta;\mathbf w_{n,i})$ for its system score. Let $T_n=t_n$ denote the observed inspection depth and let $J_n$ denote the terminal choice when one is logged. A depth-only record contributes $\ell^T$ from \eqref{eq:depth-likelihood}, while a depth-and-choice record contributes $\ell^{T,J}$ from \eqref{eq:feature-likelihood}. When page lengths vary, each session uses the policy with respect to its own page length.

When terminal choice $J_n=j$ is observed, the simulated relevance at each rank must also be consistent with item $j$ eventually winning. For $j\in\{0,1,\ldots,t_n\}$, define
\[
\mathcal W_{i,j}(H_{i-1})=
\begin{cases}
(-\infty,x_b), & j=0,\\
\R, & 1\le j\text{ and }i<j,\\
(\max\{x_b,x_1,\ldots,x_{i-1}\},\infty), & i=j\ge1,\\
(-\infty,x_j), & 1\le j<i.
\end{cases}
\]
These restrictions have a simple interpretation. If the outside option wins, every inspected item must lie below $x_b$. Before the winning item appears, the terminal choice imposes no additional restriction. When the winning item appears, it must exceed the running maximum, and every later item must remain below it. When terminal choice is not observed, there is no winner restriction and we take $\mathcal W_{i,j}=\R$ at every rank.

Now fix a GHK path $r$. At each rank $i<t_n$, intersect the continuation interval with the restriction implied by the terminal choice. Write the resulting admissible interval as $I_{n,i}^{(r)}=(\underline x_{n,i}^{(r)},\overline x_{n,i}^{(r)})$.
Under the feature model, $x_{n,i}$ is Gaussian with mean $z_{n,i}$ and standard deviation $\sigma_{\mathscr F}$. Define the standardized endpoints
\[
A_{n,i}^{(r)}
=
\frac{\underline x_{n,i}^{(r)}-z_{n,i}}{\sigma_{\mathscr F}},
\qquad
B_{n,i}^{(r)}
=
\frac{\overline x_{n,i}^{(r)}-z_{n,i}}{\sigma_{\mathscr F}}.
\]
The conditional probability that $x_{n,i}$ falls inside the admissible interval is $p_{n,i}^{(r)}=\left(\Phi(B_{n,i}^{(r)}) - \Phi(A_{n,i}^{(r)})\right)^+$, where $y^+:=\max\{y,0\}$.\footnote{Although each simulated draw satisfies the restrictions imposed up to its rank, a valid draw at rank $i$ can lead to a dynamic trust state at rank $i+1$, where no realization permits further continuation. The next admissible interval is then empty and the path receives zero weight. An empty interval can also occur if the continuation interval has no overlap with the restriction imposed by the observed terminal choice.}
When $p_{n,i}^{(r)}>0$, draw the latent relevance from the same Gaussian conditional on falling inside the interval. For a uniform draw $u_{n,i}^{(r)}\in(0,1)$, inverse transform sampling gives
$
x_{n,i}^{(r)}
=
z_{n,i}
+
\sigma_{\mathscr F}
\Phi^{-1}\!\left(
\Phi(A_{n,i}^{(r)})
+
u_{n,i}^{(r)}p_{n,i}^{(r)}
\right)
$.
We then use this draw to update $(m_i,M_i)$ and construct the admissible interval at rank $i+1$.

At the terminal rank $t_n$, no further state must be updated. We therefore set $p_{n,t_n}^{(r)}$ equal to the Gaussian mass of $\mathcal S_{t_n}$ when only depth is observed and of $\mathcal S_{t_n}\cap\mathcal W_{t_n,J_n}$ when terminal choice is also observed. If $t_n=N$, then $\mathcal S_N=\R$, so the depth-only terminal factor is one; an observed terminal choice still contributes its winner restriction.

The product of these conditional probabilities is the contribution of path $r$. Averaging across $R$ paths gives the GHK likelihood estimator
\begin{equation*}
\widehat\ell_n
=
\frac{1}{R}
\sum_{r=1}^{R}
\prod_{i=1}^{t_n}
p_{n,i}^{(r)}.
\end{equation*}
This is the sequential simulator studied by \citet{hajivassiliou1996simulation}, specialized to the continuation intervals generated by the standout rule. The recursion is almost everywhere differentiable in the system scores and policy objects, so $\widehat\ell_n$ can be differentiated by automatic differentiation. It advances through each session once per path, giving total cost $O(R\sum_n t_n)$.

\paragraph{Calibrating the User-side Primitives.} The likelihoods in \eqref{eq:depth-likelihood} and \eqref{eq:feature-likelihood} depend on the relevance model and, through the stopping region, on the user-side primitives $\theta_U := (m_0, v_0, \sigma_\eta, x_b, c)$ and the rank-specific shifts $\{\alpha_i\}$. These policy objects are not recorded in a typical production log. The session-level latent $\mu_n$, by contrast, is not an obstacle: the survival and stopping regions are policy objects that depend on the posterior over $\mu_n$ given the inspected history rather than on $\mu_n$ itself, and \eqref{eq:x-feature} replaces $\mu_n$ with a feature-driven score, removing $\mu_n$ as a primitive of the analyst's model. The real question is what to do about $\theta_U$ and $\{\alpha_i\}$. The following four steps reduce them to two scalars that train alongside $\beta$.

\emph{Normalization.} The latent relevance $x$ has no intrinsic origin or unit. The likelihood is unchanged under (i) a joint additive shift $(x_i, x_b, m_0, z_i) \mapsto (x_i + \delta, x_b + \delta, m_0 + \delta, z_i + \delta)$ for any $\delta \in \R$, and (ii) a joint multiplicative rescaling of all relevance-denominated objects by any $\lambda > 0$, with variances rescaling as $\lambda^2$.\footnote{The formal verifications appear in Appendix~\ref{app:invariance}.} Each invariance leaves one degree of freedom to normalize. We set $m_0 = 0$ to pin the origin and $\sigma_{\mathscr{F}} = 1$ to set the unit. Neither choice is restrictive since the parameter space has no canonical origin or unit. After this step, all remaining parameters are read in units of the residual standard deviation.

\emph{Calibration.} The remaining belief objects $(\{\alpha_i\}, v_0, \sigma_\eta)$ are cross-sectional moments of the system scores $\{z_{n,i}\}$ and therefore update with $\beta$.

Start with the rank-specific shifts. By definition, $\alpha_i$ is the expected boost in relevance from being placed at rank $i$, so averaging \eqref{eq:x-feature} across sessions at a fixed slot and centering across slots gives
\begin{equation}\label{eq:alpha-calibration}
\alpha_i = \E_n[z_{n,i}]
- \frac{1}{N}\sum_{k=1}^{N}\E_n[z_{n,k}].
\end{equation}
The right-hand side is centered across slots by construction, and it recovers the model's rank-specific shift because that shift is itself centered, $\sum_{i=1}^{N}\alpha_i = 0$. Two remarks. First, \eqref{eq:alpha-calibration} spares the analyst from recovering $(\sigma_x, \sigma_e)$ separately in order to evaluate $\alpha_i = \rho\sigma_z q_i$, a decomposition not separately identified by the log. Second, it does not impose the rank-to-quantile reduction in Assumption~\ref{assn:quantile}; instead, it estimates each rank-specific shift directly from the cross-session score moments. The maintained traversal policy still requires the rank-specific shifts to decrease with rank. We accordingly restrict estimation to candidate $\beta$ for which $\alpha_1(\beta)>\cdots>\alpha_N(\beta)$.

Given $\{\alpha_i\}$, the analyst can recover the session-level mean as $\hat\mu_n := \sum_{i=1}^{N}(z_{n,i}-\alpha_i)/N$, which by \eqref{eq:alpha-calibration} is simply the session average of the system scores. This average estimates $\mu_n$, but because it uses only the $N$ retrieved items on one page, it also contains finite-page sampling noise. To separate that noise from true cross-session variation in $\mu_n$, we need a model for how individual item scores vary around page quality. We therefore impose the analyst-side analogue of the user's page model on the system scores before ranking: $z_{n,k}=\mu_n+\chi_{n,k}$, where the $\chi_{n,k}$ are iid conditional on $\mu_n$, with $\E[\chi_{n,k}\mid\mu_n]=0$ and $\Var(\chi_{n,k}\mid\mu_n)=:\sigma_\chi^2$. This is an additional feature-model moment condition, not a consequence of Assumption~\ref{assn:gaussian}. Ranking permutes the pool, so the item in slot $i$ has $z_{n,i}=\mu_n+\chi_{n,(i)}$ and $\alpha_i=\E[\chi_{n,(i)}]$. Because a permutation leaves the page average unchanged, $\hat\mu_n=\mu_n+\bar\chi_n$, where $\bar\chi_n:=\sum_{k=1}^N\chi_{n,k}/{N}$.
The term $\bar\chi_n=\hat\mu_n-\mu_n$ is the sampling error in the page average. Conditional independence gives $\E[\bar\chi_n\mid\mu_n]=0$ and $\Var(\bar\chi_n\mid\mu_n)={\sigma_\chi^2}/{N}$,
so the law of total variance yields $\Var_n[\hat\mu_n] =v_0 + {\sigma_\chi^2}/{N}$,
where $v_0$ captures variation in page quality, $\sigma_\chi^2/N$ is the finite-page sampling noise, and $\Var_n$ denotes variance across sessions. The ranked scores also identify $\sigma_\chi^2$. Sorting changes which item occupies each slot but does not change the dispersion of the items on the page. Hence $\E_n[\Var_{i\mid n}(z_{n,i})]=\sigma_\chi^2$, where $\Var_{i\mid n}$ is taken across the $N$ slots of a single session. Substituting this expression for $\sigma_\chi^2$ gives
\begin{equation}\label{eq:v0-calibration}
v_0 = \Var_n\bigl[\hat\mu_n\bigr]
- \frac{1}{N}\,\E_n\bigl[\Var_{i \mid n}(z_{n,i})\bigr].
\end{equation}
The first term contains both true variation in page quality and finite-page sampling noise; the second removes the latter. The correction is $O(1/N)$ and hence vanishes for long lists, but it is not negligible at the lengths of a typical search page.

The user's perceived per-item noise $\sigma_\eta$ likewise matches the analyst's marginal within-session variance of $x_i - \alpha_i$. Decomposing $\Var_{i \mid n}(x_i - \alpha_i)$ into the variation of $\E[x_i \mid \mathbf{w}_i] - \alpha_i$ across $i$ and the residual $\sigma_{\mathscr{F}}^2$ gives
\begin{equation}\label{eq:sigmaeta-calibration}
\sigma_\eta^2 = \E_n\bigl[\Var_{i \mid n}(z_{n,i}-\alpha_i)\bigr] + 1,
\end{equation}
substituting $\sigma_{\mathscr{F}}^2 = 1$ from the normalization above. The two calibration equations use different within-session variances because they serve different purposes. Equation~\eqref{eq:v0-calibration} uses $z_{n,i}$ itself to recover the dispersion of the retrieved pool before ranking, while \eqref{eq:sigmaeta-calibration} subtracts $\alpha_i$ because the user's model already accounts for the rank-specific shifts. Let $\mathcal Q(\beta) := (\{\alpha_i(\beta)\}, v_0(\beta), \sigma_\eta(\beta))$ denote the three objects determined by \eqref{eq:alpha-calibration}--\eqref{eq:sigmaeta-calibration}.

When pages differ in length, we apply the same moment conditions separately for each page length $N$, producing $\mathcal Q_N(\beta)=(\{\alpha_{N,i}(\beta)\}_{i=1}^N,v_{0,N}(\beta),\sigma_{\eta,N}(\beta))$. This allows the calibrated belief moments to vary with page length rather than padding shorter pages and using a common calibration.

\emph{Belief consistency interpretation.} The calibration imposes consistency between the user's beliefs and the analyst's feature-based relevance model. The user need not observe $\beta$, the item features, or the system scores, since her decisions require only a marginal belief about page quality and relevance by rank. The steady-state interpretation assumes that repeated experience has calibrated these beliefs to the distribution of pages she encounters. The analyst describes the same distribution more finely through $z_i=\mathscr F(\beta;\mathbf w_i)$. At the true steady state, both $\beta_0$ and the user's calibrated beliefs $\mathcal Q_0$ are fixed, with $\mathcal Q_0=\mathcal Q(\beta_0)$. Recomputing $\mathcal Q(\beta)$ during estimation means that the optimizer compares candidate steady states; it does not mean that users update their beliefs alongside the training algorithm.

\emph{Joint estimation.} The normalization and calibration determine $\mathcal Q(\beta)$, leaving $(c, x_b)$ to be estimated with $\beta$. The extended training objective is
\begin{equation}\label{eq:extended-objective}
\min_{(\beta,c,x_b)\in\mathcal A} -\sum_{n=1}^{S} \log \ell_n(\beta, \mathcal Q(\beta); c, x_b),
\end{equation}
where the admissible domain $\mathcal A$ requires $c>0$, $v_0(\beta)>0$, $\alpha_1(\beta)>\cdots>\alpha_N(\beta)$, and the interior solution inequality in Assumption~\ref{assn:interior}. With varying page lengths, these conditions apply to every $\mathcal Q_N(\beta)$. For each admissible candidate $\beta$, the likelihood constructs the corresponding survival region from $\mathcal Q(\beta)$. The derivative follows two paths: directly through the system scores $z_{n,i}(\beta)$ and indirectly through the calibrated moments that define the stopping region.

The distribution of observed depth $T_n$ at similar feature profiles is especially informative about $c$: holding $\mathbf{w}_{n,1:t_n}$ fixed, a higher inspection cost shifts probability toward shallower depths. Depth alone also informs $x_b$ through the lower, cut-losses tail, but does not reveal which stops came from that branch. Sessions with $J_n=0$ supply the branch label and therefore sharpen the estimate of $x_b$. User heterogeneity in $(c,x_b)$ can be incorporated through a mixture likelihood.

\section{Experiments}\label{sec:experiments}

In this section, we evaluate whether the likelihood learns a useful relevance score from data. We first study our approach in a simulation and then apply it to the Baidu search logs.\footnote{Replication code is available at \url{https://github.com/shichao-ma/Replication-Bayesian-Rational-Search-Engine-User}.} In simulation, the data-generating process is exactly the analytical model developed above, so our learner should have an advantage over conventional learners that do not exploit this structure. In the Baidu exercise, we ask whether the structural likelihood \eqref{eq:extended-objective} yields a better relevance predictor on a large real-world search log.

In both exercises, GHK evaluation uses scrambled Sobol uniforms with a separate random shift for each session. Training refreshes these randomized points at every update, while held-out evaluation uses larger fixed draw sets and common random numbers across checkpoints.
Throughout the tables, \emph{structural: depth only} uses \eqref{eq:depth-likelihood}, in which only inspection depth is observed. \emph{Structural: depth + choice} uses \eqref{eq:feature-likelihood} when the terminal choice $J$ is observed.

\subsection*{Simulation}

 We generate pages of $N=10$ items. Session $n$ has four query features $\mathbf g_n$, and candidate item $k$ has 30 item features $\mathbf h_{n,k}$; every coordinate is drawn independently from $\N(0,1)$. To generate the true relevance, we fixed three two-layer tanh networks, $\widetilde\mu$, $\widetilde a$, and $\widetilde r$, all of which have 16 units per layer and are standardized to have mean zero and variance one. The true conditional mean relevance $s_{n,k}$ and realized relevance $x_{n,k}$ are
\[
s_{n,k}=\sqrt{v_0}\,\widetilde\mu(\mathbf g_n)+1.2\,\widetilde a(\mathbf h_{n,k}),
\qquad
x_{n,k}=s_{n,k}+\varepsilon_{n,k},\quad \varepsilon_{n,k}\sim\N(0,1).
\]
The incumbent orders the candidates by a noisy score that uses only the first 15 item features,
\[
z^{\mathrm{inc}}_{n,k}=1.2\,\widetilde r(\mathbf h_{n,k,1:15})+\nu_{n,k},
\qquad \nu_{n,k}\sim\N(0,1.2^2),
\]
where the superscript $\mathrm{inc}$ denotes the incumbent score and $\widetilde r$ is generated independently of the function $\widetilde a$ that determines relevance. The ordering is therefore informative but imperfect. All 30 item features enter the relevance experienced by the user, whereas the incumbent can use only the first 15. We interpret the last 15 features as user-side private attributes that co-determine relevance. They are part of the data-generating process but are not observed by the system.

The user follows the myopic standout rule with $m_0=0$, $v_0=1$, $c=0.06$, and $x_b=1.2$.\footnote{The myopic rule greatly simplifies the computation because its reservation gap is available in closed form, whereas the optimal gap requires solving the Bellman equation. Appendix~\ref{app:numerical} shows that the two rules generate nearly identical reservation gaps and inspection-depth distributions across representative configurations.} We approximate the policy's remaining population moments using 40,000 auxiliary item pools and obtain $\sigma_\eta=1.5583$ and $\boldsymbol\alpha=(0.263,0.187,0.130,0.083,0.031,-0.025,-0.073,-0.121,-0.197,-0.278)$.
Setting the learner's scores equal to the true conditional means and applying the sample analogues of \eqref{eq:alpha-calibration}--\eqref{eq:sigmaeta-calibration} to the training sample from the seed-0 replication below gives $v_0=0.993$ against $1$, $\sigma_\eta=1.555$ against $1.558$, and a maximum absolute error of $0.0123$ in the rank-specific shifts. Omitting the $1/N$ correction in \eqref{eq:v0-calibration} gives $v_0=1.138$.
In the same sample, mean inspection depth is $6.33$, the outside option wins in 22.8 percent of sessions, and the item with the highest conditional mean relevance goes uninspected in 26.3 percent. Displayed rank has a correlation of only $-0.105$ with conditional mean relevance. The incumbent ordering is therefore useful enough to shape inspection but leaves substantial room for reranking.

\paragraph{Information available to the learner.} The simulation answers two separate questions. First, can the estimator recover relevance and the user-side primitives when the relevance model is correctly specified? Second, can it improve ranking when the learner sees only the features available to the incumbent system?

To answer the first question, the simulation gives the learner every feature entering conditional mean relevance. This full-information benchmark is infeasible for the system because the last 15 item features representing user-private determinants of relevance are not available to the system. Nevertheless, it provides a correctly specified analyst's model to test structural recovery. We then remove the user-private item features and give every learner, including the oracle and all labeling heuristics, only the item features used by the incumbent ranker. This second exercise is the feasible ranking comparison. In both exercises, we estimate the structural model with and without $J$, isolating the additional information supplied by terminal choice. Display position is excluded in both exercises.

The feasible exercise without terminal choice also includes a heuristically generated click feedback tied directly to realized relevance. Holding the stopping path fixed, we draw clicks independently across inspected items. Conditional on realized relevance and inspection depth,
\[
p^{\mathrm{click}}_{n,i}
=\mathbf 1\{i\le T_n\}
\frac{\exp(x_{n,i})}{\sum_{k=1}^{T_n}\exp(x_{n,k})},
\qquad
Y^{\mathrm{click}}_{n,i}\sim\operatorname{Bernoulli}\!\left(p^{\mathrm{click}}_{n,i}\right).
\]
Thus, a more relevant inspected item is exponentially more likely to be clicked, while an uninspected item cannot be clicked. The independent Bernoulli probabilities sum to one, so the conditional expected number of clicks is one even though a session may have no click or several clicks. Across the five training logs, the mean is $0.9985$ clicks per session; 22.6 percent have no click and 19.2 percent have multiple clicks. These labels provide a strong click benchmark because they respond directly to realized relevance without changing the user's stopping behavior.

\paragraph{Estimation.} We repeat data generation and estimation under five random seeds for each exercise. Each replication contains 30,000 training sessions, 6,000 validation sessions, and 6,000 test sessions. The relevance model has two hidden layers with 48 tanh units each. We estimate its weights jointly with $c$ and $x_b$ using a penalized sample version of \eqref{eq:extended-objective}. The objective penalizes upward steps in the calibrated rank-specific shifts, a nonpositive $v_0$, and violations of the first-inspection condition in Assumption~\ref{assn:interior}.

AdamW updates the scorer in batches of 512, with learning rate $0.005$ and weight decay $10^{-4}$; cosine decay lowers the learning rate to $0.00125$ at epoch 45. After a three-epoch scorer warm-up, Adam updates $(\log c,x_b)$ from the prespecified start $(0.1,0.75)$, with learning rate $0.05$ and cosine decay to $10^{-5}$. Training runs for at most 60 epochs and stops after eight validation epochs without improvement.

The belief moments in \eqref{eq:alpha-calibration}--\eqref{eq:sigmaeta-calibration} change with the scorer. Each 512-page batch therefore provides both the system scores entering its session likelihoods and the raw moments used to calibrate the stopping region. Before evaluating the likelihood, we update an exponential moving average (EMA) of these moments, placing weight $0.25$ on the previous average and $0.75$ on the current batch. The forward pass constructs the stopping region from the smoothed moments, while a straight-through calculation passes the calibration gradient through the current batch moments. At the end of each epoch, we hold the scorer fixed and recompute the calibration moments over the full training sample in batches of 4,096. We use these full-sample moments to construct the stopping region for validation and reset the moving average to the same values before the next epoch. Among checkpoints satisfying the policy restrictions and an EMA-to-full-sample calibration gap below $0.1$, we select the one with the lowest validation loss.

Training uses 64 GHK paths per session and validation uses 256 fixed paths. This is because training requires a forward and backward likelihood pass for every mini-batch, so 64 paths keep the repeated gradient calculations tractable. Validation runs once per epoch and requires no backward pass; 256 paths therefore stabilize early stopping at much lower total cost than increasing the training draw count.

\paragraph{Comparison methods and metrics.} 

We compare the structural likelihood with conventional heuristics of turning session logs into item-level labels. When $J$ is observed, \emph{pointwise conversion} labels the converted item positive and every other displayed item negative, including items the user never inspected. \emph{Conversion + inverse-propensity weighting (IPW)} uses the same labels but weights a positive at rank $i$ by $1/\max\{\widehat{\Pr}(T\ge i),0.05\}$. \emph{Inspected items only} removes items below the observed depth. \emph{Conversion-truncated} keeps sessions with an item conversion, labels the converted item positive and earlier items negative, and discards later items. \emph{Listwise terminal choice} applies a softmax loss over the displayed list when $J\ge1$.

For the feasible exercise without $J$, we add two click-trained comparisons. The \emph{pointwise click classifier} treats every displayed item as a binary-classification observation and averages binary cross-entropy over all items, including unclicked candidates below $T_n$. For the \emph{listwise click softmax}, we drop sessions with no clicks and divide each remaining session's click vector by its number of clicks. We then compute cross-entropy against a softmax over the full displayed list and average equally across sessions. Thus, a session with two clicks places target weight $1/2$ on each clicked item. Both exercises also retain the coarser \emph{inspection label}, which uses $\mathbf 1\{i\le T\}$ as a binary relevance label, and its IPW version. Within an exercise, every method uses the same features, 48--48 tanh architecture, AdamW setup, and its own validation objective. An oracle least-squares model trained on every realized $x_i$ supplies an upper benchmark under the same feature set.

Tables~\ref{tab:simulation-full-information-results} and \ref{tab:simulation-feasible-results} report mean test performance and its Monte Carlo standard error, calculated as the sample standard deviation across the five replications divided by $\sqrt{5}$. The tables group methods by whether $J$ enters training and evaluate how well each learned score ranks true conditional mean relevance $s$.\footnote{This is the relevant target because realized relevance $x$ also contains an independent shock that is unavailable when ranking a new page; evaluating against $x$ would mix ranking performance with irreducible realization noise.} To construct graded relevance, we use the 20th, 40th, 60th, and 80th percentiles of item-level $s$ in the training sample as fixed cutoffs and assign each test item a grade from zero to four. DCG uses gain $2^{\mathrm{grade}}-1$ and discount $1/\log_2(1+\mathrm{rank})$; nDCG divides by the ideal DCG for that page. Top-1 is the fraction of pages on which the learned scorer ranks the item with the highest $s$ first. Top-1 below $T$ computes the same rate only on pages where that item was not inspected; among those pages, it reports how often the learned scorer nevertheless ranks it first.

\begin{table}[htp]
\caption{Full-information benchmark across 5 random seeds. Every learner observes the query features and all 30 item features; the oracle uses realized relevance labels. Rows are grouped by whether terminal choice $J$ is observed in training. Entries are means with Monte Carlo standard errors in parentheses. The target is conditional mean relevance; higher is better.}
\label{tab:simulation-full-information-results}
\centering
\footnotesize
\setlength{\tabcolsep}{3pt}
\renewcommand{\arraystretch}{1.15}
\begin{tabular}{lccccc}
\toprule
Model & nDCG@1 & nDCG@5 & nDCG@10 & Top-1 & Top-1 below $T$ \\
\midrule
\multicolumn{6}{l}{\textit{Panel A: terminal choice $J$ observed in training}} \\
\quad Structural: depth + choice & \begin{tabular}{@{}c@{}}0.9811\\(0.0003)\end{tabular} & \begin{tabular}{@{}c@{}}0.9846\\(0.0003)\end{tabular} & \begin{tabular}{@{}c@{}}0.9897\\(0.0003)\end{tabular} & \begin{tabular}{@{}c@{}}0.8358\\(0.0022)\end{tabular} & \begin{tabular}{@{}c@{}}0.8209\\(0.0031)\end{tabular} \\
\addlinespace[2pt]
\quad Pointwise conversion label & \begin{tabular}{@{}c@{}}0.9742\\(0.0011)\end{tabular} & \begin{tabular}{@{}c@{}}0.9811\\(0.0005)\end{tabular} & \begin{tabular}{@{}c@{}}0.9871\\(0.0004)\end{tabular} & \begin{tabular}{@{}c@{}}0.7873\\(0.0036)\end{tabular} & \begin{tabular}{@{}c@{}}0.7164\\(0.0069)\end{tabular} \\
\addlinespace[2pt]
\quad Conversion + IPW & \begin{tabular}{@{}c@{}}0.9736\\(0.0010)\end{tabular} & \begin{tabular}{@{}c@{}}0.9797\\(0.0004)\end{tabular} & \begin{tabular}{@{}c@{}}0.9866\\(0.0003)\end{tabular} & \begin{tabular}{@{}c@{}}0.7997\\(0.0040)\end{tabular} & \begin{tabular}{@{}c@{}}0.7938\\(0.0061)\end{tabular} \\
\addlinespace[2pt]
\quad Inspected items only & \begin{tabular}{@{}c@{}}0.9779\\(0.0009)\end{tabular} & \begin{tabular}{@{}c@{}}0.9830\\(0.0004)\end{tabular} & \begin{tabular}{@{}c@{}}0.9885\\(0.0004)\end{tabular} & \begin{tabular}{@{}c@{}}0.8186\\(0.0039)\end{tabular} & \begin{tabular}{@{}c@{}}0.8327\\(0.0071)\end{tabular} \\
\addlinespace[2pt]
\quad Conversion-truncated & \begin{tabular}{@{}c@{}}0.9591\\(0.0016)\end{tabular} & \begin{tabular}{@{}c@{}}0.9701\\(0.0007)\end{tabular} & \begin{tabular}{@{}c@{}}0.9805\\(0.0005)\end{tabular} & \begin{tabular}{@{}c@{}}0.7484\\(0.0069)\end{tabular} & \begin{tabular}{@{}c@{}}0.8056\\(0.0078)\end{tabular} \\
\addlinespace[2pt]
\quad Listwise terminal choice & \begin{tabular}{@{}c@{}}0.9752\\(0.0006)\end{tabular} & \begin{tabular}{@{}c@{}}0.9817\\(0.0005)\end{tabular} & \begin{tabular}{@{}c@{}}0.9877\\(0.0004)\end{tabular} & \begin{tabular}{@{}c@{}}0.7935\\(0.0020)\end{tabular} & \begin{tabular}{@{}c@{}}0.7315\\(0.0080)\end{tabular} \\
\addlinespace[4pt]
\multicolumn{6}{l}{\textit{Panel B: terminal choice $J$ not observed in training}} \\
\quad Structural: depth only & \begin{tabular}{@{}c@{}}0.9634\\(0.0032)\end{tabular} & \begin{tabular}{@{}c@{}}0.9710\\(0.0025)\end{tabular} & \begin{tabular}{@{}c@{}}0.9815\\(0.0015)\end{tabular} & \begin{tabular}{@{}c@{}}0.7793\\(0.0081)\end{tabular} & \begin{tabular}{@{}c@{}}0.7516\\(0.0088)\end{tabular} \\
\addlinespace[2pt]
\quad Inspection label & \begin{tabular}{@{}c@{}}0.5044\\(0.0051)\end{tabular} & \begin{tabular}{@{}c@{}}0.6174\\(0.0023)\end{tabular} & \begin{tabular}{@{}c@{}}0.7755\\(0.0017)\end{tabular} & \begin{tabular}{@{}c@{}}0.1491\\(0.0029)\end{tabular} & \begin{tabular}{@{}c@{}}0.0478\\(0.0029)\end{tabular} \\
\addlinespace[2pt]
\quad Inspection + IPW & \begin{tabular}{@{}c@{}}0.5044\\(0.0040)\end{tabular} & \begin{tabular}{@{}c@{}}0.6154\\(0.0033)\end{tabular} & \begin{tabular}{@{}c@{}}0.7750\\(0.0021)\end{tabular} & \begin{tabular}{@{}c@{}}0.1494\\(0.0017)\end{tabular} & \begin{tabular}{@{}c@{}}0.0507\\(0.0026)\end{tabular} \\
\addlinespace[4pt]
\multicolumn{6}{l}{\textit{Panel C: oracle}} \\
\quad Oracle regression & \begin{tabular}{@{}c@{}}0.9928\\(0.0003)\end{tabular} & \begin{tabular}{@{}c@{}}0.9938\\(0.0003)\end{tabular} & \begin{tabular}{@{}c@{}}0.9953\\(0.0003)\end{tabular} & \begin{tabular}{@{}c@{}}0.9088\\(0.0018)\end{tabular} & \begin{tabular}{@{}c@{}}0.9029\\(0.0020)\end{tabular} \\
\bottomrule
\end{tabular}

\end{table}

\begin{table}[htp]
\caption{Feasible learner comparison across 5 random seeds. Every learner observes the query features and the same 15 item features used by the incumbent ranker. In Panel B, the click rows use independent Bernoulli clicks. The oracle uses realized relevance labels but the same features. Entries are means with Monte Carlo standard errors in parentheses. The target is conditional mean relevance; higher is better.}
\label{tab:simulation-feasible-results}
\centering
\footnotesize
\setlength{\tabcolsep}{3pt}
\renewcommand{\arraystretch}{1.15}
\begin{tabular}{lccccc}
\toprule
Model & nDCG@1 & nDCG@5 & nDCG@10 & Top-1 & Top-1 below $T$ \\
\midrule
\multicolumn{6}{l}{\textit{Panel A: terminal choice $J$ observed in training}} \\
\quad Structural: depth + choice & \begin{tabular}{@{}c@{}}0.7269\\(0.0022)\end{tabular} & \begin{tabular}{@{}c@{}}0.7913\\(0.0013)\end{tabular} & \begin{tabular}{@{}c@{}}0.8745\\(0.0010)\end{tabular} & \begin{tabular}{@{}c@{}}0.3442\\(0.0025)\end{tabular} & \begin{tabular}{@{}c@{}}0.3012\\(0.0067)\end{tabular} \\
\addlinespace[2pt]
\quad Pointwise conversion label & \begin{tabular}{@{}c@{}}0.7213\\(0.0025)\end{tabular} & \begin{tabular}{@{}c@{}}0.7871\\(0.0013)\end{tabular} & \begin{tabular}{@{}c@{}}0.8723\\(0.0010)\end{tabular} & \begin{tabular}{@{}c@{}}0.3377\\(0.0030)\end{tabular} & \begin{tabular}{@{}c@{}}0.2627\\(0.0050)\end{tabular} \\
\addlinespace[2pt]
\quad Conversion + IPW & \begin{tabular}{@{}c@{}}0.7205\\(0.0022)\end{tabular} & \begin{tabular}{@{}c@{}}0.7867\\(0.0010)\end{tabular} & \begin{tabular}{@{}c@{}}0.8719\\(0.0009)\end{tabular} & \begin{tabular}{@{}c@{}}0.3398\\(0.0016)\end{tabular} & \begin{tabular}{@{}c@{}}0.3071\\(0.0042)\end{tabular} \\
\addlinespace[2pt]
\quad Inspected items only & \begin{tabular}{@{}c@{}}0.7216\\(0.0028)\end{tabular} & \begin{tabular}{@{}c@{}}0.7883\\(0.0017)\end{tabular} & \begin{tabular}{@{}c@{}}0.8727\\(0.0013)\end{tabular} & \begin{tabular}{@{}c@{}}0.3408\\(0.0032)\end{tabular} & \begin{tabular}{@{}c@{}}0.3179\\(0.0051)\end{tabular} \\
\addlinespace[2pt]
\quad Conversion-truncated & \begin{tabular}{@{}c@{}}0.6792\\(0.0030)\end{tabular} & \begin{tabular}{@{}c@{}}0.7599\\(0.0015)\end{tabular} & \begin{tabular}{@{}c@{}}0.8557\\(0.0012)\end{tabular} & \begin{tabular}{@{}c@{}}0.3016\\(0.0021)\end{tabular} & \begin{tabular}{@{}c@{}}0.3512\\(0.0055)\end{tabular} \\
\addlinespace[2pt]
\quad Listwise terminal choice & \begin{tabular}{@{}c@{}}0.7219\\(0.0013)\end{tabular} & \begin{tabular}{@{}c@{}}0.7874\\(0.0013)\end{tabular} & \begin{tabular}{@{}c@{}}0.8723\\(0.0009)\end{tabular} & \begin{tabular}{@{}c@{}}0.3379\\(0.0024)\end{tabular} & \begin{tabular}{@{}c@{}}0.2593\\(0.0067)\end{tabular} \\
\addlinespace[4pt]
\multicolumn{6}{l}{\textit{Panel B: terminal choice $J$ not observed in training}} \\
\quad Structural: depth only & \begin{tabular}{@{}c@{}}0.7169\\(0.0021)\end{tabular} & \begin{tabular}{@{}c@{}}0.7823\\(0.0017)\end{tabular} & \begin{tabular}{@{}c@{}}0.8696\\(0.0011)\end{tabular} & \begin{tabular}{@{}c@{}}0.3363\\(0.0027)\end{tabular} & \begin{tabular}{@{}c@{}}0.2857\\(0.0063)\end{tabular} \\
\addlinespace[2pt]
\quad Pointwise click classifier & \begin{tabular}{@{}c@{}}0.6940\\(0.0018)\end{tabular} & \begin{tabular}{@{}c@{}}0.7609\\(0.0010)\end{tabular} & \begin{tabular}{@{}c@{}}0.8577\\(0.0008)\end{tabular} & \begin{tabular}{@{}c@{}}0.3101\\(0.0012)\end{tabular} & \begin{tabular}{@{}c@{}}0.2012\\(0.0025)\end{tabular} \\
\addlinespace[2pt]
\quad Listwise click softmax & \begin{tabular}{@{}c@{}}0.7012\\(0.0019)\end{tabular} & \begin{tabular}{@{}c@{}}0.7671\\(0.0024)\end{tabular} & \begin{tabular}{@{}c@{}}0.8610\\(0.0013)\end{tabular} & \begin{tabular}{@{}c@{}}0.3151\\(0.0032)\end{tabular} & \begin{tabular}{@{}c@{}}0.2077\\(0.0067)\end{tabular} \\
\addlinespace[2pt]
\quad Inspection label & \begin{tabular}{@{}c@{}}0.5063\\(0.0041)\end{tabular} & \begin{tabular}{@{}c@{}}0.6176\\(0.0023)\end{tabular} & \begin{tabular}{@{}c@{}}0.7759\\(0.0017)\end{tabular} & \begin{tabular}{@{}c@{}}0.1489\\(0.0014)\end{tabular} & \begin{tabular}{@{}c@{}}0.0471\\(0.0021)\end{tabular} \\
\addlinespace[2pt]
\quad Inspection + IPW & \begin{tabular}{@{}c@{}}0.5061\\(0.0044)\end{tabular} & \begin{tabular}{@{}c@{}}0.6191\\(0.0029)\end{tabular} & \begin{tabular}{@{}c@{}}0.7763\\(0.0020)\end{tabular} & \begin{tabular}{@{}c@{}}0.1487\\(0.0017)\end{tabular} & \begin{tabular}{@{}c@{}}0.0480\\(0.0027)\end{tabular} \\
\addlinespace[4pt]
\multicolumn{6}{l}{\textit{Panel C: oracle}} \\
\quad Oracle regression & \begin{tabular}{@{}c@{}}0.7315\\(0.0023)\end{tabular} & \begin{tabular}{@{}c@{}}0.7940\\(0.0013)\end{tabular} & \begin{tabular}{@{}c@{}}0.8762\\(0.0011)\end{tabular} & \begin{tabular}{@{}c@{}}0.3477\\(0.0031)\end{tabular} & \begin{tabular}{@{}c@{}}0.3059\\(0.0071)\end{tabular} \\
\bottomrule
\end{tabular}

\end{table}

\paragraph{Infeasible full-information benchmark.} The correctly specified depth-plus-choice estimator recovers both the relevance equation and the user's stopping parameters. Figure~\ref{fig:simulation-recovery} follows the first fit from the prespecified start $(c,x_b)=(0.1,0.75)$. At the selected epoch, the training and validation negative log likelihoods are $2.598$ and $2.635$. The estimates $(\hat c,\hat x_b)=(0.0511,1.2852)$ are close to the true values $(0.06,1.2)$. The calibrated beliefs are also close: $\hat v_0=1.0492$ against $1$, $\hat\sigma_\eta=1.4962$ against $1.5583$, and $\max_i|\hat\alpha_i-\alpha_i|=0.0292$. 
\begin{figure}[ht]
\centering
\includegraphics[width=\textwidth]{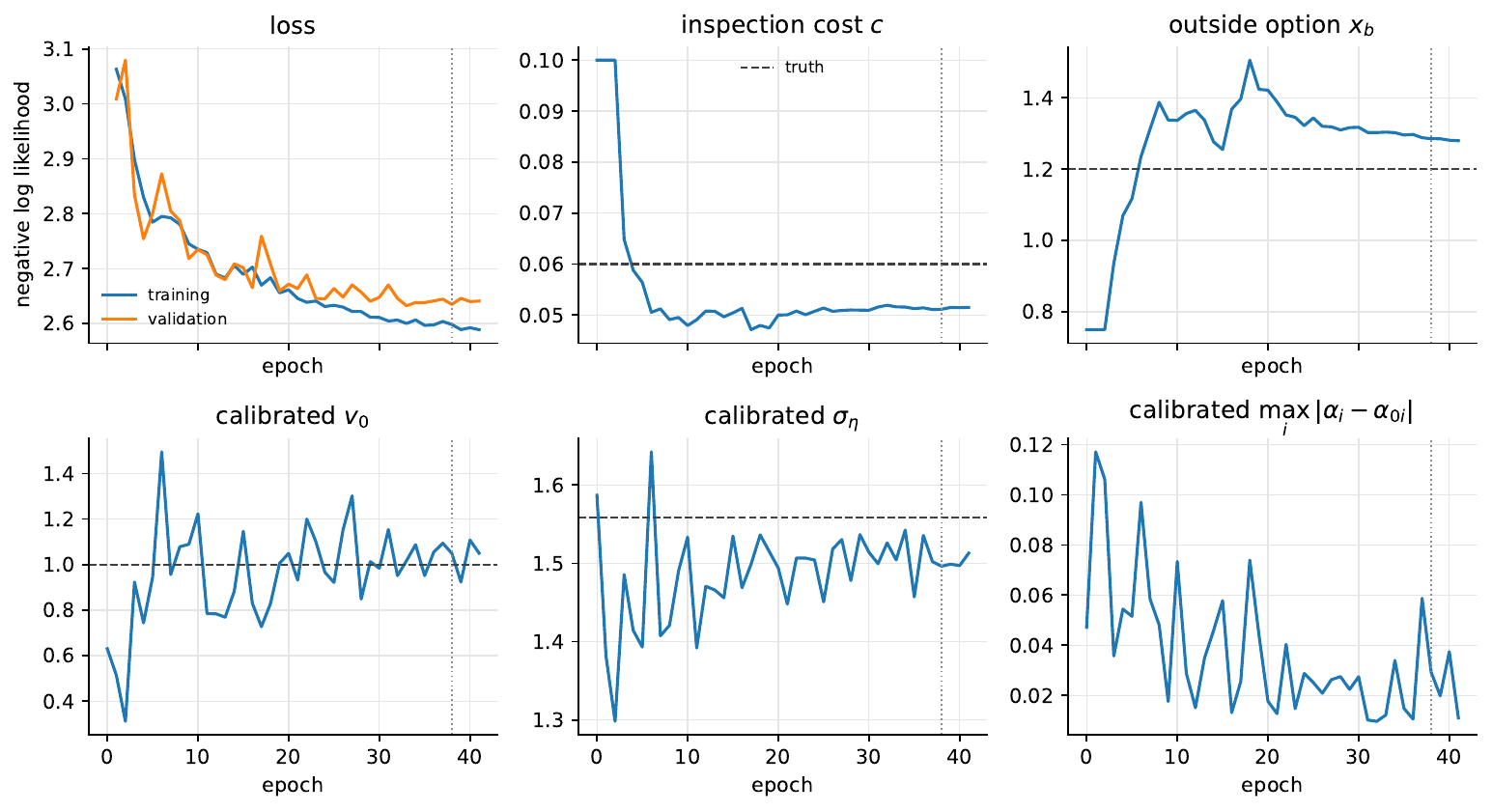}
\caption{Full-information recovery in the first depth-plus-choice replication, initialized at $(c,x_b)=(0.1,0.75)$. Each panel traces one training path; dashed horizontal lines mark the true values, and the dotted vertical line marks the selected admissible epoch. The loss panel begins at epoch 1 because the epoch-0 training loss is much larger than the remaining values. Calibrated moments are evaluated on the full training sample after each epoch.}
\label{fig:simulation-recovery}
\end{figure}

The relevance predictions recover the level as well as the ordering. Figure~\ref{fig:simulation-relevance} sorts the 60,000 test items into 20 equal-sized bins of true conditional mean relevance. Each point compares mean predicted and true relevance within a bin, and the points track the 45-degree line across the support. The shaded region shows the spread of item-level prediction errors within each bin, from the 10th to the 90th percentile.

\begin{figure}[ht]
\centering
\includegraphics[width=0.62\textwidth]{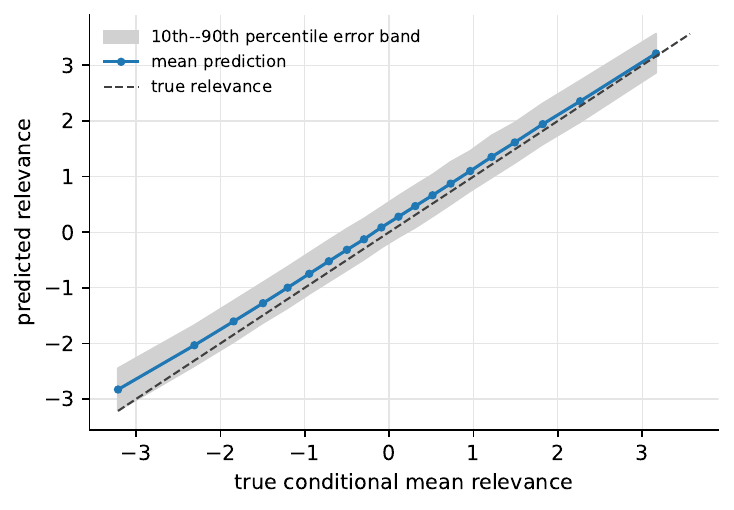}
\caption{Predicted and true conditional mean relevance for 60,000 test items in the first full-information replication. Points compare mean predicted and true relevance within 20 equal-sized bins. The shaded region shows the spread of item-level prediction errors within each bin, from the 10th to the 90th percentile. The dashed line is the 45-degree reference.}
\label{fig:simulation-relevance}
\end{figure}

Structural recovery also translates into better ranking. Table~\ref{tab:simulation-full-information-results} reports the comparison across five replications. With $J$ observed, \textit{structural: depth + choice} leads every conventional labeling method on nDCG@1, nDCG@5, nDCG@10, and overall top-1 recovery. 

Depth remains informative when $J$ is hidden. In this case the structural learner uses no item-level outcome: its only behavioral signal is the inspection depth recorded in the scroll trace. Even with this coarser information, \textit{structural: depth only} leads the inspection-label methods on every reported ranking measure. Thus, under correct specification, the likelihood recovers the structural objects and outperforms the conventional labeling rules on every overall ranking measure, with or without terminal choice.

\paragraph{Feasible learner.} We next remove the 15 user-private item features. Every method observes only $(\mathbf g_n,\mathbf h_{n,k,1:15})$; the oracle has the same inputs but retains the advantage of realized-relevance labels. In the panel without $J$, the structural estimator trains only on $T$, while the click learners train on $\mathbf Y^{\mathrm{click}}$. Table~\ref{tab:simulation-feasible-results} reports the ranking comparison.

With $J$ observed, \textit{structural: depth + choice} again leads every labeling heuristic on nDCG@1, nDCG@5, nDCG@10, and overall top-1 recovery. It also nearly matches the same-feature oracle, showing that depth and terminal choice recover most of the ranking information available in these features without explicit relevance labels. Without $J$, \textit{structural: depth only} leads every non-oracle comparison on all five reported measures, including top-1 recovery when the best item was not inspected.

The gain comes from exploiting the joint information revealed by the user's inspection and stopping behavior. Binary-classification and softmax losses do not use the restrictions on latent relevance required for the user to survive to $T$. The structural likelihood, by contrast, assigns probability to those restrictions. Together, these results show that the stopping event contains ranking information that conventional labeling rules discard.

Although the structural learner outperforms conventional labeling, omitting the 15 user-private item features misspecifies the relevance model and biases the estimated user primitives. Across the five replications, the mean estimates are $\hat c=0.0436$ and $\hat x_b=0.8515$, rather than the true values $0.06$ and $1.2$. The direction and mechanism of this bias merit further study.

\subsection*{Baidu Search Logs}

The simulation establishes what the likelihood \eqref{eq:extended-objective} learns when the analytical model generates user behavior. Whether the method remains useful when real users need not precisely follow that model is an empirical question. The Baidu exercise takes this question to real data. We apply the method to a public search log and ask whether the resulting ordering agrees with separate expert judgments.

\paragraph{Data and Empirical Mapping.} We use the public Baidu search data introduced by \citet{zou2022large} and the reranking release constructed by \citet{hager2024unbiased}. The release contains real search sessions, display information, clicks, query-document features, and a separate set of expert relevance judgments. 

For this exercise, we estimate the depth-only model. We treat a result as inspected when its released \texttt{displayed\_time} is positive, and we do not use clicks as observations of terminal choice. The reason is that Baidu is a general-purpose search engine rather than a platform organized around a single transaction. Users may be looking for a quick fact, researching a topic, or comparing products and services. What constitutes a conversion therefore varies across queries, and much of the relevant downstream behavior occurs outside the search engine's view. Thus, a click is not a reliable substitute. Users often need to open a result before they can learn whether it meets their needs, and sessions in the release can contain multiple clicks. The release also lacks item-level timestamps that would order clicks and exposures. More fundamentally, even the final click need not identify the user's terminal choice. As the analytical model shows, after inspecting several results, the user may return to an earlier item or stop because no inspected item beats her outside option. We therefore estimate the structural likelihood from inspection depth alone, and do not attempt to predict clicks.

The scorer uses 781 query-document features: the 768-dimensional Hager et al.'s BERT query-document embedding, nine lexical retrieval features, and four features recording the query, document, title, and abstract lengths; position is excluded per the construction of the likelihood. The four length variables are released as raw counts. In our preprocessing, we apply $\log(1+x)$ to them and standardize every feature. 

The reranking release omits some entries from the original SERP, so the retained positions can contain gaps. In our main specification, we sort the retained results by released position and relabel them consecutively. All ranks and page lengths below refer to this relabeled list. We define inspection depth $T$ as the number of consecutively displayed results from the top of the list. We exclude sessions with no displayed result and those in which the displayed results are not consecutive from rank~1.

Page length $N$ varies across sessions, and we retain sessions with $N\in\{5,\ldots,10\}$. When $T<N$, the user stops before reaching the end of the page. When $T=N$, the session reaches the end and is treated as censored.
We assign query groups deterministically, with 90 percent in the training sample and 10 percent in the validation sample used for early stopping. After these restrictions, the two samples contain 1,410,747 and 154,511 sessions.

\paragraph{Estimation.} For consistency, we follow the same estimation setup as in the simulation. Because page length varies across sessions, we calibrate the belief moments separately for each $N$ using $\mathcal Q_N(\beta)$ from Section~\ref{sec:partial-id} and form page-length-stratified batches of approximately 512 sessions. Relative to the simulation, the scorer learning rate is $10^{-4}$, with cosine decay to $2.5\times10^{-5}$, and the primitive parameters start from $(c,x_b)=(0.1,-1)$. We use 16 GHK paths per training session and 64 fixed paths for validation. Training runs for at most 42 epochs and stops after six validation epochs without improvement.

We fit the model with five random seeds. Within each seed, we retain the checkpoint with the lowest validation negative log likelihood; we make no selection across seeds and report all five at equal weight.

\paragraph{Assessment and Comparisons.} The assessor file contains 397,522 released rows in 6,985 query groups; Hager et al.'s preprocessing identifies 382,038 unique query-document text pairs. Expert grades run from zero to four. Following Hager et al., we weight queries equally and report DCG at ranks 1, 3, 5, and 10 and MRR@10. For MRR, we follow their released implementation and treat grades of at least one as relevant.

Our main conventional benchmark is a pointwise classifier trained on the same feature set, training sample, and validation sample as the structural model. It labels every unclicked candidate negative and selects its checkpoint by validation loss. We also evaluate a reproduction of Hager et al.'s pointwise click model on their published training and validation cohorts, together with the released BM25 score. We fit both learned benchmarks with the same five random seeds as the structural model.

For each learned model and seed, we compute each metric separately for each query and average equally across queries. We then average the resulting means equally across the five seeds. In Panels A and C of Table~\ref{tab:baidu-results}, parentheses report sample standard deviations across seeds for every learned model. BM25 has no training seeds and is shown without parentheses. Panel B transcribes the means and standard deviations reported by Hager et al.

For paired comparisons, we retain the query-level metrics computed above and average each learned model's five values within query. We subtract the comparator's query-level mean from the structural model's query-level mean and average these differences equally across the 6,985 queries; BM25 enters once through its released score. Table~\ref{tab:baidu-paired} reports the paired standard error of each mean difference. The 95 percent intervals cited below use 10,000 bootstrap replicates that resample query groups with replacement and apply the same resampled query indices to both models. 

\begin{table}[htp]
\caption{Ranking performance on the Baidu assessor sample. Panel A reports our estimates and replications. In Panels A and C, parentheses report sample standard deviations across five network seeds for every learned model; BM25 has no training seeds and is shown without parentheses. Panel B transcribes the comparable values reported in Tables~2 and 3 of \citet{hager2024unbiased}; its parentheses are their reported standard deviations across five runs. Panel C reports the page-length-five ablation and a sensitivity analysis in which both models use sessions with contiguous released positions.}
\label{tab:baidu-results}
\centering
\footnotesize
\setlength{\tabcolsep}{3pt}
{
\begin{tabular}{lccccc}
\toprule
Model & DCG@1 & DCG@3 & DCG@5 & DCG@10 & MRR@10 \\
\midrule
\multicolumn{6}{l}{\textit{Panel A: Our results and replications}} \\
\quad Structural: variable page length & \begin{tabular}{@{}c@{}}2.1147\\(0.0273)\end{tabular} & \begin{tabular}{@{}c@{}}4.4459\\(0.0431)\end{tabular} & \begin{tabular}{@{}c@{}}6.0753\\(0.0473)\end{tabular} & \begin{tabular}{@{}c@{}}9.1012\\(0.0854)\end{tabular} & \begin{tabular}{@{}c@{}}0.6723\\(0.0051)\end{tabular} \\
\addlinespace[2pt]
\quad Matched sample click classifier & \begin{tabular}{@{}c@{}}1.7209\\(0.0201)\end{tabular} & \begin{tabular}{@{}c@{}}3.6296\\(0.0258)\end{tabular} & \begin{tabular}{@{}c@{}}4.9797\\(0.0361)\end{tabular} & \begin{tabular}{@{}c@{}}7.5975\\(0.0516)\end{tabular} & \begin{tabular}{@{}c@{}}0.6181\\(0.0022)\end{tabular} \\
\addlinespace[2pt]
\quad Replication of Hager pointwise model & \begin{tabular}{@{}c@{}}1.7108\\(0.0103)\end{tabular} & \begin{tabular}{@{}c@{}}3.5934\\(0.0139)\end{tabular} & \begin{tabular}{@{}c@{}}4.9310\\(0.0247)\end{tabular} & \begin{tabular}{@{}c@{}}7.5146\\(0.0378)\end{tabular} & \begin{tabular}{@{}c@{}}0.6206\\(0.0024)\end{tabular} \\
\addlinespace[2pt]
\quad BM25 & 2.2107 & 4.6557 & 6.3765 & 9.5445 & 0.7162 \\
\addlinespace[4pt]
\multicolumn{6}{l}{\textit{Panel B: Results reported by Hager et al.}} \\
\quad Hager et al.'s BERT (fixed), pointwise naive & \begin{tabular}{@{}c@{}}1.705\\(0.013)\end{tabular} & \begin{tabular}{@{}c@{}}3.602\\(0.022)\end{tabular} & \begin{tabular}{@{}c@{}}4.944\\(0.023)\end{tabular} & \begin{tabular}{@{}c@{}}7.546\\(0.031)\end{tabular} & \begin{tabular}{@{}c@{}}0.619\\(0.002)\end{tabular} \\
\addlinespace[2pt]
\quad Hager et al.'s BERT (fixed), listwise naive & \begin{tabular}{@{}c@{}}1.798\\(0.008)\end{tabular} & \begin{tabular}{@{}c@{}}3.768\\(0.020)\end{tabular} & \begin{tabular}{@{}c@{}}5.167\\(0.029)\end{tabular} & \begin{tabular}{@{}c@{}}7.844\\(0.046)\end{tabular} & \begin{tabular}{@{}c@{}}0.631\\(0.001)\end{tabular} \\
\addlinespace[2pt]
\quad Hager et al.'s BERT (fixed), LambdaRank naive & \begin{tabular}{@{}c@{}}1.931\\(0.016)\end{tabular} & \begin{tabular}{@{}c@{}}3.993\\(0.029)\end{tabular} & \begin{tabular}{@{}c@{}}5.452\\(0.040)\end{tabular} & \begin{tabular}{@{}c@{}}8.233\\(0.055)\end{tabular} & \begin{tabular}{@{}c@{}}0.648\\(0.003)\end{tabular} \\
\addlinespace[2pt]
\quad Hager et al.'s BERT (fixed), LambdaRank PairD & \begin{tabular}{@{}c@{}}1.932\\(0.007)\end{tabular} & \begin{tabular}{@{}c@{}}3.985\\(0.019)\end{tabular} & \begin{tabular}{@{}c@{}}5.438\\(0.022)\end{tabular} & \begin{tabular}{@{}c@{}}8.216\\(0.038)\end{tabular} & \begin{tabular}{@{}c@{}}0.646\\(0.003)\end{tabular} \\
\addlinespace[2pt]
\quad Hager et al.'s BERT cross-encoder, listwise naive & --- & --- & --- & 8.478 & --- \\
\addlinespace[2pt]
\quad Hager et al.'s BM25 & 2.211 & 4.656 & 6.377 & 9.544 & 0.716 \\
\addlinespace[4pt]
\multicolumn{6}{l}{\textit{Panel C: Ablation and sensitivity}} \\
\quad Structural: page length five & \begin{tabular}{@{}c@{}}2.0528\\(0.0242)\end{tabular} & \begin{tabular}{@{}c@{}}4.2957\\(0.0643)\end{tabular} & \begin{tabular}{@{}c@{}}5.8731\\(0.0883)\end{tabular} & \begin{tabular}{@{}c@{}}8.8552\\(0.1259)\end{tabular} & \begin{tabular}{@{}c@{}}0.6538\\(0.0043)\end{tabular} \\
\addlinespace[2pt]
\quad Structural: contiguous released positions & \begin{tabular}{@{}c@{}}2.1197\\(0.0194)\end{tabular} & \begin{tabular}{@{}c@{}}4.5027\\(0.0276)\end{tabular} & \begin{tabular}{@{}c@{}}6.1716\\(0.0427)\end{tabular} & \begin{tabular}{@{}c@{}}9.2671\\(0.0651)\end{tabular} & \begin{tabular}{@{}c@{}}0.6706\\(0.0055)\end{tabular} \\
\addlinespace[2pt]
\quad Matched contiguous-sample click classifier & \begin{tabular}{@{}c@{}}1.7554\\(0.0131)\end{tabular} & \begin{tabular}{@{}c@{}}3.6960\\(0.0191)\end{tabular} & \begin{tabular}{@{}c@{}}5.0671\\(0.0201)\end{tabular} & \begin{tabular}{@{}c@{}}7.7234\\(0.0340)\end{tabular} & \begin{tabular}{@{}c@{}}0.6209\\(0.0028)\end{tabular} \\
\bottomrule
\end{tabular}
}

\end{table}

\begin{table}[ht]
\caption{Paired differences for the main variable-page-length structural specification. Each cell reports the average query-level difference between the structural model and the row comparator, with its paired standard error across the 6,985 assessor queries in parentheses. For learned comparators, both models are averaged across their five seeds within query before differences are formed; BM25 enters once through its released score. Positive values favor the structural model.}
\label{tab:baidu-paired}
\centering
\footnotesize
\setlength{\tabcolsep}{3pt}
{
\begin{tabular}{lccccc}
\toprule
Comparison against & DCG@1 & DCG@3 & DCG@5 & DCG@10 & MRR@10 \\
\midrule
Matched sample click classifier & \begin{tabular}{@{}c@{}}0.3938\\(0.0263)\end{tabular} & \begin{tabular}{@{}c@{}}0.8163\\(0.0419)\end{tabular} & \begin{tabular}{@{}c@{}}1.0956\\(0.0505)\end{tabular} & \begin{tabular}{@{}c@{}}1.5037\\(0.0629)\end{tabular} & \begin{tabular}{@{}c@{}}0.0542\\(0.0038)\end{tabular} \\
\addlinespace[2pt]
Page length five & \begin{tabular}{@{}c@{}}0.0619\\(0.0139)\end{tabular} & \begin{tabular}{@{}c@{}}0.1502\\(0.0187)\end{tabular} & \begin{tabular}{@{}c@{}}0.2022\\(0.0208)\end{tabular} & \begin{tabular}{@{}c@{}}0.2460\\(0.0231)\end{tabular} & \begin{tabular}{@{}c@{}}0.0185\\(0.0019)\end{tabular} \\
\addlinespace[2pt]
Replication of Hager pointwise model & \begin{tabular}{@{}c@{}}0.4039\\(0.0270)\end{tabular} & \begin{tabular}{@{}c@{}}0.8525\\(0.0425)\end{tabular} & \begin{tabular}{@{}c@{}}1.1443\\(0.0511)\end{tabular} & \begin{tabular}{@{}c@{}}1.5867\\(0.0635)\end{tabular} & \begin{tabular}{@{}c@{}}0.0517\\(0.0039)\end{tabular} \\
\addlinespace[2pt]
BM25 & \begin{tabular}{@{}c@{}}-0.0961\\(0.0251)\end{tabular} & \begin{tabular}{@{}c@{}}-0.2097\\(0.0354)\end{tabular} & \begin{tabular}{@{}c@{}}-0.3012\\(0.0405)\end{tabular} & \begin{tabular}{@{}c@{}}-0.4432\\(0.0473)\end{tabular} & \begin{tabular}{@{}c@{}}-0.0439\\(0.0036)\end{tabular} \\
\bottomrule
\end{tabular}
}

\end{table}

\paragraph{Main Findings.} The main comparison asks what inspection depth contributes when the feature representation and behavioral sample are held fixed. The structural model uses the same 781 features and sessions as the matched-sample click classifier, but learns from inspection depth rather than treating unclicked results as negative labels. It leads the classifier on every reported metric, and the query-bootstrap intervals exclude zero at every cutoff. For example, for DCG@10, the 95 percent interval is $[1.3809,1.6296]$. Because the two models use the same features and training queries, this advantage comes from the information in inspection and stopping behavior.

The comparison with Hager et al.\ points in the same direction. Our pointwise reproduction closely matches their published pointwise results, confirming that the evaluation pipelines are aligned. The structural model also exceeds every click-trained reranker they report: its DCG@10 is 7.4 percent above their strongest result, a BERT cross-encoder trained with the listwise naive objective. Together, these results show that inspection depth contains ranking information beyond what these click-trained models recover. 

All comparisons in this exercise use expert grades as the evaluation target. These grades measure an assessor's judgment of query-document relevance, which need not coincide with the user's private notion of relevance that governed her decisions. Because the user's private relevance is unobserved, the expert grades provide the best available external benchmark for evaluating the rankings learned from behavior. BM25's advantage means that its ordering aligns more closely with that benchmark, not necessarily with the private relevance that generated each session.

\paragraph{Ablation and Sensitivity Analyses.} The \emph{page length five} ablation holds the training sample and estimator constant but censors every session at rank~5. This coarser specification exceeds every click-trained result reported by Hager et al.\ The variable-page-length model then improves on the fixed-five specification on every measure for every matched training seed. Thus, the first five positions already contain enough stopping information to outperform the click-trained benchmarks, while inspection at ranks 6 through 10 provides additional relevance information.

The \emph{contiguous released positions} sensitivity analysis asks whether relabeling retained results consecutively drives the main comparison. It restricts the sample to sessions whose retained candidates occupy the literal contiguous positions $1,\ldots,N$, leaving 619,319 training and 67,959 validation sessions. Its ranking performance remains close to the main specification. Thus, the result does not depend on closing the gaps in the reranking release.

\section{Summary}

We study a Bayesian user who inspects a noisily ranked list. Linear traversal is optimal, and the stopping decision follows a standout rule: the user continues until her best find exceeds her posterior mean of page quality by a depth-dependent threshold. This policy reduces the state to a one-dimensional lead chain whose closed-form transition yields the full distribution of inspection depth. It also separates observationally similar stops into trust, commit, and cut-losses regimes, explains why a more informative ranker can concentrate inspections at rank~$1$, and predicts a nonmonotone relationship between realized relevance and continuation.

The same stopping rule turns inspection depth into a relevance-learning objective. Survival confines the latent relevance history to a polyhedron; an uncensored stop adds a terminal condition, and an observed choice identifies the converted winner. A differentiable sequential GHK simulator evaluates the resulting Gaussian session likelihood in time linear in depth. Normalization and moment calibration leave the relevance model, inspection cost, and outside option to be estimated from depth and, when available, terminal choice.

The experiments show that inspection depth contains useful ranking information. Under correct specification, the structural estimator recovers conditional mean relevance and the user primitives; with only the incumbent's item features, it continues to outperform conventional click-based learners. Depth alone remains informative when terminal choice is unavailable, both in simulation and in the Baidu logs.

The natural empirical extensions are to allow heterogeneity in $(c,x_b)$ and to apply the depth-and-choice likelihood in settings where terminal choices are observed. On the theory side, allowing unknown signal variances or carrying the finite-pool order statistics directly into the state would relax the conjugate and quantile reductions. Endogenizing the system's item selection would then connect the user's stopping policy to the ranker's design problem.


\appendix
\numberwithin{equation}{section}
\numberwithin{figure}{section}
\numberwithin{table}{section}

\section{Numerical Validation: Optimal versus Myopic Reservation Gaps}\label{app:numerical}

Proposition~\ref{prop:myopic} gives the myopic reservation gap $\kappa_t$ in closed form and proves it is strictly decreasing, with explicit endpoints $\kappa_0$ and $\kappa_\infty$. The optimal gap $\kappa_t^\star$ of Proposition~\ref{prop:reservation} has no closed form. This appendix documents numerically that $\kappa_t^\star$ inherits the same shape and that the closed-form $\kappa_t$ is an accurate proxy for it, so the option-value intuition read off the myopic rule transfers to the optimal policy. Lemma~\ref{lem:myopic-lb} already establishes analytically that $\kappa_t \le \kappa_t^\star$ with equality at $t = N-1$; the numerical calculations below quantify how small the difference is and confirm that $\kappa_t^\star$ is itself strictly decreasing.

\paragraph{Method.} By Steps~S8 and~S9 in the proof of Proposition~\ref{prop:reservation}, the value function is translation equivariant: $V_t(M, m) = m + V_t(M - m, 0)$. Writing $W_t(L) := V_t(L, 0)$, the value function is expressed in the lead, with $V_t(M, m) = m + W_t(L)$ and terminal condition $W_N(L) = L$. We obtain $\{\kappa_t^\star\}$ by backward induction on $L$, iterating three steps from $t = N-1$ down to $0$:
\begin{enumerate}
\item \emph{Continuation value.} Form $C_t(L, 0)$ from \eqref{eq:Cdef}, the expected payoff of inspecting rank $t+1$ once more: $C_t(L, 0) = -c + \E\,W_{t+1}(L')$, where $L'$ is the resulting lead under the predictive draw $x_{t+1} \sim \N(\alpha_{t+1},\, v_t + \sigma_\eta^2)$.
\item \emph{Value update.} Take the better of stopping and continuing, $W_t(L) = \max\{L,\, C_t(L, 0)\}$.
\item \emph{Threshold.} The stop-continue gap $\Delta_t(L, 0) = C_t(L, 0) - L$ is strictly decreasing in $L$ (Step~S4) and vanishes at a unique $L = \alpha_{t+1} + \kappa_t^\star$, which defines $\kappa_t^\star$.
\end{enumerate}
Given the thresholds, the stop-epoch distributions come from direct simulation: each session draws a page mean $\mu$, then item relevances $x_i = \mu + \alpha_i + \eta_i$ with $\eta_i \sim \N(0, \sigma_\eta^2)$. The optimal and myopic policies are run on the same draws, so the comparison is paired session by session.

\paragraph{Configurations.} We fix $m_0 = 0$ and $\sigma_x = 1$ (see translation and scale invariance below) and vary the primitives that carry economic content. Table~\ref{tab:numerical-validation} lists four configurations: a baseline (C1), a cheap-search case with deeper sessions (C2), a diffuse-prior long-list case (C3) chosen to maximize the multi-step option value that the myopic rule discards, and a more informative ranker with an attractive outside option (C4).

\begin{table}[h]
\caption{Closeness of the myopic approximation. For each configuration: the maximum absolute and relative gap between $\kappa_t^\star$ and $\kappa_t$ over $t$, the expected inspection depth under each policy, the total-variation distance $d_{\mathrm{TV}}$ between the two stop-epoch distributions, and the fraction of simulated sessions in which the two policies stop at the same epoch.}
\label{tab:numerical-validation}
\centering
\begin{tabular}{llcccccc}
\toprule
& Configuration $(N,\rho,v_0,c,x_b)$ & $\max_t|\kappa_t^\star-\kappa_t|$ & rel. & $\E[T]_{\mathrm{opt}}$ & $\E[T]_{\mathrm{myo}}$ & $d_{\mathrm{TV}}$ & agree \\
\midrule
C1 & $(10,0.50,1,0.1,0)$ & 0.002 & 0.2\% & 2.236 & 2.236 & 0.000 & 100.0\% \\
C2 & $(10,0.50,1,0.03,0)$ & 0.049 & 2.5\% & 3.742 & 3.736 & 0.003 & 99.5\% \\
C3 & $(12,0.40,4,0.1,0)$ & 0.073 & 2.6\% & 2.401 & 2.398 & 0.001 & 99.8\% \\
C4 & $(10,0.80,1,0.1,0.4)$ & 0.000 & 0.0\% & 1.282 & 1.282 & 0.000 & 100.0\% \\
\bottomrule
\end{tabular}

\end{table}

\paragraph{Findings.} Figure~\ref{fig:kappa-compare} overlays $\kappa_t^\star$ and $\kappa_t$ for each configuration. In every panel $\kappa_t^\star$ is strictly decreasing in $t$ and lies weakly above $\kappa_t$, as Lemma~\ref{lem:myopic-lb} requires, and the two curves meet at $N-1$. Their difference is concentrated at the first epoch, where the one-step rule forgoes the most option value, and is largest under the diffuse prior of C3. Even there it is $0.073$ in absolute terms and $2.6\%$ in relative terms (Table~\ref{tab:numerical-validation}), and at every later epoch the two gaps are visually indistinguishable. 

Because the thresholds nearly coincide, so do the induced behaviors: Figure~\ref{fig:depth-compare} shows nearly overlapping stop-epoch distributions under the two policies. Across the four configurations the total-variation distance between the optimal and myopic depth distributions never exceeds $0.003$, expected depth differs by less than $0.01$, and the two rules select the same stopping epoch in roughly $99.5\%$ or more of simulated sessions; the expected-utility loss from following the closed-form rule stays below $10^{-3}$ in relevance units throughout. The closed-form myopic gap is therefore an accurate and analytically transparent stand-in for the optimal one, and the qualitative reading of Proposition~\ref{prop:myopic}, namely the option-value decomposition, the monotone decline in $t$, and the irreducible floor $\kappa_\infty$, carries over to the optimal policy.

\begin{figure}[ht]
\centering
\begin{tikzpicture}[x=1cm,y=1cm,>=stealth,line cap=round,font=\footnotesize]
\draw[black!50] (1.050,5.050) rectangle (6.050,8.250);
\draw[black!10] (1.050,5.679) -- (6.050,5.679);
\draw[black!50] (1.050,5.679) -- (0.980,5.679) node[left,scale=0.6,black]{0.60};
\draw[black!10] (1.050,6.399) -- (6.050,6.399);
\draw[black!50] (1.050,6.399) -- (0.980,6.399) node[left,scale=0.6,black]{0.80};
\draw[black!10] (1.050,7.119) -- (6.050,7.119);
\draw[black!50] (1.050,7.119) -- (0.980,7.119) node[left,scale=0.6,black]{1.00};
\draw[black!10] (1.050,7.840) -- (6.050,7.840);
\draw[black!50] (1.050,7.840) -- (0.980,7.840) node[left,scale=0.6,black]{1.20};
\draw[black!50] (1.050,5.050) -- (1.050,4.980) node[below,scale=0.6,black]{0};
\draw[black!50] (2.717,5.050) -- (2.717,4.980) node[below,scale=0.6,black]{3};
\draw[black!50] (4.383,5.050) -- (4.383,4.980) node[below,scale=0.6,black]{6};
\draw[black!50] (6.050,5.050) -- (6.050,4.980) node[below,scale=0.6,black]{9};
\draw[black!45,dotted,thick] (1.050,5.317) -- (6.050,5.317);
\node[anchor=west,scale=0.58,black!55] at (1.130,5.447) {$\kappa_\infty$};
\draw[blue!55!black,thick] plot coordinates {(1.050,7.983) (1.606,6.320) (2.161,5.939) (2.717,5.768) (3.272,5.671) (3.828,5.608) (4.383,5.564) (4.939,5.532) (5.494,5.507) (6.050,5.487)};
\fill[blue!55!black] (1.050,7.983) circle (1.0pt);
\fill[blue!55!black] (1.606,6.320) circle (1.0pt);
\fill[blue!55!black] (2.161,5.939) circle (1.0pt);
\fill[blue!55!black] (2.717,5.768) circle (1.0pt);
\fill[blue!55!black] (3.272,5.671) circle (1.0pt);
\fill[blue!55!black] (3.828,5.608) circle (1.0pt);
\fill[blue!55!black] (4.383,5.564) circle (1.0pt);
\fill[blue!55!black] (4.939,5.532) circle (1.0pt);
\fill[blue!55!black] (5.494,5.507) circle (1.0pt);
\fill[blue!55!black] (6.050,5.487) circle (1.0pt);
\draw[red!75!black,semithick,dashed] plot coordinates {(1.050,7.975) (1.606,6.320) (2.161,5.939) (2.717,5.768) (3.272,5.671) (3.828,5.608) (4.383,5.564) (4.939,5.532) (5.494,5.507) (6.050,5.487)};
\node[anchor=south,scale=0.7] at (3.550,8.280) {C1: $N{=}10,\rho{=}0.50,v_0{=}1,c{=}0.1,x_b{=}0$};
\node[anchor=north,scale=0.64] at (3.550,4.650) {epoch $t$};
\node[anchor=south,rotate=90,scale=0.64] at (0.470,6.650) {reservation gap};
\draw[blue!55!black,thick] (4.000,8.030) -- (4.420,8.030);
\node[anchor=west,scale=0.58] at (4.480,8.030) {$\kappa_t^\star$ (optimal)};
\draw[red!70!black,thick,densely dashed] (4.000,7.710) -- (4.420,7.710);
\node[anchor=west,scale=0.58] at (4.480,7.710) {$\kappa_t$ (myopic)};
\draw[black!50] (8.450,5.050) rectangle (13.450,8.250);
\draw[black!10] (8.450,5.464) -- (13.450,5.464);
\draw[black!50] (8.450,5.464) -- (8.380,5.464) node[left,scale=0.6,black]{1.00};
\draw[black!10] (8.450,6.106) -- (13.450,6.106);
\draw[black!50] (8.450,6.106) -- (8.380,6.106) node[left,scale=0.6,black]{1.25};
\draw[black!10] (8.450,6.748) -- (13.450,6.748);
\draw[black!50] (8.450,6.748) -- (8.380,6.748) node[left,scale=0.6,black]{1.50};
\draw[black!10] (8.450,7.390) -- (13.450,7.390);
\draw[black!50] (8.450,7.390) -- (8.380,7.390) node[left,scale=0.6,black]{1.75};
\draw[black!50] (8.450,5.050) -- (8.450,4.980) node[below,scale=0.6,black]{0};
\draw[black!50] (10.117,5.050) -- (10.117,4.980) node[below,scale=0.6,black]{3};
\draw[black!50] (11.783,5.050) -- (11.783,4.980) node[below,scale=0.6,black]{6};
\draw[black!50] (13.450,5.050) -- (13.450,4.980) node[below,scale=0.6,black]{9};
\draw[black!45,dotted,thick] (8.450,5.317) -- (13.450,5.317);
\node[anchor=west,scale=0.58,black!55] at (8.530,5.447) {$\kappa_\infty$};
\draw[blue!55!black,thick] plot coordinates {(8.450,7.983) (9.006,6.328) (9.561,5.938) (10.117,5.763) (10.672,5.665) (11.228,5.604) (11.783,5.561) (12.339,5.529) (12.894,5.505) (13.450,5.485)};
\fill[blue!55!black] (8.450,7.983) circle (1.0pt);
\fill[blue!55!black] (9.006,6.328) circle (1.0pt);
\fill[blue!55!black] (9.561,5.938) circle (1.0pt);
\fill[blue!55!black] (10.117,5.763) circle (1.0pt);
\fill[blue!55!black] (10.672,5.665) circle (1.0pt);
\fill[blue!55!black] (11.228,5.604) circle (1.0pt);
\fill[blue!55!black] (11.783,5.561) circle (1.0pt);
\fill[blue!55!black] (12.339,5.529) circle (1.0pt);
\fill[blue!55!black] (12.894,5.505) circle (1.0pt);
\fill[blue!55!black] (13.450,5.485) circle (1.0pt);
\draw[red!75!black,semithick,dashed] plot coordinates {(8.450,7.858) (9.006,6.294) (9.561,5.926) (10.117,5.760) (10.672,5.665) (11.228,5.604) (11.783,5.561) (12.339,5.529) (12.894,5.504) (13.450,5.485)};
\node[anchor=south,scale=0.7] at (10.950,8.280) {C2: $N{=}10,\rho{=}0.50,v_0{=}1,c{=}0.03,x_b{=}0$};
\node[anchor=north,scale=0.64] at (10.950,4.650) {epoch $t$};
\node[anchor=south,rotate=90,scale=0.64] at (7.870,6.650) {reservation gap};
\draw[black!50] (1.050,0.750) rectangle (6.050,3.950);
\draw[black!10] (1.050,1.505) -- (6.050,1.505);
\draw[black!50] (1.050,1.505) -- (0.980,1.505) node[left,scale=0.6,black]{1.00};
\draw[black!10] (1.050,2.098) -- (6.050,2.098);
\draw[black!50] (1.050,2.098) -- (0.980,2.098) node[left,scale=0.6,black]{1.50};
\draw[black!10] (1.050,2.691) -- (6.050,2.691);
\draw[black!50] (1.050,2.691) -- (0.980,2.691) node[left,scale=0.6,black]{2.00};
\draw[black!10] (1.050,3.284) -- (6.050,3.284);
\draw[black!50] (1.050,3.284) -- (0.980,3.284) node[left,scale=0.6,black]{2.50};
\draw[black!50] (1.050,0.750) -- (1.050,0.680) node[below,scale=0.6,black]{0};
\draw[black!50] (2.414,0.750) -- (2.414,0.680) node[below,scale=0.6,black]{3};
\draw[black!50] (4.232,0.750) -- (4.232,0.680) node[below,scale=0.6,black]{7};
\draw[black!50] (6.050,0.750) -- (6.050,0.680) node[below,scale=0.6,black]{11};
\draw[black!45,dotted,thick] (1.050,1.017) -- (6.050,1.017);
\node[anchor=west,scale=0.58,black!55] at (1.130,1.147) {$\kappa_\infty$};
\draw[blue!55!black,thick] plot coordinates {(1.050,3.683) (1.505,1.499) (1.959,1.285) (2.414,1.202) (2.868,1.159) (3.323,1.132) (3.777,1.114) (4.232,1.101) (4.686,1.091) (5.141,1.083) (5.595,1.076) (6.050,1.071)};
\fill[blue!55!black] (1.050,3.683) circle (1.0pt);
\fill[blue!55!black] (1.505,1.499) circle (1.0pt);
\fill[blue!55!black] (1.959,1.285) circle (1.0pt);
\fill[blue!55!black] (2.414,1.202) circle (1.0pt);
\fill[blue!55!black] (2.868,1.159) circle (1.0pt);
\fill[blue!55!black] (3.323,1.132) circle (1.0pt);
\fill[blue!55!black] (3.777,1.114) circle (1.0pt);
\fill[blue!55!black] (4.232,1.101) circle (1.0pt);
\fill[blue!55!black] (4.686,1.091) circle (1.0pt);
\fill[blue!55!black] (5.141,1.083) circle (1.0pt);
\fill[blue!55!black] (5.595,1.076) circle (1.0pt);
\fill[blue!55!black] (6.050,1.071) circle (1.0pt);
\draw[red!75!black,semithick,dashed] plot coordinates {(1.050,3.597) (1.505,1.492) (1.959,1.283) (2.414,1.202) (2.868,1.159) (3.323,1.132) (3.777,1.114) (4.232,1.101) (4.686,1.090) (5.141,1.083) (5.595,1.076) (6.050,1.071)};
\node[anchor=south,scale=0.7] at (3.550,3.980) {C3: $N{=}12,\rho{=}0.40,v_0{=}4,c{=}0.1,x_b{=}0$};
\node[anchor=north,scale=0.64] at (3.550,0.350) {epoch $t$};
\node[anchor=south,rotate=90,scale=0.64] at (0.470,2.350) {reservation gap};
\draw[black!50] (8.450,0.750) rectangle (13.450,3.950);
\draw[black!10] (8.450,1.055) -- (13.450,1.055);
\draw[black!50] (8.450,1.055) -- (8.380,1.055) node[left,scale=0.6,black]{0.20};
\draw[black!10] (8.450,1.679) -- (13.450,1.679);
\draw[black!50] (8.450,1.679) -- (8.380,1.679) node[left,scale=0.6,black]{0.40};
\draw[black!10] (8.450,2.303) -- (13.450,2.303);
\draw[black!50] (8.450,2.303) -- (8.380,2.303) node[left,scale=0.6,black]{0.60};
\draw[black!10] (8.450,2.927) -- (13.450,2.927);
\draw[black!50] (8.450,2.927) -- (8.380,2.927) node[left,scale=0.6,black]{0.80};
\draw[black!10] (8.450,3.551) -- (13.450,3.551);
\draw[black!50] (8.450,3.551) -- (8.380,3.551) node[left,scale=0.6,black]{1.00};
\draw[black!50] (8.450,0.750) -- (8.450,0.680) node[below,scale=0.6,black]{0};
\draw[black!50] (10.117,0.750) -- (10.117,0.680) node[below,scale=0.6,black]{3};
\draw[black!50] (11.783,0.750) -- (11.783,0.680) node[below,scale=0.6,black]{6};
\draw[black!50] (13.450,0.750) -- (13.450,0.680) node[below,scale=0.6,black]{9};
\draw[black!45,dotted,thick] (8.450,1.017) -- (13.450,1.017);
\node[anchor=west,scale=0.58,black!55] at (8.530,1.147) {$\kappa_\infty$};
\draw[blue!55!black,thick] plot coordinates {(8.450,3.683) (9.006,1.589) (9.561,1.342) (10.117,1.244) (10.672,1.191) (11.228,1.159) (11.783,1.136) (12.339,1.120) (12.894,1.108) (13.450,1.098)};
\fill[blue!55!black] (8.450,3.683) circle (1.0pt);
\fill[blue!55!black] (9.006,1.589) circle (1.0pt);
\fill[blue!55!black] (9.561,1.342) circle (1.0pt);
\fill[blue!55!black] (10.117,1.244) circle (1.0pt);
\fill[blue!55!black] (10.672,1.191) circle (1.0pt);
\fill[blue!55!black] (11.228,1.159) circle (1.0pt);
\fill[blue!55!black] (11.783,1.136) circle (1.0pt);
\fill[blue!55!black] (12.339,1.120) circle (1.0pt);
\fill[blue!55!black] (12.894,1.108) circle (1.0pt);
\fill[blue!55!black] (13.450,1.098) circle (1.0pt);
\draw[red!75!black,semithick,dashed] plot coordinates {(8.450,3.683) (9.006,1.589) (9.561,1.341) (10.117,1.243) (10.672,1.191) (11.228,1.158) (11.783,1.136) (12.339,1.120) (12.894,1.107) (13.450,1.098)};
\node[anchor=south,scale=0.7] at (10.950,3.980) {C4: $N{=}10,\rho{=}0.80,v_0{=}1,c{=}0.1,x_b{=}0.4$};
\node[anchor=north,scale=0.64] at (10.950,0.350) {epoch $t$};
\node[anchor=south,rotate=90,scale=0.64] at (7.870,2.350) {reservation gap};
\end{tikzpicture}
\caption{Optimal reservation gap $\kappa_t^\star$ (solid, with markers) and closed-form myopic gap $\kappa_t$ (dashed) as functions of epoch $t$, for the four configurations of Table~\ref{tab:numerical-validation}. The dotted line marks the no-learning floor $\kappa_\infty$. In every case $\kappa_t^\star$ is strictly decreasing and lies weakly above $\kappa_t$, with the two coinciding at $N-1$.}
\label{fig:kappa-compare}
\end{figure}

\begin{figure}[ht]
\centering
\begin{tikzpicture}[x=1cm,y=1cm,>=stealth,line cap=round,font=\footnotesize]
\draw[black!50] (1.050,5.050) rectangle (6.050,8.250);
\draw[black!10] (1.050,5.050) -- (6.050,5.050);
\draw[black!50] (1.050,5.050) -- (0.980,5.050) node[left,scale=0.6,black]{0.00};
\draw[black!10] (1.050,6.166) -- (6.050,6.166);
\draw[black!50] (1.050,6.166) -- (0.980,6.166) node[left,scale=0.6,black]{0.20};
\draw[black!10] (1.050,7.282) -- (6.050,7.282);
\draw[black!50] (1.050,7.282) -- (0.980,7.282) node[left,scale=0.6,black]{0.40};
\draw[black!50] (1.534,5.050) -- (1.534,4.980) node[below,scale=0.58,black]{1};
\fill[blue!28,draw=blue!55!black,line width=0.4pt] (1.211,5.050) rectangle (1.534,5.956);
\fill[red!20,draw=red!70!black,line width=0.4pt] (1.534,5.050) rectangle (1.856,5.956);
\draw[black!50] (2.340,5.050) -- (2.340,4.980) node[below,scale=0.58,black]{2};
\fill[blue!28,draw=blue!55!black,line width=0.4pt] (2.018,5.050) rectangle (2.340,7.857);
\fill[red!20,draw=red!70!black,line width=0.4pt] (2.340,5.050) rectangle (2.663,7.857);
\draw[black!50] (3.147,5.050) -- (3.147,4.980) node[below,scale=0.58,black]{3};
\fill[blue!28,draw=blue!55!black,line width=0.4pt] (2.824,5.050) rectangle (3.147,6.578);
\fill[red!20,draw=red!70!black,line width=0.4pt] (3.147,5.050) rectangle (3.469,6.578);
\draw[black!50] (3.953,5.050) -- (3.953,4.980) node[below,scale=0.58,black]{4};
\fill[blue!28,draw=blue!55!black,line width=0.4pt] (3.631,5.050) rectangle (3.953,5.370);
\fill[red!20,draw=red!70!black,line width=0.4pt] (3.953,5.050) rectangle (4.276,5.370);
\draw[black!50] (4.760,5.050) -- (4.760,4.980) node[below,scale=0.58,black]{5};
\fill[blue!28,draw=blue!55!black,line width=0.4pt] (4.437,5.050) rectangle (4.760,5.068);
\fill[red!20,draw=red!70!black,line width=0.4pt] (4.760,5.050) rectangle (5.082,5.068);
\draw[black!50] (5.566,5.050) -- (5.566,4.980) node[below,scale=0.58,black]{6};
\fill[blue!28,draw=blue!55!black,line width=0.4pt] (5.244,5.050) rectangle (5.566,5.050);
\fill[red!20,draw=red!70!black,line width=0.4pt] (5.566,5.050) rectangle (5.889,5.050);
\node[anchor=south,scale=0.7] at (3.550,8.280) {C1};
\node[anchor=north,scale=0.64] at (3.550,4.650) {inspection depth $T$};
\node[anchor=south,rotate=90,scale=0.64] at (0.470,6.650) {$\Pr(T=t)$};
\fill[blue!28,draw=blue!55!black] (4.100,7.990) rectangle (4.360,8.170);
\node[anchor=west,scale=0.58] at (4.420,8.070) {optimal};
\fill[red!20,draw=red!70!black] (4.100,7.670) rectangle (4.360,7.850);
\node[anchor=west,scale=0.58] at (4.420,7.760) {myopic};
\draw[black!50] (8.450,5.050) rectangle (13.450,8.250);
\draw[black!10] (8.450,5.050) -- (13.450,5.050);
\draw[black!50] (8.450,5.050) -- (8.380,5.050) node[left,scale=0.6,black]{0.00};
\draw[black!10] (8.450,6.057) -- (13.450,6.057);
\draw[black!50] (8.450,6.057) -- (8.380,6.057) node[left,scale=0.6,black]{0.10};
\draw[black!10] (8.450,7.064) -- (13.450,7.064);
\draw[black!50] (8.450,7.064) -- (8.380,7.064) node[left,scale=0.6,black]{0.20};
\draw[black!50] (8.867,5.050) -- (8.867,4.980) node[below,scale=0.58,black]{1};
\fill[blue!28,draw=blue!55!black,line width=0.4pt] (8.589,5.050) rectangle (8.867,5.183);
\fill[red!20,draw=red!70!black,line width=0.4pt] (8.867,5.050) rectangle (9.144,5.192);
\draw[black!50] (9.561,5.050) -- (9.561,4.980) node[below,scale=0.58,black]{2};
\fill[blue!28,draw=blue!55!black,line width=0.4pt] (9.283,5.050) rectangle (9.561,6.783);
\fill[red!20,draw=red!70!black,line width=0.4pt] (9.561,5.050) rectangle (9.839,6.804);
\draw[black!50] (10.256,5.050) -- (10.256,4.980) node[below,scale=0.58,black]{3};
\fill[blue!28,draw=blue!55!black,line width=0.4pt] (9.978,5.050) rectangle (10.256,7.857);
\fill[red!20,draw=red!70!black,line width=0.4pt] (10.256,5.050) rectangle (10.533,7.843);
\draw[black!50] (10.950,5.050) -- (10.950,4.980) node[below,scale=0.58,black]{4};
\fill[blue!28,draw=blue!55!black,line width=0.4pt] (10.672,5.050) rectangle (10.950,7.659);
\fill[red!20,draw=red!70!black,line width=0.4pt] (10.950,5.050) rectangle (11.228,7.646);
\draw[black!50] (11.644,5.050) -- (11.644,4.980) node[below,scale=0.58,black]{5};
\fill[blue!28,draw=blue!55!black,line width=0.4pt] (11.367,5.050) rectangle (11.644,6.798);
\fill[red!20,draw=red!70!black,line width=0.4pt] (11.644,5.050) rectangle (11.922,6.796);
\draw[black!50] (12.339,5.050) -- (12.339,4.980) node[below,scale=0.58,black]{6};
\fill[blue!28,draw=blue!55!black,line width=0.4pt] (12.061,5.050) rectangle (12.339,5.864);
\fill[red!20,draw=red!70!black,line width=0.4pt] (12.339,5.050) rectangle (12.617,5.863);
\draw[black!50] (13.033,5.050) -- (13.033,4.980) node[below,scale=0.58,black]{7};
\fill[blue!28,draw=blue!55!black,line width=0.4pt] (12.756,5.050) rectangle (13.033,5.258);
\fill[red!20,draw=red!70!black,line width=0.4pt] (13.033,5.050) rectangle (13.311,5.258);
\node[anchor=south,scale=0.7] at (10.950,8.280) {C2};
\node[anchor=north,scale=0.64] at (10.950,4.650) {inspection depth $T$};
\node[anchor=south,rotate=90,scale=0.64] at (7.870,6.650) {$\Pr(T=t)$};
\draw[black!50] (1.050,0.750) rectangle (6.050,3.950);
\draw[black!10] (1.050,0.750) -- (6.050,0.750);
\draw[black!50] (1.050,0.750) -- (0.980,0.750) node[left,scale=0.6,black]{0.00};
\draw[black!10] (1.050,1.473) -- (6.050,1.473);
\draw[black!50] (1.050,1.473) -- (0.980,1.473) node[left,scale=0.6,black]{0.10};
\draw[black!10] (1.050,2.195) -- (6.050,2.195);
\draw[black!50] (1.050,2.195) -- (0.980,2.195) node[left,scale=0.6,black]{0.20};
\draw[black!10] (1.050,2.918) -- (6.050,2.918);
\draw[black!50] (1.050,2.918) -- (0.980,2.918) node[left,scale=0.6,black]{0.30};
\draw[black!50] (1.534,0.750) -- (1.534,0.680) node[below,scale=0.58,black]{1};
\fill[blue!28,draw=blue!55!black,line width=0.4pt] (1.211,0.750) rectangle (1.534,2.150);
\fill[red!20,draw=red!70!black,line width=0.4pt] (1.534,0.750) rectangle (1.856,2.159);
\draw[black!50] (2.340,0.750) -- (2.340,0.680) node[below,scale=0.58,black]{2};
\fill[blue!28,draw=blue!55!black,line width=0.4pt] (2.018,0.750) rectangle (2.340,3.557);
\fill[red!20,draw=red!70!black,line width=0.4pt] (2.340,0.750) rectangle (2.663,3.557);
\draw[black!50] (3.147,0.750) -- (3.147,0.680) node[below,scale=0.58,black]{3};
\fill[blue!28,draw=blue!55!black,line width=0.4pt] (2.824,0.750) rectangle (3.147,2.753);
\fill[red!20,draw=red!70!black,line width=0.4pt] (3.147,0.750) rectangle (3.469,2.746);
\draw[black!50] (3.953,0.750) -- (3.953,0.680) node[below,scale=0.58,black]{4};
\fill[blue!28,draw=blue!55!black,line width=0.4pt] (3.631,0.750) rectangle (3.953,1.537);
\fill[red!20,draw=red!70!black,line width=0.4pt] (3.953,0.750) rectangle (4.276,1.536);
\draw[black!50] (4.760,0.750) -- (4.760,0.680) node[below,scale=0.58,black]{5};
\fill[blue!28,draw=blue!55!black,line width=0.4pt] (4.437,0.750) rectangle (4.760,0.948);
\fill[red!20,draw=red!70!black,line width=0.4pt] (4.760,0.750) rectangle (5.082,0.948);
\draw[black!50] (5.566,0.750) -- (5.566,0.680) node[below,scale=0.58,black]{6};
\fill[blue!28,draw=blue!55!black,line width=0.4pt] (5.244,0.750) rectangle (5.566,0.779);
\fill[red!20,draw=red!70!black,line width=0.4pt] (5.566,0.750) rectangle (5.889,0.779);
\node[anchor=south,scale=0.7] at (3.550,3.980) {C3};
\node[anchor=north,scale=0.64] at (3.550,0.350) {inspection depth $T$};
\node[anchor=south,rotate=90,scale=0.64] at (0.470,2.350) {$\Pr(T=t)$};
\draw[black!50] (8.450,0.750) rectangle (13.450,3.950);
\draw[black!10] (8.450,0.750) -- (13.450,0.750);
\draw[black!50] (8.450,0.750) -- (8.380,0.750) node[left,scale=0.6,black]{0.00};
\draw[black!10] (8.450,1.532) -- (13.450,1.532);
\draw[black!50] (8.450,1.532) -- (8.380,1.532) node[left,scale=0.6,black]{0.20};
\draw[black!10] (8.450,2.314) -- (13.450,2.314);
\draw[black!50] (8.450,2.314) -- (8.380,2.314) node[left,scale=0.6,black]{0.40};
\draw[black!10] (8.450,3.097) -- (13.450,3.097);
\draw[black!50] (8.450,3.097) -- (8.380,3.097) node[left,scale=0.6,black]{0.60};
\draw[black!50] (8.934,0.750) -- (8.934,0.680) node[below,scale=0.58,black]{1};
\fill[blue!28,draw=blue!55!black,line width=0.4pt] (8.611,0.750) rectangle (8.934,3.556);
\fill[red!20,draw=red!70!black,line width=0.4pt] (8.934,0.750) rectangle (9.256,3.557);
\draw[black!50] (9.740,0.750) -- (9.740,0.680) node[below,scale=0.58,black]{2};
\fill[blue!28,draw=blue!55!black,line width=0.4pt] (9.418,0.750) rectangle (9.740,1.855);
\fill[red!20,draw=red!70!black,line width=0.4pt] (9.740,0.750) rectangle (10.063,1.854);
\draw[black!50] (10.547,0.750) -- (10.547,0.680) node[below,scale=0.58,black]{3};
\fill[blue!28,draw=blue!55!black,line width=0.4pt] (10.224,0.750) rectangle (10.547,0.750);
\fill[red!20,draw=red!70!black,line width=0.4pt] (10.547,0.750) rectangle (10.869,0.750);
\draw[black!50] (11.353,0.750) -- (11.353,0.680) node[below,scale=0.58,black]{4};
\fill[blue!28,draw=blue!55!black,line width=0.4pt] (11.031,0.750) rectangle (11.353,0.750);
\fill[red!20,draw=red!70!black,line width=0.4pt] (11.353,0.750) rectangle (11.676,0.750);
\draw[black!50] (12.160,0.750) -- (12.160,0.680) node[below,scale=0.58,black]{5};
\fill[blue!28,draw=blue!55!black,line width=0.4pt] (11.837,0.750) rectangle (12.160,0.750);
\fill[red!20,draw=red!70!black,line width=0.4pt] (12.160,0.750) rectangle (12.482,0.750);
\draw[black!50] (12.966,0.750) -- (12.966,0.680) node[below,scale=0.58,black]{6};
\fill[blue!28,draw=blue!55!black,line width=0.4pt] (12.644,0.750) rectangle (12.966,0.750);
\fill[red!20,draw=red!70!black,line width=0.4pt] (12.966,0.750) rectangle (13.289,0.750);
\node[anchor=south,scale=0.7] at (10.950,3.980) {C4};
\node[anchor=north,scale=0.64] at (10.950,0.350) {inspection depth $T$};
\node[anchor=south,rotate=90,scale=0.64] at (7.870,2.350) {$\Pr(T=t)$};
\end{tikzpicture}
\caption{Stop-epoch distributions $\Pr(T = t)$ under the optimal policy (blue) and the myopic policy (red) for the configurations of Table~\ref{tab:numerical-validation}, from common-random-number simulation. The two distributions nearly coincide.}
\label{fig:depth-compare}
\end{figure}

\clearpage

\section{Derivation of the Quantile Reduction}\label{app:quantile}

This section derives Assumption~\ref{assn:quantile} using the probability integral transform.

Let $X$ be a real-valued random variable with continuous, strictly increasing CDF $F$. Then $U := F(X) \sim \text{Uniform}(0,1)$. Conversely, if $U \sim \text{Uniform}(0,1)$, then $X := F^{-1}(U)$ has CDF $F$. The maps $F$ and $F^{-1}$ thus provide a bijection between the original ``data scale'' and the ``probability scale,'' translating any continuous distribution to a uniform and back.

Because $F$ is strictly increasing, applying $F$ preserves the order of a sample: if $X_1, \ldots, X_N$ are iid with CDF $F$ and $U_i := F(X_i)$, then $X_i < X_j \iff U_i < U_j$. In particular, writing $U_{(1)} \le \cdots \le U_{(N)}$ for the ascending order statistics, the $k$-th largest of the $X$'s satisfies
\begin{equation}\label{eq:bridge}
X_{[k]} = F^{-1}\big(U_{(N+1-k)}\big).
\end{equation}
Uniform order statistics have a known finite-sample distribution:
\begin{equation*}
U_{(j)} \sim \text{Beta}(j,N+1-j), \quad \E\left[U_{(j)}\right] = \frac{j}{N+1}, \quad \text{Var}\left[U_{(j)}\right] = \frac{j(N+1-j)}{(N+1)^2(N+2)} = O\left(\tfrac{1}{N}\right).
\end{equation*}
Setting $j = N+1-k$ yields $\E[U_{(N+1-k)}] = 1 - k/(N+1)$ and a vanishing variance, so $U_{(N+1-k)}$ concentrates on its mean as $N$ grows.

Under Assumption~\ref{assn:gaussian}, $z_i \mid \mu \sim \N(\mu, \sigma_z^2)$ with $F^{-1}(p) = \mu + \sigma_z \Phi^{-1}(p)$. Applying the bridge \eqref{eq:bridge},
\begin{equation*}
z_{[k]} = \mu + \sigma_z \Phi^{-1}\big(U_{(N+1-k)}\big).
\end{equation*}
This expression holds at every finite $N$. Replacing the random uniform order statistic by its mean, $U_{(N+1-k)} \approx 1 - k/(N+1)$, yields
\begin{equation*}
z_{[k]} \approx \mu + \sigma_z \Phi^{-1}\left(1 - \tfrac{k}{N+1}\right),
\end{equation*}
which is the content of Assumption~\ref{assn:quantile}. The approximation error is $O_p(N^{-1/2})$ for central ranks (those with $k/N$ bounded away from $0$ and $1$), and therefore vanishes as $N \to \infty$.

\paragraph{Refinements.}
Two standard refinements are available. Blom's formula replaces $k/(N+1)$ with $(k - 3/8)/(N + 1/4)$, correcting systematic bias in the tails. Alternatively, the finite-sample expected normal order statistics $\E[z_{[k]}] = \mu + \sigma_z e_{k,N}$ can be computed numerically. Substituting $e_{k,N}$ for $\Phi^{-1}(1 - k/(N+1))$ preserves every downstream step.

\section{Translation and Scale Invariance of the Likelihood}\label{app:invariance}

We verify that the two likelihood invariances cited in Section~\ref{sec:partial-id} justify the normalizations $m_0 = 0$ and $\sigma_{\mathscr{F}} = 1$.

\paragraph{Translation invariance.} We claim that the map
\[
(x_i, x_b, m_0, z_i) \longmapsto (x_i + \delta, x_b + \delta, m_0 + \delta, z_i + \delta), \qquad \delta \in \R,
\]
leaves the likelihood unchanged. Each inequality defining $\mathcal{C}_t$ has the form $x_j < a_s + \gamma_s \sum_{r=1}^{s} x_r + r_s$, with $a_s = (v_s/v_0) m_0 - \gamma_s \sum_r \alpha_r$, $\gamma_s = v_s/\sigma_\eta^2$. The shift moves the left-hand side by $\delta$ and the right-hand side by $(v_s/v_0)\delta + \gamma_s s \delta = v_s \delta (1/v_0 + s/\sigma_\eta^2) = \delta$, where the last equality uses the variance recursion $1/v_s = 1/v_0 + s/\sigma_\eta^2$. Both sides move by $\delta$, so every inequality is preserved and $\mathcal{C}_t$ is unchanged. The Gaussian density $\phi_{\mathscr{F}}(x_i-z_i)$ is invariant under the joint shift of $x_i$ and $z_i$. Hence the likelihood is unchanged.

\paragraph{Scale invariance.} We claim that for any $\lambda > 0$, the joint rescaling
\[
(x, x_b, m_0, z, c, \alpha, \sigma_{\mathscr{F}}, \sigma_\eta) \longmapsto \lambda \cdot (\text{each}), \qquad v_0 \longmapsto \lambda^2 v_0,
\]
also leaves the likelihood unchanged. The argument has two parts.

\emph{Part 1: Policy and survival region rescale by $\lambda$.} The Bellman equation
\[
V_t(M_t, m_t) = \max\bigl\{M_t, -c + \E[V_{t+1}(\max(M_t, x_{t+1}), m_{t+1})]\bigr\}
\]
maps to itself under the rescaling with $V_t \mapsto \lambda V_t$: scaling all of $M, m, c, x$ by $\lambda$ scales every term in both branches of the max by $\lambda$. By backward induction the reservation gap satisfies $\kappa_s^\star \mapsto \lambda \kappa_s^\star$, so $r_s = \alpha_{s+1} + \kappa_s^\star \mapsto \lambda r_s$. The variance recursion gives $v_s \mapsto \lambda^2 v_s$, so $\gamma_s = v_s/\sigma_\eta^2$ is invariant and $a_s = (v_s/v_0) m_0 - \gamma_s \sum_r \alpha_r \mapsto \lambda a_s$. Each inequality defining $\mathcal{C}_t$ becomes $\lambda x_j < \lambda(a_s + \gamma_s \sum_{r=1}^{s} x_r + r_s)$, which after dividing both sides by $\lambda$ is the same inequality with rescaled variables. Hence $\mathcal{C}_t \mapsto \lambda \mathcal{C}_t$, a dilation, and the same argument gives $\mathcal{S}_t \mapsto \lambda \mathcal{S}_t$.

\emph{Part 2: The Gaussian integral is invariant under the dilation.} Write the likelihood as a $t$-dimensional Gaussian integral over the rescaled region. Substituting $x_i = \lambda y_i$, each differential contributes $dx_i = \lambda dy_i$, for a total Jacobian $\lambda^t$. Each Gaussian density rescales as
\[
\phi_{\lambda \sigma_{\mathscr{F}}}(\lambda u) = \frac{1}{\lambda \sigma_{\mathscr{F}} \sqrt{2\pi}} \exp\Bigl(-\frac{\lambda^2 u^2}{2 \lambda^2 \sigma_{\mathscr{F}}^2}\Bigr) = \frac{1}{\lambda} \phi_{\mathscr{F}}(u),
\]
contributing a factor $\lambda^{-t}$ across the $t$ densities. The Jacobian and the density factor cancel, so the integral equals its pre-rescaling value. Hence the likelihood is unchanged.

\section{Auxiliary Lemmas and Proofs}\label{app:auxiliary}
\begin{auxlemma}[Gaussian Truncation Identity]\label{lem:gaussiantrunc}
	We claim that for any $X \sim \N(\mu, \sigma^2)$ with $\sigma > 0$ and any threshold $a \in \R$,
	\begin{equation}\label{eq:trunc}
	\E[(X - a)^+] = \sigma \phi\left(\tfrac{a - \mu}{\sigma}\right) - (a - \mu) \Phi\left(-\tfrac{a - \mu}{\sigma}\right).
	\end{equation}
\end{auxlemma}

\begin{proof}
	To prove \eqref{eq:trunc}, write $X = \mu + \sigma Z$ with $Z \sim \N(0,1)$, and let $d := (a - \mu)/\sigma$. Then $(X - a)^+ = \sigma (Z - d)^+$, so
	\begin{equation*}
	\E[(X-a)^+] = \sigma \int_d^\infty z \phi(z) dz - \sigma d \int_d^\infty \phi(z) dz.
	\end{equation*}
	For the first integral, note $\tfrac{d}{dz}\phi(z) = -z\phi(z)$, so $\int_d^\infty z \phi(z) dz = \phi(d)$.
	For the second, $\int_d^\infty \phi(z) dz = 1 - \Phi(d) = \Phi(-d)$. Combining and substituting $d$ with $(a - \mu)/\sigma$ give \eqref{eq:trunc}.
\end{proof}

\begin{auxlemma}\label{lem:gproof}
	The function $g$ defined in \eqref{eq:gdef} is a strictly decreasing bijection from $\R$ onto $(0, \infty)$.
\end{auxlemma}
\begin{proof}
	Let $g(d) := \phi(d) - d\Phi(-d)$. Differentiating and using $\phi'(d) = -d\phi(d)$ together with $\tfrac{d}{dd}\Phi(-d) = -\phi(-d) = -\phi(d)$, $g'(d) = -\Phi(-d)$.
	Since $\Phi(-d) \in (0, 1)$ for all $d \in \R$, $g'(d) < 0$, so $g$ is strictly decreasing on $\R$. For the range:
	\begin{itemize}
	\item As $d \to -\infty$: $\phi(d) \to 0$, $\Phi(-d) \to 1$, and $-d \cdot \Phi(-d) \to +\infty$, so $g(d) \to +\infty$.
	\item As $d \to +\infty$: $\phi(d) \to 0$, and by the Mills ratio expansion $\Phi(-d) \sim \phi(d)/d$ for large $d$, so $d \cdot \Phi(-d) \sim \phi(d) \to 0$, giving $g(d) \to 0^+$.
	\end{itemize}
	By continuity and strict monotonicity, $g : \R \to (0, \infty)$ is a bijection and $g^{-1} : (0, \infty) \to \R$ is well defined.
\end{proof}

\begin{auxlemma}[Interior Solution Forces Inspection]\label{lem:tau-positive}
Under Assumption~\ref{assn:interior}, the optimal stopping rule defined by the Bellman equation \eqref{eq:bellman} satisfies $T \ge 1$ with probability one. Moreover, for any fixed $c, x_b, m_0, \alpha_1 \in \R$, condition \eqref{eq:noinspect} holds for all sufficiently large $v_0$.
\end{auxlemma}
\begin{proof}
Take the Bellman equation \eqref{eq:bellman} at $t = 0$ with initial state $(M_0, m_0) = (x_b, m_0)$,
\begin{equation}\label{eq:bellman-t1}
V_0(x_b, m_0) = \max\Big\{ x_b, -c + \E_{x_1}\left[V_1\big(\max(x_b, x_1), m_1(x_1)\big)\right]\Big\},
\end{equation}
with $x_1 \sim \N(m_0 + \alpha_1, v_0 + \sigma_\eta^2)$ by Lemma~\ref{lem:posterior}. The stop option at $t = 1$ gives the trivial lower bound $V_1(M, m) \ge M$ for all $(M, m)$; equivalently, one inspection followed by forced stopping is weakly dominated by optimal continuation. Taking expectations under $x_1$,
\begin{equation*}
\E_{x_1}\left[V_1\big(\max(x_b, x_1), m_1(x_1)\big)\right] \ge \E_{x_1}[\max(x_b, x_1)] = x_b + \E_0\left[(x_1 - x_b)^+\right].
\end{equation*}
Under Assumption~\ref{assn:interior}, $c < \E_{x_1}[(x_1 - x_b)^+]$, so
\begin{equation*}
-c + \E_{x_1}\left[V_1\big(\max(x_b, x_1), m_1(x_1)\big)\right] \ge -c + x_b + \E_{x_1}\left[(x_1 - x_b)^+\right] > x_b.
\end{equation*}
The inspect branch in \eqref{eq:bellman-t1} therefore strictly dominates stopping, and $T = 0$ cannot be optimal. Hence $T \ge 1$ with probability one. The closed-form expression for $\E_{x_1}\left[(x_1 - x_b)^+\right]$ follows from Lemma~\ref{lem:gaussiantrunc} applied to $x_1 \sim \N(m_0 + \alpha_1, v_0 + \sigma_\eta^2)$.

For the second claim, fix $x_b, m_0, \alpha_1 \in \R$ and let $v_0 \to \infty$. We have $d_0 \to 0$, so $\phi(d_0) \to 1/\sqrt{2\pi}$ and $\Phi(-d_0) \to 1/2$. Hence the only unbounded term is $\sqrt{v_0 + \sigma_\eta^2}$, which swamps any finite cost $c$. So \eqref{eq:noinspect} holds for all sufficiently large $v_0$.
\end{proof}

\begin{auxlemma}[Myopic Gap Lower-Bounds the Optimal Gap]\label{lem:myopic-lb}
Fix the page length $N$. For every epoch $t \in \{0, 1, \ldots, N-1\}$, the optimal reservation gap $\kappa_t^\star$ of Proposition~\ref{prop:reservation} and the myopic gap $\kappa_t$ of Proposition~\ref{prop:myopic} satisfy $\kappa_t \le \kappa_t^\star$, with equality at $t = N-1$.
\end{auxlemma}

\begin{proof}
Recall the optimal continuation value $C_t(M, m)$ from \eqref{eq:Cdef} and the stop-continue gap $\Delta_t(M, m) = C_t(M, m) - M$ from \eqref{eq:Deltadef}. Define the \emph{myopic} continuation value as the one obtained by treating epoch $t+1$ as terminal:
\[
C_t^{\mathrm{myo}}(M, m) := -c + \E_{x_{t+1}}\!\left[\max(M, x_{t+1})\right],
\qquad x_{t+1} \sim \N(m + \alpha_{t+1},\, \sigma_t^{\star 2}),
\]
and write $\Delta_t^{\mathrm{myo}}(M, m) := C_t^{\mathrm{myo}}(M, m) - M$. Using $\max(M, x_{t+1}) = M + (x_{t+1} - M)^+$ and Lemma~\ref{lem:gaussiantrunc}, $\E_{x_{t+1}}[(x_{t+1} - M)^+] = \sigma_t^\star g(d_t)$ with $d_t = (M - m - \alpha_{t+1})/\sigma_t^\star$, so $\Delta_t^{\mathrm{myo}}(M, m) = -c + \sigma_t^\star g(d_t)$; its unique zero in $M$ is the myopic reservation $m + \alpha_{t+1} + \kappa_t$ of Proposition~\ref{prop:myopic}.

\emph{Step 1: The myopic continuation value never exceeds the optimal one.} The stop branch of the Bellman equation \eqref{eq:bellman} gives the pointwise lower bound $V_{t+1}(M', m') \ge M'$ for every state $(M', m')$. Evaluating this at $M' = \max(M, x_{t+1})$ and $m' = m_{t+1}(x_{t+1})$ and taking the predictive expectation,
\[
C_t(M, m) = -c + \E_{x_{t+1}}\!\left[V_{t+1}\big(\max(M, x_{t+1}),\, m_{t+1}(x_{t+1})\big)\right]
\ge -c + \E_{x_{t+1}}\!\left[\max(M, x_{t+1})\right] = C_t^{\mathrm{myo}}(M, m).
\]
Subtracting $M$ from both sides, $\Delta_t(M, m) \ge \Delta_t^{\mathrm{myo}}(M, m)$ for every $(M, m)$.

\emph{Step 2: Ordering of the two thresholds.} Fix $m$. Both gaps are strictly decreasing in $M$: for $\Delta_t$ this is Step~S4 of the proof of Proposition~\ref{prop:reservation}; for $\Delta_t^{\mathrm{myo}} = -c + \sigma_t^\star g(d_t)$ it follows because $g$ is strictly decreasing (Lemma~\ref{lem:gproof}) and $d_t$ is strictly increasing in $M$. Let $R_t^\star(m) = m + \alpha_{t+1} + \kappa_t^\star$ be the zero of $\Delta_t(\cdot, m)$ (Step~S7 and Step~S9 of the same proof) and $R_t^{\mathrm{myo}}(m) = m + \alpha_{t+1} + \kappa_t$ the zero of $\Delta_t^{\mathrm{myo}}(\cdot, m)$. Evaluating the Step~1 inequality at $M = R_t^{\mathrm{myo}}(m)$,
\[
\Delta_t\big(R_t^{\mathrm{myo}}(m), m\big) \ge \Delta_t^{\mathrm{myo}}\big(R_t^{\mathrm{myo}}(m), m\big) = 0.
\]
Since $\Delta_t(\cdot, m)$ is strictly decreasing and vanishes at $R_t^\star(m)$, the inequality $\Delta_t(R_t^{\mathrm{myo}}(m), m) \ge 0$ forces $R_t^{\mathrm{myo}}(m) \le R_t^\star(m)$. Cancelling the common term $m + \alpha_{t+1}$ gives $\kappa_t \le \kappa_t^\star$.

\emph{Step 3: Equality at $N-1$.} At $t = N-1$ the terminal condition $V_N(M', m') = M'$ holds with equality, so the lower bound in Step~1 is an identity: $C_{N-1} \equiv C_{N-1}^{\mathrm{myo}}$ and hence $\Delta_{N-1} \equiv \Delta_{N-1}^{\mathrm{myo}}$. Their zeros coincide, and therefore $\kappa_{N-1} = \kappa_{N-1}^\star$.
\end{proof}

\subsection{Proof of Lemma~\ref{lem:predictive}}

Under Assumption~\ref{assn:gaussian}, $x_i \mid \mu \sim \N(\mu, \sigma_x^2)$ and $e_i \sim \N(0, \sigma_e^2)$ independent of $x_i$. Since $z_i = x_i + e_i$,
\begin{equation*}
\begin{pmatrix} x_i \\ z_i \end{pmatrix} \Bigg| \mu \sim \N\left(\begin{pmatrix} \mu \\ \mu \end{pmatrix}, \begin{pmatrix} \sigma_x^2 & \sigma_x^2 \\ \sigma_x^2 & \sigma_z^2 \end{pmatrix}\right),
\end{equation*}
where $\sigma_z^2 = \sigma_x^2 + \sigma_e^2$ and $\Cov(x_i, z_i \mid \mu) = \Var(x_i \mid \mu) = \sigma_x^2$. By the conditional distribution of the bivariate normal distribution,
\begin{align*}
\E[x_i \mid z_i, \mu] &= \mu + \frac{\sigma_x^2}{\sigma_z^2}(z_i - \mu) = \rho z_i + (1-\rho)\mu, \\[4pt]
\Var[x_i \mid z_i, \mu] &= \sigma_x^2\left(1 - \frac{\sigma_x^2}{\sigma_z^2}\right) = \sigma_x^2(1 - \rho) = \sigma_\eta^2,
\end{align*}
so that
\begin{equation}\label{eq:xgivenz}
x_i \mid z_i, \mu \sim \N\big(\rho z_i + (1-\rho)\mu, \sigma_\eta^2\big).
\end{equation}
Assumption~\ref{assn:quantile} fixes $z_i$ deterministically given $\mu$ as $z_i = \mu + \sigma_z q_i$. Substituting into the conditional mean in \eqref{eq:xgivenz}, $\E[x_i \mid \mu] = \mu + \rho\sigma_z q_i = \mu + \alpha_i$. Since $z_i$ is deterministic given $\mu$, the conditional variance is unchanged, $\Var[x_i \mid \mu] = \sigma_\eta^2$. Therefore $x_i \mid \mu \sim \N(\mu + \alpha_i, \sigma_\eta^2)$.

For independence across $i$, by Assumption~\ref{assn:gaussian}, the pairs $\{(x_i, e_i)\}_{i=1}^N$ are iid conditional on $\mu$. Because each constraint $z_i = \mu + \sigma_z q_i$ imposed by Assumption~\ref{assn:quantile} involves only its own slot's pair $(x_i, e_i)$, conditioning on all $N$ constraints reshapes the marginal of each $x_i$ but introduces no dependence across slots; $\{x_i\}_{i=1}^N$ remain conditionally independent given $\mu$. Combined with the Step 3 marginals, $x_i \mid \mu \sim \N(\mu + \alpha_i, \sigma_\eta^2)$ independent across $i$, as claimed. $\hfill\qed$

\subsection{Proof of Lemma~\ref{lem:rankorder}}

The proof has three steps. We first establish how the value changes with the best relevance found so far. We then compare two possible inspection orders while preserving the option to stop after the first item. A final induction applies this pairwise comparison after every history.

\paragraph{Step 1: The value of a better current best.}
Before rank order has been established, the epoch $t$ does not reveal the composition of $\Uunins_t$. We therefore extend the value-function notation from the main text to $V_t(\Uunins_t,M_t,m_t)$. Under a rank-order history, $\Uunins_t=\{t+1,\ldots,N\}$, and suppressing $\Uunins_t$ gives the value function $V_t(M_t,m_t)$ in \eqref{eq:bellman}.

By Lemma~\ref{lem:predictive}, every rank-adjusted relevance $y_i=x_i-\alpha_i$ supplies the same signal about $\mu$. After $t$ inspections, the posterior variance is $v_t$ regardless of which ranks were inspected. Thus,
\[
V_t(\Uunins_t,M_t,m_t)
=
\max\left\{
M_t,
\max_{i\in\Uunins_t}
\left[
-c+\E\,V_{t+1}\!\left(
\Uunins_t\setminus\{i\},
\max\{M_t,\alpha_i+y_i\},
m_{t+1}(y_i)
\right)
\right]
\right\},
\]
where $y_i\sim\N(m_t,v_t+\sigma_\eta^2)$ and
\[
m_{t+1}(y_i)
=
v_{t+1}\left(\frac{m_t}{v_t}+\frac{y_i}{\sigma_\eta^2}\right).
\]
The terminal condition is $V_N(\Uunins_N,M_N,m_N)=M_N$. The finite horizon and Gaussian predictive distributions make all values finite.

For every $d\ge0$,
\[
0\le V_t(\Uunins_t,M_t+d,m_t)-V_t(\Uunins_t,M_t,m_t)\le d.
\]
This bound follows by induction on $|\Uunins_t|$. It holds at the terminal condition. For any inspection, raising $M_t$ by $d$ raises $\max\{M_t,\alpha_i+y_i\}$ by an amount between zero and $d$, without changing the predictive distribution or the posterior update. Taking expectations and then maximizing over stopping and inspection preserves the bound.

\paragraph{Step 2: Comparing two inspection orders.}
Fix $(\Uunins_t,M_t,m_t)$ and two ranks $i,j\in\Uunins_t$ with $\alpha_i\ge\alpha_j$. All expectations in this step condition on the current state.

\emph{Coupling the two orders.} The rank-adjusted signals have the same distribution at every rank, so their joint predictive distribution is exchangeable. We can therefore compare the two orders using the same two draws, denoted by $y_i$ and $y_j$, and change only which rank-specific shift is attached to each draw. Under the $i$-then-$j$ order, the two revealed relevances are $\alpha_i+y_i$ and $\alpha_j+y_j$. Under the $j$-then-$i$ order, they are $\alpha_j+y_i$ and $\alpha_i+y_j$. Both orders reveal the same two rank-adjusted signals, so they produce the same posterior mean after the second inspection:
\[
m_{t+2}(y_i,y_j)
=
v_{t+2}\left(\frac{m_t}{v_t}+\frac{y_i+y_j}{\sigma_\eta^2}\right),
\]
which is symmetric in $y_i$ and $y_j$.

\emph{The option to stop after the first inspection.} Holding the first draw $y_i$ fixed, the best relevance after inspecting rank $i$ first or rank $j$ first is, respectively,
\[
M_{t+1}^{i}(y_i):=\max\{M_t,\alpha_i+y_i\},
\qquad
M_{t+1}^{j}(y_i):=\max\{M_t,\alpha_j+y_i\}.
\]
The superscript records which rank receives the first draw. Because $\alpha_i\ge\alpha_j$, inspecting $i$ first produces a weakly better current best: $M_{t+1}^{i}(y_i)\ge M_{t+1}^{j}(y_i)$ for every $y_i$.

Let $C_{ij}(y_i)$ be the value of inspecting $j$ after first inspecting $i$ and observing $y_i$, and then proceeding optimally over the remaining items:
\[
C_{ij}(y_i)
:=
-c+\E\!\left[
V_{t+2}\!\left(
\Uunins_t\setminus\{i,j\},
\max\{M_{t+1}^{i}(y_i),\alpha_j+y_j\},
m_{t+2}(y_i,y_j)
\right)
\,\middle|\,y_i
\right].
\]
The corresponding value after inspecting $j$ first is
\[
C_{ji}(y_i)
:=
-c+\E\!\left[
V_{t+2}\!\left(
\Uunins_t\setminus\{i,j\},
\max\{M_{t+1}^{j}(y_i),\alpha_i+y_j\},
m_{t+2}(y_i,y_j)
\right)
\,\middle|\,y_i
\right].
\]

For every second draw $y_j$,
\[
\max\{M_{t+1}^{i}(y_i),\alpha_j+y_j\}
\le
\max\{M_{t+1}^{j}(y_i),\alpha_i+y_j\}+M_{t+1}^{i}(y_i)-M_{t+1}^{j}(y_i).
\]
To see the inequality, consider which term determines the maximum on the left. If it is $M_{t+1}^{i}(y_i)$, subtracting that term from the right-hand side leaves $\max\{M_{t+1}^{j}(y_i),\alpha_i+y_j\} - M_{t+1}^{j}(y_i)\ge 0$. If the left-hand side equals $\alpha_j+y_j$, the maximum on the right is already at least $\alpha_i+y_j\ge\alpha_j+y_j$, and the remaining difference is nonnegative. Thus, after both inspections, placing $i$ first can raise the current best by no more than it raised the stop payoff after the first inspection.

Step~1 shows that the optimal value over the remaining items can increase by no more than the current best. Applying that bound inside the conditional expectation gives
\[
C_{ij}(y_i)-C_{ji}(y_i)
\le
M_{t+1}^{i}(y_i)-M_{t+1}^{j}(y_i),
\]
or, equivalently,
\[
M_{t+1}^{i}(y_i)-C_{ij}(y_i)
\ge
M_{t+1}^{j}(y_i)-C_{ji}(y_i).
\]
After the same first signal, the advantage of stopping over making the second inspection is therefore weakly greater when the higher-ranked item comes first.

\emph{Comparison before the first inspection.} Before the first inspection, the value of the $i$-then-$j$ strategy is
\[
-c+\E\!\left[\max\{M_{t+1}^{i}(y_i),C_{ij}(y_i)\}\right].
\]
Using $\max\{a,b\}=b+(a-b)^+$ separates this value into two parts:
\[
-c
+\E[C_{ij}(y_i)]
+\E\!\left[\bigl(M_{t+1}^{i}(y_i)-C_{ij}(y_i)\bigr)^+\right].
\]
The first expectation is the value of committing to both inspections. Exchangeability gives
\[
\E[C_{ij}(y_i)]=\E[C_{ji}(y_i)],
\]
because exchanging $y_i$ and $y_j$ maps the committed $i$-then-$j$ strategy into the committed $j$-then-$i$ strategy without changing the joint predictive distribution. The second expectation is the value of retaining the option to stop after the first inspection. The preceding inequality remains valid after taking positive parts and expectations, so this stopping-option term is weakly greater when $i$ comes first. The commitment values are equal, while the stopping option favors the higher-ranked item first. Hence the $i$-then-$j$ strategy is at least as valuable as the $j$-then-$i$ strategy.

\paragraph{Step 3: Optimal rank-order traversal.}
We induct on the number of uninspected items. The result is immediate when $|\Uunins_t|\le1$. Suppose it holds whenever fewer than $|\Uunins_t|$ items remain, and let $i$ be the highest-ranked item in $\Uunins_t$. Consider any other possible first inspection $j\ne i$. After observing $j$, the user faces a problem with one fewer uninspected item. The induction hypothesis implies that an optimal continuation either stops immediately or inspects $i$ next, since $i$ is then the highest-ranked remaining item. After both inspections, she proceeds optimally over $\Uunins_t\setminus\{i,j\}$. Thus, the best policy that begins with $j$ has exactly the $j$-then-$i$ form compared in Step~2, including the option to stop after $j$.

If the user begins with $i$, the corresponding $i$-then-$j$ strategy is feasible, even if it is not her only option. Step~2 shows that this strategy is weakly better than the best policy beginning with $j$. Since $j$ was arbitrary, inspecting $i$ weakly dominates every other inspection action. Thus, at every state, an optimal action is either to stop or to inspect the highest-ranked remaining item. Applying this choice recursively constructs a rank-order policy $\pi'$ satisfying
\[
\E[U(\pi')]
=
V_0(\Uunins_0,x_b,m_0)
\ge
\E[U(\pi)]
\]
for every policy $\pi$. This proves the lemma. $\hfill\qed$

\subsection{Proof of Lemma~\ref{lem:posterior}}

\emph{Posterior of $\mu$.}
By Lemma~\ref{lem:predictive}, $y_i \mid \mu \sim \N(\mu, \sigma_\eta^2)$ iid, so the model follows the standard conjugate Gaussian--Gaussian setting. The conditional joint density of $(y_1, \ldots, y_t)$ given $\mu$ is
\[
p(y_1, \ldots, y_t \mid \mu) = \prod_{s=1}^{t} \frac{1}{\sqrt{2\pi}\sigma_\eta} \exp\left(-\frac{(y_s - \mu)^2}{2\sigma_\eta^2}\right) \propto \exp\left(-\frac{1}{2\sigma_\eta^2}\sum_{s=1}^{t}(y_s - \mu)^2\right).
\]
Because $\alpha_i$ is deterministic, conditioning on $H_t$ is equivalent to conditioning on $(y_1, \ldots, y_t)$.
Combining the likelihood with the Gaussian prior $\mu \sim \N(m_0, v_0)$ and absorbing $\mu$-free factors,
\[
p(\mu \mid H_t) \propto \exp\left(-\frac{(\mu - m_0)^2}{2 v_0}\right) \exp\left(-\frac{1}{2\sigma_\eta^2}\sum_{s=1}^{t}(y_s - \mu)^2\right).
\]
Denote $v_t = (1/v_0 + t/\sigma_\eta^2)^{-1}$ and $m_t = v_t(m_0/v_0 + \sum_{s=1}^{t}y_s/\sigma_\eta^2)$. Expanding the squares and collecting terms in $\mu$ inside the exponent,
\begin{equation*}
	-\frac{(\mu - m_0)^2}{2 v_0} - \frac{1}{2\sigma_\eta^2}\sum_{s=1}^{t}(y_s - \mu)^2 = -\frac{1}{2}\left(\frac{1}{v_0}+\frac{t}{\sigma_\eta^2}\right)\mu^2 + \left(\frac{m_0}{v_0}+\sum_{s=1}^{t}\frac{y_s}{\sigma_\eta^2}\right)\mu + C_1 = -\frac{1}{2v_t}(\mu - m_t)^2 + C_2
\end{equation*}
where $C_1$ and $C_2$ absorb $\mu$-free terms. So
\[
p(\mu \mid H_t) \propto \exp\left(-\frac{1}{2v_t}(\mu - m_t)^2 \right).
\]
This is a Gaussian density with mean $m_t$ and variance $v_t$. Hence $\mu \mid H_t \sim \N(m_t, v_t)$, as claimed.

\emph{Posterior predictive of $x_{t+1}$.}
By Lemma~\ref{lem:predictive}, decompose $x_{t+1} = \mu + \alpha_{t+1} + \eta_{t+1}$, where $\eta_{t+1} \mid \mu \sim \N(0, \sigma_\eta^2)$. Note that $\eta_{t+1}$ is independent of $(\mu, \eta_1, \ldots, \eta_t)$, and hence independent of $H_t$. It follows that $\eta_{t+1}$ is independent of $(\mu, H_t)$ jointly, so $\eta_{t+1} \mid H_t \sim \N(0, \sigma_\eta^2)$.

By the posterior just derived, $\mu \mid H_t \sim \N(m_t, v_t)$. The conditional law of $x_{t+1}$ given $H_t$ is therefore the convolution of two conditionally independent Gaussians (plus the deterministic shift $\alpha_{t+1}$):
\[
x_{t+1} \mid H_t = (\mu \mid H_t) + \alpha_{t+1} + (\eta_{t+1} \mid H_t) \sim \N\big(m_t + \alpha_{t+1}, v_t + \sigma_\eta^2\big),
\]
where independence makes the variances additive. $\hfill\qed$

\subsection{Proof of Proposition~\ref{prop:reservation}}

\paragraph{Roadmap.} Fix $t \in \{0, 1, \ldots, N-1\}$. Define the continuation value (inspect branch of \eqref{eq:bellman})
\begin{equation}\label{eq:Cdef}
C_t(M, m) := -c + \E_{x_{t+1}}\left[V_{t+1}\big(\max(M, x_{t+1}), m_{t+1}(x_{t+1})\big)\right],
\end{equation}
with $x_{t+1} \sim \N(m + \alpha_{t+1}, v_t + \sigma_\eta^2)$ and $m_{t+1}(\cdot)$ as in \eqref{eq:update}. The Bellman equation \eqref{eq:bellman} reads $V_t(M, m) = \max\{M, C_t(M, m)\}$,
so the optimal policy stops at epoch $t$ if and only if $C_t(M, m) \le M$. Define the stop-continue gap
\begin{equation}\label{eq:Deltadef}
\Delta_t(M, m) := C_t(M, m) - M.
\end{equation}
The proposition is then equivalent to: there is a constant $\kappa_t^\star$, independent of $m$ and $x_b$, such that the equation $\Delta_t(M, m) = 0$ has the unique solution $M = m + \alpha_{t+1} + \kappa_t^\star$ for every $m$, and the policy stops at $t$ if and only if $M \ge m + \alpha_{t+1} + \kappa_t^\star$.

The argument has two parts. Steps S1--S7 show by backward induction that for each $m$ the equation $\Delta_t(M, m) = 0$ has a unique solution $R_t(m)$ that characterizes the stop region. Steps S8--S9 use translation equivariance of the Bellman operator to reduce $R_t(m)$ to a $1$-Lipschitz affine function of $m$ with slope $1$, yielding the linear form \eqref{eq:R-affine-opt}.

The nine steps are:
\begin{enumerate}
\item[(S1)] Two-sided bound on $V_t$ (so the predictive integrals are finite).
\item[(S2)] $V_t$ and $C_t$ are continuous in $(M, m)$.
\item[(S3)] $V_t$ is $1$-Lipschitz and nondecreasing in $M$.
\item[(S4)] $\Delta_t(\cdot, m)$ is \emph{strictly} decreasing in $M$.
\item[(S5)] $V_t$, hence $\Delta_t$, is nondecreasing in $m$.
\item[(S6)] $\Delta_t(M, m) \to +\infty$ as $M \to -\infty$, and $\Delta_t(M, m) \to -c$ as $M \to +\infty$.
\item[(S7)] For each $m \in \R$, there is a unique $R_t(m) \in \R$ with $\Delta_t(R_t(m), m) = 0$, and the policy stops at $t$ if and only if $M \ge R_t(m)$.
\item[(S8)] Translation equivariance: $V_t(M+\delta, m+\delta) = V_t(M, m) + \delta$ for every $\delta \in \R$.
\item[(S9)] Affine form: $R_t(m) = m + \alpha_{t+1} + \kappa_t^\star$, with $\kappa_t^\star$ a deterministic constant independent of $m$ and of $x_b$.
\end{enumerate}
We proceed by backward induction on $t$, starting from the terminal condition $V_N(M, m) = M$.

\paragraph{Step 1: A two-sided bound on $V_t$.}
We claim
\begin{equation}\label{eq:Vbounds}
M \le V_t(M, m) \le M + \sum_{s = t}^{N-1} \E_{t,m}^{(s)}\left[(x_{s+1} - M)^+\right],
\end{equation}
where $\E_{t,m}^{(s)}$ denotes expectation under $x_{s+1} \sim \N(m + \alpha_{s+1}, v_t + \sigma_\eta^2)$. This is the predictive distribution of rank $s+1$ given the current time-$t$ belief $\mu \sim \N(m,v_t)$.

\emph{Lower bound.} The stop branch gives $V_t(M, m) \ge M$ directly.

\emph{Upper bound.} For any policy, dropping all remaining inspection costs and revealing every remaining item can only increase its payoff. For every realization,
\[
\max(M,x_{t+1},\ldots,x_N) \le M + \sum_{s=t}^{N-1}(x_{s+1}-M)^+.
\]
Taking expectations under the current belief and then optimizing over policies gives \eqref{eq:Vbounds}. At $t=N$, the sum is empty and $V_N(M,m)=M$. Since each $\E_{t,m}^{(s)}[(x_{s+1} - M)^+] \le \E_{t,m}^{(s)}[|x_{s+1}|] + |M| < \infty$, the bound is finite.

\paragraph{Step 2: Continuity of $V_t$ and $C_t$ in $(M, m)$.}
At $t = N$, $V_N(M, m) = M$ is continuous. Assume $V_{t+1}$ is continuous in $(M, m)$. For each fixed $x_{t+1}$, the integrand of \eqref{eq:Cdef},
\[
\Psi_{t+1}(M, m; x_{t+1}) := V_{t+1}\big(\max(M, x_{t+1}), m_{t+1}(x_{t+1})\big),
\]
is continuous in $(M, m)$ for each $x_{t+1}$. To see this, note that $\max(M, x_{t+1})$ is continuous in $M$, and the update map \eqref{eq:update},
\[
m_{t+1}(x_{t+1}; m) = \frac{v_{t+1}}{v_t}m + \frac{v_{t+1}}{\sigma_\eta^2}(x_{t+1} - \alpha_{t+1}),
\]
is affine, and hence continuous, in $m$. The inductive hypothesis ($V_{t+1}$ is continuous in $(M, m)$) gives continuity of $\Psi_{t+1}$ in $(M, m)$. To apply dominated convergence under a fixed probability law, write $x_{t+1}=m+\alpha_{t+1}+\sqrt{v_t+\sigma_\eta^2}\,Z$ with $Z\sim\N(0,1)$. For $M'=\max(M,x_{t+1})$ and $m'=m_{t+1}(x_{t+1};m)$, Step 1 gives
\[
|V_{t+1}(M',m')|
\le (N-t)|M'|+\sum_{s=t+1}^{N-1}\E_{t+1,m'}^{(s)}[|x_{s+1}|].
\]
On any compact set of $(M,m)$, both $|M'|$ and $|m'|$ are bounded by a constant times $1+|Z|$. Each Gaussian absolute moment in the sum is at most $|m'+\alpha_{s+1}|+\sqrt{v_{t+1}+\sigma_\eta^2}\,\E[|Z|]$. Thus, the integrand has an integrable envelope $K(1+|Z|)$ for a constant $K$ depending only on the compact set and the primitives. Dominated convergence gives continuity of $C_t$, and hence of $V_t=\max(M,C_t)$.

\paragraph{Step 3: $V_t$ is $1$-Lipschitz and nondecreasing in $M$.}
We claim by backward induction on $t$ that for all $M_1 \le M_2$ and $m \in \R$,
\begin{equation}\label{eq:Lipschitz}
0 \le V_t(M_2, m) - V_t(M_1, m) \le M_2 - M_1.
\end{equation}
At $t = N$, $V_N(M, m) = M$ satisfies \eqref{eq:Lipschitz} with equality on the right. Assume \eqref{eq:Lipschitz} at $t+1$.

\emph{The stop branch} $M$ trivially satisfies \eqref{eq:Lipschitz}.

\emph{The continuation branch.} Fix $x_{t+1}$. The map $M \mapsto \max(M, x_{t+1})$ is $1$-Lipschitz nondecreasing, so $\max(M_2, x_{t+1}) - \max(M_1, x_{t+1}) \in [0, M_2 - M_1]$ (note that $M_1 \le M_2$). Combining with the inductive hypothesis \eqref{eq:Lipschitz} at $t+1$,
\begin{equation*}
0 \le V_{t+1}\big(\max(M_2, x_{t+1}), m_{t+1}\big) - V_{t+1}\big(\max(M_1, x_{t+1}), m_{t+1}\big) 
\le \max(M_2, x_{t+1}) - \max(M_1, x_{t+1}) \le M_2 - M_1.
\end{equation*}
Taking expectation over $x_{t+1}$ (whose predictive law does not depend on $M$) preserves the bound on $C_t$. Finally, $V_t = \max(M, C_t)$ is the pointwise max of two $1$-Lipschitz nondecreasing functions, hence itself $1$-Lipschitz nondecreasing.

\paragraph{Step 4: $\Delta_t(\cdot, m)$ is strictly decreasing in $M$.}
We show that for $M_1 < M_2$,
\begin{equation}\label{eq:Cgap}
C_t(M_2, m) - C_t(M_1, m) \le (M_2 - M_1)\Pr(x_{t+1} \le M_2) < M_2 - M_1,
\end{equation}
the strict inequality on the right because $x_{t+1}$ has full support on $\R$, so $\Pr(x_{t+1} \le M_2) < 1$.

Fix any realization $x_{t+1} = x$, and let $\Psi(M) := V_{t+1}(\max(M, x), m_{t+1}(x; m))$. We compare $\Psi(M_2)$ to $\Psi(M_1)$ in three regimes:
\begin{itemize}
\item \emph{If $x \ge M_2$:} $\max(M_i, x) = x$ for $i = 1, 2$, so $\Psi(M_2) - \Psi(M_1) = 0$.
\item \emph{If $M_1 < x < M_2$:} $\max(M_2, x) = M_2$ and $\max(M_1, x) = x$, so by Step 3
\[
\Psi(M_2) - \Psi(M_1) = V_{t+1}(M_2, \cdot) - V_{t+1}(x, \cdot) \in [0, M_2 - x] \subseteq [0, M_2 - M_1].
\]
\item \emph{If $x \le M_1$:} $\max(M_i, x) = M_i$, so $\Psi(M_2) - \Psi(M_1) = V_{t+1}(M_2, \cdot) - V_{t+1}(M_1, \cdot) \in [0, M_2 - M_1]$.
\end{itemize}
In all three cases, $\Psi(M_2) - \Psi(M_1) \le (M_2 - M_1)\mathbf{1}\{x \le M_2\}$ pointwise. Taking expectations,
\[
C_t(M_2, m) - C_t(M_1, m) = \E_{x_{t+1}}[\Psi(M_2) - \Psi(M_1)] \le (M_2 - M_1)\Pr(x_{t+1} \le M_2),
\]
which is \eqref{eq:Cgap}. Subtracting $M_2 - M_1$ from both sides,
\begin{equation*}
\Delta_t(M_2, m) - \Delta_t(M_1, m) \le (M_2 - M_1)\big(\Pr(x_{t+1} \le M_2) - 1\big) < 0,
\end{equation*}
so $\Delta_t(\cdot, m)$ is strictly decreasing on $\R$.

\paragraph{Step 5: $V_t$ and $\Delta_t$ are nondecreasing in $m$.}
We show $V_{t}(M, m)$ is nondecreasing in $m$ by backward induction. At $t = N$, $V_N(M, m) = M$, which is nondecreasing in $m$. 
Suppose $V_{t+1}(M, m)$ is nondecreasing in $m$, write $x_{t+1} = m + \alpha_{t+1} + \xi$ where $\xi \sim \N(0, v_t + \sigma_\eta^2)$ has a law that does not depend on $m$. Substituting into \eqref{eq:update} and \eqref{eq:Cdef},
\[
m_{t+1}(x_{t+1}; m) = \frac{v_{t+1}}{v_t}m + \frac{v_{t+1}}{\sigma_\eta^2}(m + \xi) = v_{t+1}\left(\frac{1}{v_t} + \frac{1}{\sigma_\eta^2}\right) m + \frac{v_{t+1}}{\sigma_\eta^2}\xi = m + \frac{v_{t+1}}{\sigma_\eta^2}\xi,
\]
where the last equality uses the precision identity $1/v_{t+1} = 1/v_t + 1/\sigma_\eta^2$. Thus, the continuation value may be expressed in terms of the $m$-free noise $\xi$,
\begin{equation}\label{eq:Cdef-eps}
C_t(M, m) = -c + \E_\xi\left[V_{t+1}\Big(\max(M, m + \alpha_{t+1} + \xi), m + \tfrac{v_{t+1}}{\sigma_\eta^2}\xi\Big)\right].
\end{equation}
For each fixed $\xi$, both arguments of $V_{t+1}$ are nondecreasing in $m$. By Step 3 and the inductive hypothesis, $V_{t+1}$ is jointly nondecreasing in its two arguments, so the integrand of \eqref{eq:Cdef-eps} is nondecreasing in $m$. Taking expectation over $\xi$ (whose law is $m$-free) preserves this. Therefore, $C_t$ is nondecreasing in $m$, and so is $V_t = \max(M, C_t)$. The stop branch $M$ does not depend on $m$, so $\Delta_t(M, m) = C_t(M, m) - M$ is also nondecreasing in $m$.

\paragraph{Step 6: Boundary behavior of $\Delta_t$ as $M \to \pm\infty$.}

\emph{As $M \to -\infty$.} Then $\max(M, x_{t+1}) = x_{t+1}$ a.s., and by Step 1's lower bound applied at $t+1$,
\[
V_{t+1}\big(\max(M, x_{t+1}), m_{t+1}\big) \ge \max(M, x_{t+1}) = x_{t+1}.
\]
Hence
\[
C_t(M, m) \ge -c + \E_{x_{t+1}}[x_{t+1}] = -c + m + \alpha_{t+1},
\]
which is bounded below uniformly in $M$. Therefore $\Delta_t(M, m) = C_t(M, m) - M \to +\infty$ as $M \to -\infty$.

\emph{As $M \to +\infty$.} Applying Step 1 at $t+1$ with $M'=\max(M,x_{t+1})$ and using $(x_{s+1}-M')^+\le(x_{s+1}-M)^+$ gives
\[
\Delta_t(M,m)\le -c+\E_{x_{t+1}}[(x_{t+1}-M)^+]
+\sum_{s=t+1}^{N-1}\E_{x_{t+1}}\E_{t+1,m_{t+1}}^{(s)}[(x_{s+1}-M)^+].
\]
The inner expectations condition on the updated time-$(t+1)$ belief. Averaging over $x_{t+1}$ returns the time-$t$ predictive distribution by the law of iterated expectations. Also, $V_{t+1}(M',m_{t+1})\ge M'\ge M$ implies $\Delta_t(M,m)\ge-c$. Therefore,
\[
-c\le\Delta_t(M,m)\le-c+\sum_{s=t}^{N-1}\E_{t,m}^{(s)}[(x_{s+1}-M)^+].
\]
For fixed $m$, every term in the finite sum tends to zero: for $M\ge0$, $(x_{s+1}-M)^+\le x_{s+1}^+$, which is integrable under the time-$t$ Gaussian predictive distribution. Dominated convergence and the two-sided bound give $\Delta_t(M,m)\to-c<0$.

\paragraph{Step 7: $R_t(m)$ exists, is unique, and characterizes the stop region.}
Fix $m \in \R$. By Step 2, $\Delta_t(\cdot, m)$ is continuous; by Step 4, it is strictly decreasing; by Step 6, $\Delta_t(M, m) \to +\infty$ as $M \to -\infty$ and $\Delta_t(M, m) \to -c < 0$ as $M \to +\infty$. The intermediate value theorem, together with strict monotonicity, gives a unique $R_t(m) \in \R$ with
\begin{equation*}
\Delta_t(R_t(m), m) = 0, \qquad \text{i.e.,}\qquad C_t(R_t(m), m) = R_t(m).
\end{equation*}
Since $\Delta_t(\cdot, m)$ is strictly decreasing,
\[
M \ge R_t(m) \iff \Delta_t(M, m) \le 0 \iff C_t(M, m) \le M,
\]
i.e., stopping is (weakly) optimal. This proves the threshold characterization in the proposition.

\paragraph{Step 8: Translation equivariance of $V_t$.}
We show by the same backward induction that $V_t(M+\varepsilon, m+\varepsilon) = V_t(M, m) + \varepsilon$ for every $\varepsilon \in \R$. At $t = N$, $V_N(M+\varepsilon, m+\varepsilon) = M+\varepsilon = V_N(M, m) + \varepsilon$. Suppose the claim holds at $t+1$. Using the $m$-free noise $\xi = x_{t+1} - m - \alpha_{t+1} \sim \N(0, v_t + \sigma_\eta^2)$ and the simplified update $m_{t+1}(x_{t+1}; m) = m + (v_{t+1}/\sigma_\eta^2)\xi$ from Step 5,
\begin{align*}
C_t(M+\varepsilon, m+\varepsilon) &= -c + \E_\xi\Big[V_{t+1}\big(\max(M+\varepsilon, m+\varepsilon+\alpha_{t+1}+\xi), m+\varepsilon+(v_{t+1}/\sigma_\eta^2)\xi\big)\Big] \\
&= -c + \E_\xi\Big[V_{t+1}\big(\max(M, m+\alpha_{t+1}+\xi)+\varepsilon, m+(v_{t+1}/\sigma_\eta^2)\xi+\varepsilon\big)\Big] \\
&= -c + \E_\xi\Big[V_{t+1}\big(\max(M, m+\alpha_{t+1}+\xi), m+(v_{t+1}/\sigma_\eta^2)\xi\big)\Big] + \varepsilon \\
&= C_t(M, m) + \varepsilon,
\end{align*}
where the second line uses $\max(a+\varepsilon, b+\varepsilon) = \max(a,b)+\varepsilon$ and the third applies the inductive hypothesis. Therefore,
\[
V_t(M+\varepsilon, m+\varepsilon) = \max\{M+\varepsilon, C_t(M+\varepsilon, m+\varepsilon)\} = \max\{M, C_t(M, m)\} + \varepsilon = V_t(M, m) + \varepsilon.
\]

\paragraph{Step 9: Affine form and $x_b$-independence.}
Setting $\varepsilon = -m$ in Step 8 gives $V_t(M, m) = V_t(M-m, 0) + m$, and applying the same shift to $C_t$ gives $C_t(M, m) = C_t(M-m, 0) + m$. Define $\bar C_t(X) := C_t(X, 0)$. By Step 7, $R_t(m)$ is the unique solution of $C_t(R_t(m), m) = R_t(m)$; substituting the translated form,
\[
\bar C_t(R_t(m) - m) = R_t(m) - m,
\]
so $R_t(m) - m$ is a fixed point of $\bar C_t$. As a function of $X$, the gap $\bar C_t(X) - X$ is the restriction of $\Delta_t(\cdot, 0)$ to its first argument. By Step~4 it is strictly decreasing in $X$, and by Step~6 it tends to $+\infty$ as $X \to -\infty$ and to $-c$ as $X \to +\infty$. The fixed point is therefore unique and, in particular, independent of $m$. Call this constant $\alpha_{t+1} + \kappa_t^\star$. Then $R_t(m) = m + \alpha_{t+1} + \kappa_t^\star$ is \eqref{eq:R-affine-opt}, and $R_t$ is automatically continuous and strictly increasing in $m$ as an affine function with slope $1$.

Finally, the Bellman recursion at $t \in \{0, \ldots, N-1\}$ has terminal condition $V_N(M, m) = M$ and a transition that depends only on $(M, m, x_{t+1})$ and the deterministic primitives $(\alpha_{t+1}, v_t, \sigma_\eta^2, c)$; the outside option $x_b$ enters the model only through the initial condition $M_0 = x_b$ and not through the recursion. Hence $V_t$, $C_t$, $\bar C_t$, and the fixed point $\kappa_t^\star$ are all $x_b$-independent. $\hfill\qed$

\subsection{Proof of Proposition~\ref{prop:myopic}}

At epoch $t \in \{0, \ldots, N-1\}$, the one-step look-ahead rule treats the next epoch as terminal: inspect if and only if the expected payoff of inspecting rank $t+1$ and stopping weakly exceeds the immediate stop value,
\[
-c + \E\left[\max(M_t, x_{t+1}) \big| H_t\right] \ge M_t.
\]
Using $\max(M_t, x_{t+1}) = M_t + (x_{t+1} - M_t)^+$ and taking conditional expectation, this is equivalent to
\begin{equation*}
c \le \E\left[(x_{t+1} - M_t)^+ \big| H_t\right].
\end{equation*}
By Lemma~\ref{lem:posterior}, $x_{t+1} \mid H_t \sim \N(m_t + \alpha_{t+1}, \sigma_t^{\star 2})$. Applying Lemma~\ref{lem:gaussiantrunc} with $a = M_t$, $\mu = m_t + \alpha_{t+1}$, $\sigma = \sigma_t^\star$, and writing $d_t := (M_t - m_t - \alpha_{t+1})/\sigma_t^\star$,
\begin{equation*}
\E\left[(x_{t+1} - M_t)^+ \mid H_t\right] = \sigma_t^\star\big(\phi(d_t) - d_t\Phi(-d_t)\big) = \sigma_t^\star g(d_t),
\end{equation*}
using the definition of $g$ in \eqref{eq:gdef}.
Thus, the myopic rule continues at epoch $t$ if and only if $\sigma_t^\star g(d_t) \ge c$ and stops if and only if $\sigma_t^\star g(d_t) \le c$. The reservation $R_t^{\mathrm{myo}}(m_t)$ is the value of $M_t$ at which equality holds:
\[
g\left(\frac{R_t^{\mathrm{myo}}(m_t) - m_t - \alpha_{t+1}}{\sigma_t^\star}\right) = \frac{c}{\sigma_t^\star}.
\]
By Lemma~\ref{lem:gproof}, $g : \R \to (0, \infty)$ is a strictly decreasing bijection, so we may invert,
\[
R_t^{\mathrm{myo}}(m_t) = m_t + \alpha_{t+1} + \sigma_t^\star g^{-1}\left(c/\sigma_t^\star\right) = m_t + \alpha_{t+1} + \kappa_t,
\]
which is \eqref{eq:Rmyo-general}: stopping is myopically optimal if and only if $M_t - m_t \ge \alpha_{t+1} + \kappa_t$.

As $v_{t+1}^{-1} = v_t^{-1} + \sigma_\eta^{-2}$, $v_t$ and hence $\sigma_t^{\star 2} = v_t + \sigma_\eta^2$ are strictly decreasing in $t$. View $\kappa$ as a function of $s > 0$ via $\kappa(s) := s g^{-1}(c/s)$. The implicit-function characterization $s g(\kappa/s) = c$ becomes $F(\kappa, s) := s\phi(\kappa/s) - \kappa\Phi(-\kappa/s) - c = 0$. Using $\phi'(d) = -d\phi(d)$ and $\phi(-d) = \phi(d)$,
\[
\frac{\partial F}{\partial \kappa} = -\Phi(-\kappa/s) < 0, \qquad \frac{\partial F}{\partial s} = \phi(\kappa/s) > 0.
\]
By the implicit function theorem, $d\kappa/ds = \phi(\kappa/s)/\Phi(-\kappa/s) > 0$, so $\kappa(s)$ is strictly increasing. Composing with $s = \sigma_t^\star$, which is strictly decreasing in $t$, gives $\kappa_t = \kappa(\sigma_t^\star)$ strictly decreasing in $t$.

At $t = 0$, $\sigma_0^{\star 2} = v_0 + \sigma_\eta^2$ gives the stated value of $\kappa_0$. As $t \to \infty$, $v_t \to 0$, so $\sigma_t^\star \to \sigma_\eta$, and by continuity of $g^{-1}$, $\kappa_t \to \sigma_\eta g^{-1}(c/\sigma_\eta) =: \kappa_\infty$. This proves \eqref{eq:kappa-bounds-general}.$\hfill\qed$

\subsection{Proof of Lemma~\ref{lem:lead-markov}}

\emph{Markov property and one-step recursion \eqref{eq:L-recursion}.}
By definition, $x_t = m_{t-1} + \alpha_t + \xi_t$ with $\xi_t \mid H_{t-1} \sim \N(0, \sigma_{t-1}^{\star 2})$ independent of $H_{t-1}$ (Lemma~\ref{lem:posterior}). Substituting into the dynamics
\[
M_t = \max(M_{t-1}, x_t), \qquad m_t = m_{t-1} + \frac{v_t}{\sigma_\eta^2}\bigl(x_t - \alpha_t - m_{t-1}\bigr),
\]
we obtain $M_t - m_{t-1} = \max(L_{t-1}, \alpha_t + \xi_t)$ and $m_t - m_{t-1} = \omega_t\xi_t$, using $v_t/\sigma_\eta^2 = \omega_t$. Subtracting yields \eqref{eq:L-recursion}. Because $L_t$ is a measurable function of $L_{t-1}$ alone and an exogenous Gaussian shock $\xi_t$ independent of $H_{t-1}$, the chain $\{L_t\}$ is Markov with respect to its natural filtration.

\emph{Conditional density \eqref{eq:L-density}.}
Define $\Psi_t(l, \xi) := \max(l, \alpha_t + \xi) - \omega_t\xi$, so $L_t = \Psi_t(L_{t-1}, \xi_t)$. The map $\xi \mapsto \Psi_t(l, \xi)$ is V-shaped: on $\xi \le l - \alpha_t$ it equals $l - \omega_t\xi$; on $\xi > l - \alpha_t$ it equals $\alpha_t + (1-\omega_t)\xi$; and the two branches meet at the kink with common value $L_{\min}(l) = (1-\omega_t)l + \omega_t\alpha_t$. Hence $\Psi_t(l, \cdot)$ takes values in $[L_{\min}(l), \infty)$. For $y \ge L_{\min}(l)$, the sublevel set $\{\xi : \Psi_t(l, \xi) \le y\}$ is the interval $[\xi_-(y), \xi_+(y)]$ with
\[
\xi_-(y) = \frac{l - y}{\omega_t}, \qquad \xi_+(y) = \frac{y - \alpha_t}{1 - \omega_t}.
\]
For $y < L_{\min}(l)$ the sublevel set is empty. Hence
\begin{equation*}
\Pr\bigl(L_t \le y \bigm| L_{t-1} = l\bigr) = \Phi\left(\frac{y - \alpha_t}{(1-\omega_t)\sigma_{t-1}^\star}\right) - \Phi\left(\frac{l - y}{\omega_t\sigma_{t-1}^\star}\right) \quad \text{for } y \ge L_{\min}(l),
\end{equation*}
and zero below. Differentiating in $y$, with $\xi_+'(y) = 1/(1 - \omega_t)$ and $\xi_-'(y) = -1/\omega_t$, delivers \eqref{eq:L-density}.$\hfill\qed$

\subsection{Proof of Proposition~\ref{prop:cont-interval}}

By Proposition~\ref{prop:reservation}, continuation past epoch $t$ is the event $\{L_t < r_t\}$. The strategy is to rewrite this event in the surprise coordinate $\xi_t$, where the lead has the piecewise-affine form of \eqref{eq:L-twobranch}, and then convert back to the relevance coordinate $x_t$ via $\xi_t = x_t - (m_{t-1} + \alpha_t)$.

\emph{Step 1: the lead as a V-shape in $\xi_t$.} By \eqref{eq:L-twobranch}, the lead at epoch $t$ satisfies
\begin{equation*}
L_t = \begin{cases}
L_{t-1} - \omega_t\xi_t & \text{if } \xi_t \le L_{t-1} - \alpha_t \quad (\text{disappointment}),\\
\alpha_t + (1 - \omega_t)\xi_t & \text{if } \xi_t > L_{t-1} - \alpha_t \quad (\text{discovery}).
\end{cases}
\end{equation*}
The two branches meet at the kink $\xi_t = L_{t-1} - \alpha_t$ with $L_{\min} := (1 - \omega_t)L_{t-1} + \omega_t\alpha_t$. On the disappointment branch $L_t$ is strictly decreasing in $\xi_t$ with slope $-\omega_t$; on the discovery branch it is strictly increasing with slope $1 - \omega_t$. Hence $L_t \ge L_{\min}$ for every $\xi_t$, with equality only at the kink.

\emph{Step 2: the trust regime, $r_t \le L_{\min}$.} Then $L_t \ge L_{\min} \ge r_t$ for every $\xi_t$, so $\{L_t < r_t\} = \emptyset$. The user stops at rank $t$ for every realization of $x_t$, which is part (a).

\emph{Step 3: the explore regime, $r_t > L_{\min}$.} On the disappointment branch, continuation requires $L_{t-1} - \omega_t \xi_t < r_t$, which is $\xi_t > (L_{t-1} - r_t)/\omega_t$. Note that we require $\xi_t \le L_{t-1} - \alpha_t$ to be on the disappointment branch. Therefore, $(L_{t-1} - r_t)/\omega_t < \xi_t \le L_{t-1} - \alpha_t$. Similarly, on the discovery branch, continuation requires $L_{t-1} - \alpha_t < \xi_t < (r_t - \alpha_t)/(1 - \omega_t)$.

Hence, the two restrictions glue into a single open interval, and
\begin{equation}\label{eq:continuation-interval}
	\{L_t < r_t\} \Longleftrightarrow \frac{L_{t-1} - r_t}{\omega_t} < \xi_t < \frac{r_t - \alpha_t}{1 - \omega_t}.
\end{equation}
Substituting $\xi_t = x_t - (m_{t-1} + \alpha_t)$ into \eqref{eq:continuation-interval} yields
\[
m_{t-1} + \alpha_t + \frac{L_{t-1} - r_t}{\omega_t} < x_t < m_{t-1} + \alpha_t + \frac{r_t - \alpha_t}{1 - \omega_t}.
\]
The endpoints are $\underline{x}_t(H_{t-1})$ and $\overline{x}_t(H_{t-1})$ as defined in the proposition, so \eqref{eq:cont-interval} holds.$\hfill\qed$

\subsection{Proof of Corollary~\ref{cor:tau-one}}

If $h(x_b) \ge 0$, then $h \ge 0$ everywhere, the stopping set is $\R$, and $\Pr(T = 1) = 1$. Substituting the explicit form of $h$ into $h(x_b) \ge 0$ yields $\alpha_1 - \alpha_2 \ge \kappa_1^\star + (1-\omega_1)(m_0 + \alpha_1 - x_b)$.

For the case of $h(x_b) < 0$, solving $h(\cdot) = 0$ on each branch directly yields the closed forms in \eqref{eq:s1-explicit}: on the lower branch $h$ has slope $-\omega_1$, so $\omega_1 \underline{x}_1 = x_b + \omega_1\alpha_1 - (1-\omega_1)m_0 - \alpha_2 - \kappa_1^\star$, and on the upper branch $h$ has slope $1-\omega_1$, so $(1-\omega_1) \overline{x}_1 = (1-\omega_1)m_0 + \alpha_2 + \kappa_1^\star - \omega_1\alpha_1$. A direct calculation shows that the inequalities $\underline{x}_1 < x_b$ and $\overline{x}_1 > x_b$ are each equivalent to $h(x_b) < 0$, so both hold in this regime. The stopping set is $(-\infty, \underline{x}_1] \cup [\overline{x}_1, \infty)$; integrating $x_1 \sim \N(m_0 + \alpha_1, v_0 + \sigma_\eta^2)$ over this union yields \eqref{eq:Ptau1-twopiece}.$\hfill\qed$

\subsection{Proof of Corollary~\ref{cor:reliable-ranker-starvation}}

Along the stated path,
\[
        \sigma_\eta^2(\rho)=\sigma_x^2(1-\rho)\to 0,
        \qquad
        \alpha_i(\rho)=\sigma_x\sqrt{\rho}q_i\to \alpha_i^\infty:=\sigma_xq_i,
\]
and $\alpha_1^\infty>\alpha_2^\infty>\cdots>\alpha_N^\infty$ inherits the strict order of the quantile sequence. The Bayesian weight on the first item is
\[
        \omega_1(\rho)
        =
        \frac{v_0}{v_0+\sigma_\eta^2(\rho)}
        \to 1,
\]
and the posterior variance after one inspection is
\[
        v_1(\rho)=\frac{v_0\sigma_\eta^2(\rho)}{v_0+\sigma_\eta^2(\rho)}\to 0.
\]

Recall from the proof of Proposition~\ref{prop:reservation} the inspect-branch value $C_1(M_1,m_1;\rho)$ from \eqref{eq:Cdef} and the stop-continue gap $\Delta_1(M_1,m_1;\rho)=C_1(M_1,m_1;\rho)-M_1$ from \eqref{eq:Deltadef}. By Step~S8 of that proof, $C_1$ is translation-equivariant: $C_1(M_1+a,m_1+a;\rho)=C_1(M_1,m_1;\rho)+a$ for every $a\in\R$. Setting $a=-m_1$ gives $C_1(M_1,m_1;\rho)=m_1+C_1(l,0;\rho)$ with $l:=M_1-m_1$, so
\[
        \Delta_1(M_1,m_1;\rho)=C_1(M_1,m_1;\rho)-M_1=C_1(l,0;\rho)-l
\]
depends on $(M_1,m_1)$ only through $l$. With a slight abuse, write $\Delta_1(l;\rho)$. By Step~S4, $\Delta_1(\cdot;\rho)$ is strictly decreasing in $l$, and $r_1(\rho):=\alpha_2(\rho)+\kappa_1^\star(\rho)$ is its unique zero.

\emph{Step 1: Upper bound on $r_1(\rho)$ via free information.}
At epoch~$1$ with state $(M_1,m_1)$, any continuation policy pays at least one inspection cost ($c$) and produces a terminal payoff bounded above by $\max\{M_1,x_2,\ldots,x_N\}$, so
\[
        C_1(M_1,m_1;\rho)
        \le
        -c+\E\bigl[\max\{M_1,x_2,\ldots,x_N\}\bigm| M_1,m_1\bigr].
\]
Subtracting $M_1$ from both sides gives an upper bound on the stop-continue gap:
\begin{equation}\label{eq:Delta1-ub-raw}
        \Delta_1(M_1,m_1;\rho)
        \le
        -c+\E\bigl[\max\{0,x_2-M_1,\ldots,x_N-M_1\}\bigm| M_1,m_1\bigr].
\end{equation}

We now rewrite the right-hand side as a function of $l$ alone. Set $\delta:=\mu-m_1$, the residual uncertainty about $\mu$ at epoch~$1$. By the Bayesian update, the posterior on $\mu$ given the user's epoch-1 information is $\N(m_1,v_1(\rho))$, so $\delta$ has conditional distribution $\N(0,v_1(\rho))$ given that information. The future noises $\eta_i\sim\N(0,\sigma_\eta^2(\rho))$ for $i\ge 2$ are independent of $\mu$ and of one another, hence also independent of $\delta$ given the epoch-1 information. From $x_i=\mu+\alpha_i(\rho)+\eta_i$ and $M_1-m_1=l$, $x_i-M_1 = \delta+\alpha_i(\rho)+\eta_i-l$. Substituting into \eqref{eq:Delta1-ub-raw},
\[
        \Delta_1(l;\rho)
        \le
        -c+\E\bigl[\max\bigl\{0,\delta+\alpha_2(\rho)+\eta_2-l,\ldots,\delta+\alpha_N(\rho)+\eta_N-l\bigr\}\bigr]
        =:\bar\Delta_1(l;\rho),
\]
where the right-hand side now depends only on $l$ and $\rho$.

We take the limit as $\rho\uparrow 1$ by dominated convergence. The expectation $\bar\Delta_1(l;\rho)$ depends on $(\delta,\eta_2,\ldots,\eta_N)$ only through their joint distribution, so we may replace them by any equally distributed realization. A convenient one is the \emph{coupling} that fixes a single set of $\N(0,1)$ random variables $Z_0,Z_2,\ldots,Z_N$ (iid, independent of $\rho$) once and for all, and sets
\[
        \delta:=\sqrt{v_1(\rho)}Z_0,
        \qquad
        \eta_i:=\sigma_\eta(\rho)Z_i\quad(i\ge 2).
\]
Each $\delta$ has the required $\N(0,v_1(\rho))$ marginal and each $\eta_i$ has the required $\N(0,\sigma_\eta^2(\rho))$ marginal, and they remain mutually independent; the only thing the coupling adds is that as $\rho$ varies, the same underlying $Z$'s are reused, so $\delta$ and $\eta_i$ become deterministic functions of $\rho$ on a fixed sample space. As $\rho\uparrow 1$, $v_1(\rho)\to 0$ and $\sigma_\eta(\rho)\to 0$, so $\delta\to 0$ and $\eta_i\to 0$ almost surely, while $\alpha_i(\rho)\to\alpha_i^\infty$. The integrand therefore converges almost surely:
\[
        \max\bigl\{0,\delta+\alpha_2(\rho)+\eta_2-l,\ldots,\delta+\alpha_N(\rho)+\eta_N-l\bigr\}
        \to
        \max\{0,\alpha_2^\infty-l\},
\]
using $\alpha_2^\infty>\alpha_3^\infty>\cdots>\alpha_N^\infty$ to drop the lower-rank terms. The integrand is bounded in absolute value by $|l|+\sqrt{v_0}|Z_0|+\sigma_x\max_i|q_i|+\sigma_x\max_i|Z_i|$ for all $\rho$ sufficiently close to $1$ (using $v_1(\rho)\le v_0$ and $\sigma_\eta(\rho)\le\sigma_x$), which is integrable. Dominated convergence gives $\bar\Delta_1(l;\rho) \to -c+\max\{0,\alpha_2^\infty-l\}$ for every $l\in\R$.

Fix any $\varepsilon\in(0,c)$ and set $l_\varepsilon:=\alpha_2^\infty-c+\varepsilon\in(\alpha_2^\infty-c,\alpha_2^\infty)$. The limit at $l_\varepsilon$ equals $-c+(\alpha_2^\infty-l_\varepsilon)=-\varepsilon<0$, so for all $\rho$ sufficiently close to $1$, $\Delta_1(l_\varepsilon;\rho)\le\bar\Delta_1(l_\varepsilon;\rho)<0$. Strict monotonicity of $\Delta_1(\cdot;\rho)$ then forces $r_1(\rho)\le l_\varepsilon$. As $\varepsilon\in(0,c)$ is arbitrary,
\begin{equation}\label{eq:limsup-r1}
        \limsup_{\rho\uparrow 1} r_1(\rho)\le\alpha_2^\infty-c.
\end{equation}

\emph{Step 2: The trust slack at $x_b$ is positive in the limit.}
Define
\[
        \Gamma(\rho)
        :=
        \omega_1(\rho)\alpha_1(\rho)
        +(1-\omega_1(\rho))(x_b-m_0)
        -
        r_1(\rho).
\]
This equals $h_\rho(x_b)$, where $h_\rho(x_1):=M_1(x_1)-m_1(x_1)-r_1(\rho)$ is the V-shaped first-stop function from Corollary~\ref{cor:tau-one}. Because $\omega_1(\rho)\to 1$ and $\alpha_1(\rho)\to\alpha_1^\infty$, combining with \eqref{eq:limsup-r1},
\[
        \liminf_{\rho\uparrow 1}\Gamma(\rho)
        \ge
        \alpha_1^\infty-\limsup_{\rho\uparrow 1} r_1(\rho)
        \ge
        \alpha_1^\infty-(\alpha_2^\infty-c)
        =\sigma_x(q_1-q_2)+c
        >0.
\]
Hence there exists $\bar\rho<1$ such that $\Gamma(\rho)>0$ for all $\rho\in(\bar\rho,1)$.

\emph{Step 3: Trust regime and conclusion.}
For any fixed $\rho$, the first-stop function is
\[
        h_\rho(x_1)
        =
        \begin{cases}
        -\omega_1(\rho)x_1+x_b+\omega_1(\rho)\alpha_1(\rho)
        -(1-\omega_1(\rho))m_0-r_1(\rho),
        & x_1\le x_b,\\[3pt]
        (1-\omega_1(\rho))x_1+\omega_1(\rho)\alpha_1(\rho)
        -(1-\omega_1(\rho))m_0-r_1(\rho),
        & x_1>x_b.
        \end{cases}
\]
The left branch is decreasing in $x_1$ and the right branch is increasing in $x_1$, so $\min_{x_1}h_\rho=h_\rho(x_b)=\Gamma(\rho)$. Therefore $\Gamma(\rho)>0$ implies $h_\rho(x_1)>0$ for every $x_1\in\R$, which by Corollary~\ref{cor:tau-one} is the trust regime: the user stops immediately after the first inspection. Hence $\Pr_\rho(T_\rho=1)=1$ for all $\rho\in(\bar\rho,1)$, and the survival probability to any lower rank is zero.$\hfill\qed$

\subsection{Proof of Proposition~\ref{prop:tau-recursion}}

\emph{Stopping mass \eqref{eq:tau-mass}.}
The optimal policy stops at epoch $t$ if and only if the user has survived to $t$ \emph{and} the threshold criterion fires there. Formally, $\{T = t\} = \{T \ge t\} \cap \{L_t \ge \alpha_{t+1} + \kappa_t^\star\}$. By the definition of $\nu_t$ as the law of $L_t$ restricted to $\{T \ge t\}$,
\[
\Pr(T = t)
= \Pr\bigl\{L_t \ge \alpha_{t+1} + \kappa_t^\star, T \ge t\bigr\}
= \int_{[\alpha_{t+1} + \kappa_t^\star, \infty)} \nu_t(dl),
\]
which is \eqref{eq:tau-mass}.

\emph{Push-forward \eqref{eq:tau-pushforward}.}
Survival to epoch $t+1$ means the user survived to $t$ and did \emph{not} stop there: $\{T \ge t+1\} = \{T \ge t\} \cap \{L_t < \alpha_{t+1} + \kappa_t^\star\}$. For any Borel set $A \subseteq \R$, unpacking the definition of $\nu_{t+1}$ gives
\begin{equation}\label{eq:pf-expand}
\nu_{t+1}(A)
= \Pr\bigl\{L_{t+1} \in A, T \ge t+1\bigr\}
= \Pr\bigl\{L_{t+1} \in A, T \ge t, L_t < \alpha_{t+1} + \kappa_t^\star\bigr\}. 
\end{equation}
The indicators $\indic\{T \ge t\}$ and $\indic\{L_t < \alpha_{t+1} + \kappa_t^\star\}$ are both $H_t$-measurable (the first because $T$ is a stopping time, the second because $L_t$ and $\kappa_t^\star$ are determined by $H_t$). By the Markov property of $\{L_t\}$ established in Lemma~\ref{lem:lead-markov}, $\Pr\{L_{t+1} \in A \mid H_t\} = K_t^L(L_t, A)$. Applying the law of iterated expectations $\E[\cdot] = \E[\E[\cdot \mid H_t]]$ to \eqref{eq:pf-expand}:
\begin{align*}
\nu_{t+1}(A)
&= \E\Bigl[\indic\{T \ge t\}\indic\{L_t < \alpha_{t+1} + \kappa_t^\star\}\E\bigl[\indic\{L_{t+1} \in A\} \bigm| H_t\bigr]\Bigr] \\
&= \E\Bigl[\indic\{T \ge t\}\indic\{L_t < \alpha_{t+1} + \kappa_t^\star\}K_t^L(L_t, A)\Bigr] \\
&= \int_{(-\infty, \alpha_{t+1} + \kappa_t^\star)} K_t^L(l, A)\nu_t(dl),
\end{align*}
where the last line rewrites the expectation as an integral against $\nu_t$ (the law of $L_t$ weighted by $\indic\{T \ge t\}$). This is \eqref{eq:tau-pushforward}.

\emph{Termination and total mass.}
The convention $\alpha_{N+1} + \kappa_N^\star := -\infty$ makes the stopping region at the horizon equal to all of $\R$, so $\Pr(T = N) = \nu_N(\R)$. It remains to verify that the masses sum to one. Observe that $\nu_t(\R) = \Pr(T \ge t)$. For each $t \le N - 1$,
\[
\Pr(T \ge t)
= \Pr(T = t)
+ \Pr(T \ge t+1)
= \nu_t\bigl([\alpha_{t+1} + \kappa_t^\star, \infty)\bigr) + \nu_{t+1}(\R),
\]
so summing $\Pr(T = t)$ for $t = 0, \ldots, N - 1$ telescopes:
\[
\sum_{t=0}^{N-1} \Pr(T = t)
= \nu_0(\R) - \nu_N(\R)
= \Pr(T \ge 0) - \Pr(T = N)
= 1 - \Pr(T = N).
\]
Adding $\Pr(T = N)$ to both sides gives $\sum_{t=0}^{N}\Pr(T = t) = 1$. $\hfill\qed$

\end{document}